\documentclass[10pt, conference]{IEEEtran}

\IEEEoverridecommandlockouts
\usepackage{xspace,exscale}
\usepackage{fancybox}
\usepackage{graphicx}
\usepackage{color}
\usepackage{amsfonts}
\usepackage{amsthm}
\usepackage{amssymb}
\usepackage{url}

\usepackage[noadjust]{cite}

\usepackage{amsmath}
	\makeatletter
	\let\over=\@@over \let\overwithdelims=\@@overwithdelims
	\let\atop=\@@atop \let\atopwithdelims=\@@atopwithdelims
  	\let\above=\@@above \let\abovewithdelims=\@@abovewithdelims
  	\makeatother
\usepackage{fixltx2e}
\usepackage{rotating}

\usepackage{ifpdf}

\usepackage{subfigure}
\usepackage{psfrag}

\usepackage[%dvips,
            CJKbookmarks=true,
            bookmarksnumbered=true,
            bookmarksopen=true,
            colorlinks=true,
            citecolor=red,
            linkcolor=blue,
            anchorcolor=red,
            urlcolor=blue,
            pdfauthor={Tang,Polyanskiy,Wang,Wornell}
            ]{hyperref}

\usepackage{enumerate}
\usepackage{mathtools}

\usepackage{enumitem}
\usepackage{thmtools}
\usepackage{thm-restate}
\usepackage{cleveref}
\usepackage{bbm}
\DeclarePairedDelimiter{\floor}{\lfloor}{\rfloor}

\Crefname{figure}{Fig.}{Fig.}
\crefname{figure}{fig.}{fig.}

\newcommand{\matx}{\ensuremath{\mathcal{X}}}
\newcommand{\matr}{\ensuremath{\mathcal{R}}}

\newcommand{\matf}{\ensuremath{\mathcal{F}}}
\newcommand{\matg}{\ensuremath{\mathcal{G}}}

\newcommand{\mreals}{\ensuremath{\mathbb{R}}}

\newcommand{\ZZ}{\ensuremath{\mathbb{Z}}}

\newcommand{\TV}{{\sf TV}}

\ifx\eqref\undefined
	\newcommand{\eqref}[1]{~(\ref{#1})}
\fi
\ifx\mod\undefined
	\def\mod{\mathop{\rm mod}}
\fi

\newcommand{\vect}[1]{{\underline{#1}}}

\def\argmax{\mathop{\rm argmax}}

\def\EE{\mathbb{E}\,}

\def\PP{\mathbb{P}}

\def\eqdef{\stackrel{\triangle}{=}}

\def\unifto{\mathop{{\mskip 3mu plus 2mu minus 1mu%
	\setbox0=\hbox{$\mathchar"3221$}%
	\raise.6ex\copy0\kern-\wd0%
	\lower0.5ex\hbox{$\mathchar"3221$}}\mskip 3mu plus 2mu minus 1mu}}

\ifx\lesssim\undefined
\def\simleq{{{\mskip 3mu plus 2mu minus 1mu%
	\setbox0=\hbox{$\mathchar"013C$}%
	\raise.2ex\copy0\kern-\wd0%
	\lower0.9ex\hbox{$\mathchar"0218$}}\mskip 3mu plus 2mu minus 1mu}}
\else
\def\simleq{\lesssim}
\fi

\ifx\gtrsim\undefined
\def\simgeq{{{\mskip 3mu plus 2mu minus 1mu%
	\setbox0=\hbox{$\mathchar"013E$}%
	\raise.2ex\copy0\kern-\wd0%
	\lower0.9ex\hbox{$\mathchar"0218$}}\mskip 3mu plus 2mu minus 1mu}}
\else
\def\simgeq{\gtrsim}
\fi

\newif\ifmapx
{\catcode`/=0 \catcode`\\=12/gdef/mkillslash\#1{#1}}
\edef\jobnametmp{\expandafter\string\csname itpaper_apx\endcsname}
\edef\jobnameapx{\expandafter\mkillslash\jobnametmp}
\edef\jobnameexpand{\jobname}
\ifx\jobnameexpand\jobnameapx
\mapxtrue
\else
\mapxfalse
\fi

\newtheorem{theorem}{Theorem}
\newtheorem{lemma}[theorem]{Lemma}
\newtheorem{corollary}[theorem]{Corollary}
\newtheorem{definition}{Definition}
\newtheorem{proposition}[theorem]{Proposition}
\newtheorem{remark}{Remark}

\newcommand\numberthis{\addtocounter{equation}{1}\tag{\theequation}}
\begin{document}

\title{Defect tolerance: fundamental limits and examples}

\author{Jennifer Tang, Da Wang, Yury Polyanskiy, Gregory Wornell
\thanks{JT, YP and GW are with the Department of Electrical Engineering 
and Computer Science, MIT, Cambridge, MA 02139 USA.
	\mbox{e-mail:~{\ttfamily\{jstang,yp,gww\}@mit.edu}.} \textcolor{black}{DW was with the Department of Electrical Engineering and Computer Science, MIT, Cambridge, MA, 02139 and is now with Projection Analytics LLC, Hoboken, NJ, 07030. e-mail:~{\ttfamily dawang@alum.mit.edu}}.}%
\thanks{
 The research was supported by the Center for Science of Information (CSoI),
and NSF Science and Technology Center, under grant agreement CCF-09-39370. This work was also supported in part by Systems on Nanoscale Information fabriCs (SONIC), an SRC STARnet Center sponsored by MARCO and DARPA.}
\thanks{This research was presented at ISIT 2016.}
\thanks{\textcolor{black}{Copyright (c) 2017 IEEE. Personal use of this material is permitted.  However, permission to use this material for any other purposes must be obtained from the IEEE by sending a request to pubs-permissions@ieee.org.}}}
% avoiding spaces at the end of the author lines is not a problem with
% conference papers because we don't use \thanks or \IEEEmembership

\maketitle

\begin{abstract}
%This paper addresses the problem of adding redundancy to a collection of physical objects so that the overall system is
%more robust to failures. Physical redundancy can (generally) be achieved by employing copy/substitute procedures.
%This is fundamentally different from information redundancy, where a single parity check simultaneously protects a large
%number of data bits against a single erasure. We propose a bipartite graph model of designing defect-tolerant systems where defective objects are repaired by reconnecting them to strategically placed redundant objects. The
%fundamental limits of this model are characterized under various asymptotic settings and both asymptotic and finite-size optimal systems
%are constructed. 

%As argued in the paper, the following is
%a mathematical model for constructing defect-tolerant systems. We say that a $k$ by $m$ bipartite graph corrects $t$ defects over an alphabet of size $q$ if for every
%$q$-coloring of $k$ left vertices there exists a $q$-coloring of $m$ right vertices such that every left vertex is connected to
%at least $t$ same-colored right vertices. We study the trade-off between redundancy \textcolor{black}{ $m / kt$} and the total number of
%edges in the graph divided by \textcolor{black}{$kt$}. The question is trivial when $q\ge k$: the optimal solution is a simple $t$-fold replication.
%However, when $q<k$ non-trivial savings are possible by leveraging the inherent repetition of colors.

This paper addresses the problem of adding redundancy to a collection of physical objects so that the overall system is
more robust to failures. In
contrast to its information counterpart, which can exploit parity to
protect multiple information symbols from a single erasure, physical
redundancy can only be realized through duplication and substitution
of objects.  We propose a bipartite graph model for designing
defect-tolerant systems in which defective objects are replaced by
judiciously connected redundant objects.  The fundamental
limits of this model are characterized under various asymptotic
settings and both asymptotic and finite-size systems that approach
these limits are constructed.  Among other results, we show that
simple modular redundancy is in general suboptimal.  As we develop,
this combinatorial problem of defect tolerant system design has a
natural interpretation as one of graph coloring, and the analysis is
significantly different from that traditionally used in information
redundancy for error-control codes.
\end{abstract}

\begin{IEEEkeywords}
Defect-tolerant circuits, bipartite graphs, coloring, combinatorics, worst-case errors
\end{IEEEkeywords}

\tableofcontents

% \fi
\IEEEpeerreviewmaketitle

%%
%% MIGHTY HACK TO FIT TO 5 PAGES
%%
%\abovedisplayskip=6pt plus 1pt minus 2pt
%\belowdisplayskip=6pt plus 1pt minus 2pt
%\baselineskip=11.61pt

%------------------INTRODUCTION--------------------%
%------------------INTRODUCTION--------------------%
%------------------INTRODUCTION--------------------%
%------------------INTRODUCTION--------------------%

\section{Introduction}

Classical Shannon theory established principles of adding redundancy to data for combating noise and, dually, of
removing redundancy from data for more efficient storage. The central object of the classical theory is information,
which unlike physical objects, can be freely copied and combined. In fact, the marvel of
error-correcting codes is principally based on the counter-intuitive property that multiple unrelated information
bits $X_1,\ldots, X_k$ can be simultaneously protected by adding ``parity-checks'' such as 
\begin{equation}\label{eq:pc}
	Y = X_1 + \cdots + X_k \mod 2\,.
\end{equation}
In this example, the added parity-check $Y$ allows the recovery of the original message even if the vector 
$$ (X_1, X_2, \ldots, X_k, Y) $$
undergoes an erasure of an arbitrary element.

Physical objects (e.g., transistors in a chip) may also be subject to erasures (failures) and thus it is natural to
ask about ways of insuring the system against failure events. Note, however, that for physical objects operations such as~\eqref{eq:pc} are meaningless.  If the failure renders an object completely useless, then protecting against these failures would entail adding spare (redundant) elements. 
The required operation is to copy and then substitute.\footnote{\textcolor{black}{We assume for physical objects, the ``error-correction'' should provide an
exact copy of the object, not merely something functionally equivalent to the object. An example of what
is not considered as ``correction'' would be replacing a cell storing two bits $(b_1,b_2)$ with a cell storing $(b_1, b_1\oplus b_2)$.}} 
It may, therefore, seem that nothing better than simple replication can guard against failures. This paper shows
otherwise. Indeed, there exist non-trivial ways to add redundancy as long as the objects' diversity does not
exceed their number. That is, if the number of types of objects is smaller than the total number of them.

The objective of this paper is to develop a study of adding redundancy to a physical system where certain objects in the system fail and can only be replaced by substitutes. This paper will explore what are good design choices in this scenario and find fundamental limits for specific settings.

\subsection{Reconfigurable defect-tolerant circuits}\label{sec:reconf}

To facilitate defining the problem we intend to study, we will first present the application which informed the main model we developed for studying redundancy of physical objects, and that is the application of reconfigurable circuits. 

Consider a chip design process, in which the chip is composed of
many similar cells (e.g., standard-cell designs of ASICs). Layout of elements in each cell is dictated by the chip manufacturer. 
Each cell has $k$ input/output buses and $k$ placeholders (nodes) that can be
filled in with logic realizing one of $q$ functions. Now because of
manufacturing defects, not all $k$ elements in the cell will operate correctly (call these \emph{primary} elements). For this reason, each
cell also contains provisions for redundant elements. In particular, there are $m$ placeholders designated as redundant elements. The designer then selects what type of logic to
instantiate into these redundant elements. Once the chip is manufactured and placed on the testbed,
the testing equipment probes each cell and determines which primary elements are defective.
Programmable switches are then used to reconnect input/output buses from the defective primary
elements to one of the redundant elements containing the same logic.
So the summary of the events happening to each cell during this process is:
\begin{enumerate} 
\item Choose the layout of the placeholders and interconnect (these are provisional wires) \label{case::reconfig1}
\item Choose components (from available collection of possible types) to fill in the primary elements for the reconfigurable circuit \label{case::reconfig2}
\item Based on primary elements chosen, choose redundant components (from the same collection) to place in redundant placeholders
\item Build the circuit with these components
\item Based on where the defects occurred, reconfigure the interconnect (i.e., enable provisional wires with programmable switches) of the circuit to correct the defects. \label{case::reconfig5}
\end{enumerate}

\textcolor{black}{In the above summary, step \ref{case::reconfig5} of reconfiguring the defects is a simple operation which requires minimal programming (or switching) of the provisional wiring. This is precisely the advantage of choosing a good layout for the placeholders and provisional wires in step \ref{case::reconfig1}. Notice that this layout is universal in the sense that any choice of components in step \ref{case::reconfig2} (which may be arbitrarily dictated by the manufacturer later) should still lead to guarantees on the number of correctable defects. The focus of this work is to study optimal choices of layouts in step \ref{case::reconfig1} so that the rest of the steps in the procedure are possible. } 

With respect to this application, our goal is to understand what wiring topologies for the layout the chip manufacturer should try to implement in order to attain the optimal trade-off between the number of redundant elements, provisional wires and defect-tolerance. Notice that the two metrics, redundancy and wiring, both correspond to necessary additional resources. Adding redundancy requires more silicon area and the provisional wires requires additional metal and programmable switches.\footnote{There are certainly other metrics (such as geometric constraints \textcolor{black}{or resources to adjust the wiring between primary elements}) which are relevant for circuit applications, but we leave consideration of them to future work.} 

\textcolor{black}{Certainly, there are other procedures and layout constraints we could have chosen to study defects in hardware. For instance, there could be a $2$-hop system between the primary elements and redundant elements, decreasing the amount of wiring needed. However, multi-hop interconnects could
introduce more latency and make signal propagation delays unpredictable, which is why we do not discuss this in this work, but this is a scenario left for future work (see \Cref{sec::future}).}

\subsection{Relation to prior work}

Prior work on the subject of designing digital electronics robust to noise has been traditionally approached with the goal of
combating dynamic noise. This is epitomized in the line of work started by von Neumann~\cite{von1956probabilistic} and
contemporary variations~\cite{moore56}. Although significant progress has been made in understanding fundamental
limits in von Neumann's model,
see e.g.,~\cite{dobrushin1977lower,dobrushin1977upper,pippenger85,pippenger1988reliable,pippenger1991lower,
hajek1991maximum,evans1999signal,evans2003maximum,unger2010better}, the practical applications are limited due to a
prohibitively high level of redundancy required~\cite{norman04}.

Here, instead, we are interested in circuits robust to static manufacturing failures. As illustrated previously, this scenario has the advantage of being able to test which parts of the circuit failed and attempt to configure out (or ``wire around'') the defective parts. This side information enables significant savings in redundancy~\cite{nikolic02}. In fact, this method of testing the performance of a device followed by some configuration is rather popular in practice, used in multi-core CPUs~\cite{gizopoulos11}, analog-to-digital
converters~\cite{flynn03}, sense-amplifiers~\cite{verma08}, self-replicating automatons~\cite{mange00}, parallel computing~\cite{leighton85,teramac98}, etc.

This paper can be
seen as an attempt to provide theoretical foundations for the static defect scenario. (In fact, this was our original
motivation.)

\subsection{Problem formulation}

We study the following problem formulation: Given $k$
objects (``primary nodes''), connect each one of them to some of the available $m$ spares
(``redundant nodes'') in such a way that in the event that $t\ge 1$ of the objects fail (originals or
spares) the overall system can be made to function after a repair step. %The repair step consists of replacing each failed primary node with one of the spares that it was connected to. The key restriction is that objects must be one of $q$ types and the spares have to be programmed to one of the $q$ types \textit{before} the failure events are known.
Such a repair step consists of replacing each failed primary node with one of the working spares that it is connected to. Each spare can only replace one failed primary node. The key assumptions are 1) the primary nodes are one of $q$ different types (called labels) 2) the spares have to be programmed to one of the $q$ labels \textit{before} the failure events are known and 3) the same connections need to repair all possible choices of labels for the $k$ primary nodes.  
We are interested in minimizing the number of spare nodes and the number of connections to spare nodes. 

Key to our problem formulation is the idea that we want to design the interconnect (wires) before any of the node labels are determined.
%\footnote{We will return to this and other models in \ref{sec::otherScenarios}}. 
One might argue that in some applications the interconnect could be allowed to depend on the labeling of primary nodes. Indeed, the latter will be known before the final topology for the chip is made. However, our procedure insists that the interconnect does not depend on this labeling. The advantage of this is that in the reconfigurable circuits framework, the provisional wire-layout is usable regardless of where any element is placed, providing the same defect tolerance guarantee for every possible placement. We seek a
\emph{universal} design, which is independent of element types and thus could serve as the new
standard cell for all defect-tolerant circuits. We further discuss alternative design methodologies in
Section~\ref{sec::discussion}. 

We intentionally abstracted our problem to a simple model which is more fundamental and relates to other applications needing redundancy for objects and a universal design. For example, instead of parts of a reconfigurable circuit, objects can represent elements in a programmable logic device (e.g., look-up tables (LUTs) in an FPGA). As part of periodical firmware update, a manufacturer assigns values of LUTs (both primary and redundant) without knowledge of
locations of device-specific failures. Then, a built-in algorithm for each failed LUT $T$ reconnects it to an adjacent spare LUT $R$, with the requirement that $R$ and $T$ be equivalent. This built-in local algorithm is a computationally non-demanding way to reconfigure around defective LUTs. Note that the interconnect of the LUTs need to be universal so that any update chosen by the manufacturer (these updates change the configuration of the primary LUTs) has the same guarantee against defects. 

For $q=2$ our problem is equivalent to finding sparsity vs. edge-size trade-off for $(t,t)$-colorable hypergraphs,
cf.~\cite{alon2009power}. See Section~\ref{sec:hyper}. Other applications potentially arise in warehouse planning, operations research, public safety etc. \textcolor{black}{Such applications can be conceived after
realizing that our interconnect may be thought of as a transportation network between a collection of ``sinks'' and ``sources'' 
so that each sink can be serviced by at least $t$ sources, where each sink has a type and can be
serviced only by sources of the same type.}

Expressed mathematically, we are looking for a $k \times m$ bipartite graph with the property that for any $q$-coloring of the left-side
nodes there is a $q$-coloring of the right-side nodes such that each of the $k$ left-side nodes is connected to at least $t$ nodes of its
color. The goal is to find bipartite designs which have efficient trade-off in redundancy ${m/ kt}$
vs. number of edges. 

The high-level summary of our main findings is that when $q\ge k$, no strategy is better than straightforward $t$-fold replication. When $q<k$, there exist designs that provide savings compared to repetition. We fully or partially characterize the fundamental trade-off between redundancy $m/kt$ and
the average number of edges per primary node in the following cases: 
\begin{enumerate}
	\item $q$, $t$ fixed and $k,m\to \infty$; 
	\item $q$ fixed and $k,m,t\to\infty$; 
	\item $q$, $k$ fixed and  $m, t \to \infty$.
\end{enumerate}
Perhaps surprisingly, in this (combinatorial) problem it is possible to obtain exact analysis for
asymptotics.
The organization of the paper is as follows. Section~\ref{sec:def} introduces the problem formally and overviews main results. Section~\ref{sec:example} demonstrates small-size examples
that show non-triviality of the problem. Sections~\ref{sec::finiteT}
and~\ref{sec::asymSection} address the trade-off in the regime of fixed $t$ and $t\to\infty$ respectively. Finally,
Section~\ref{sec::discussion} discusses implications and extensions of our results.

The notation $[n]$ denotes positive integers $1, 2,... ,n$. The notation $\mathbbm{1}\{\cdot\}$ denotes the indicator function. An underlined letter (e.g, $\vect x$) stands for a vector quantity.

\begin{figure}
\centering
\subfigure[Design]{
	\includegraphics[scale = .2]{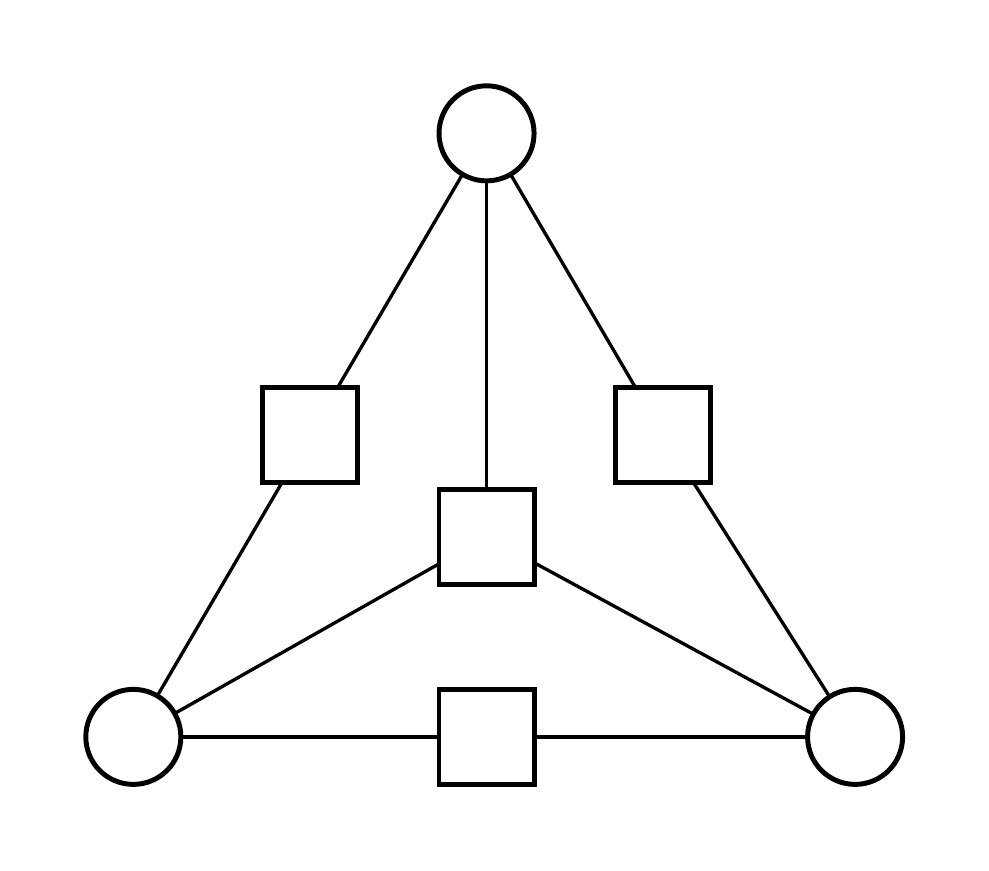}
    \label{fig::exampleBlank}  
}
\subfigure[Primary node labeling]{
	\includegraphics[scale = .2]{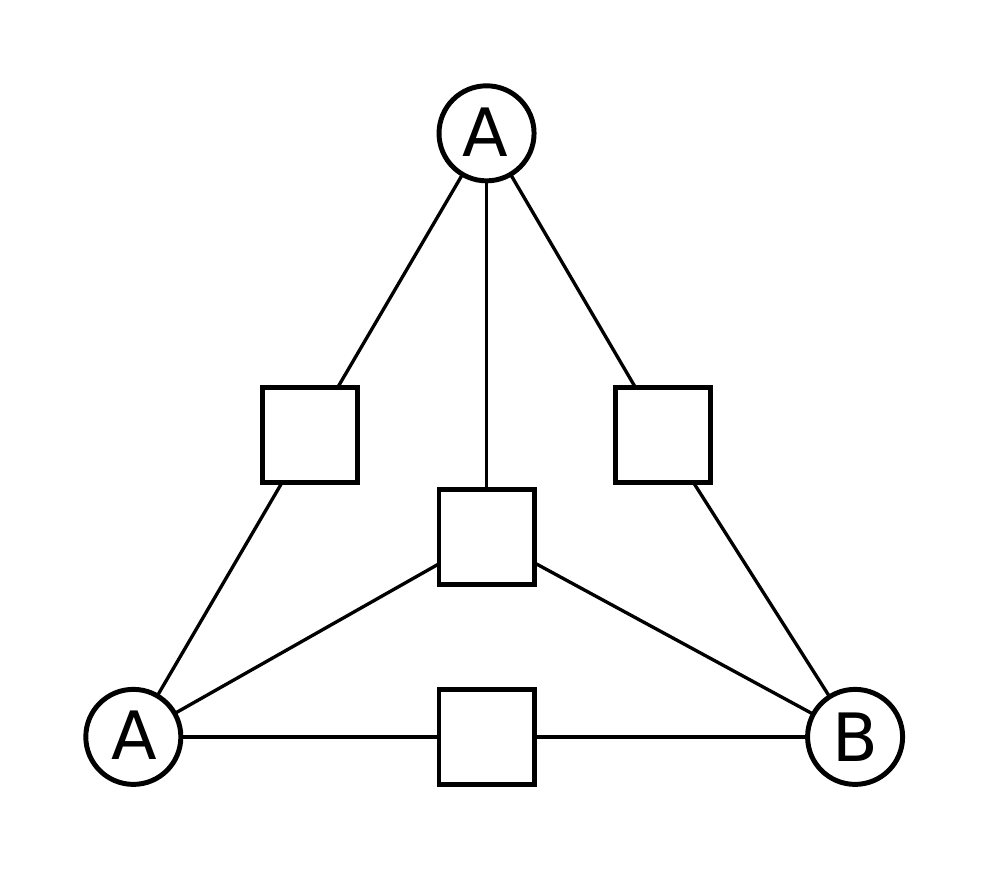}
	\label{fig::exampleFunc}
} 
\subfigure[Redundant node labeling]{
	\includegraphics[scale = .2]{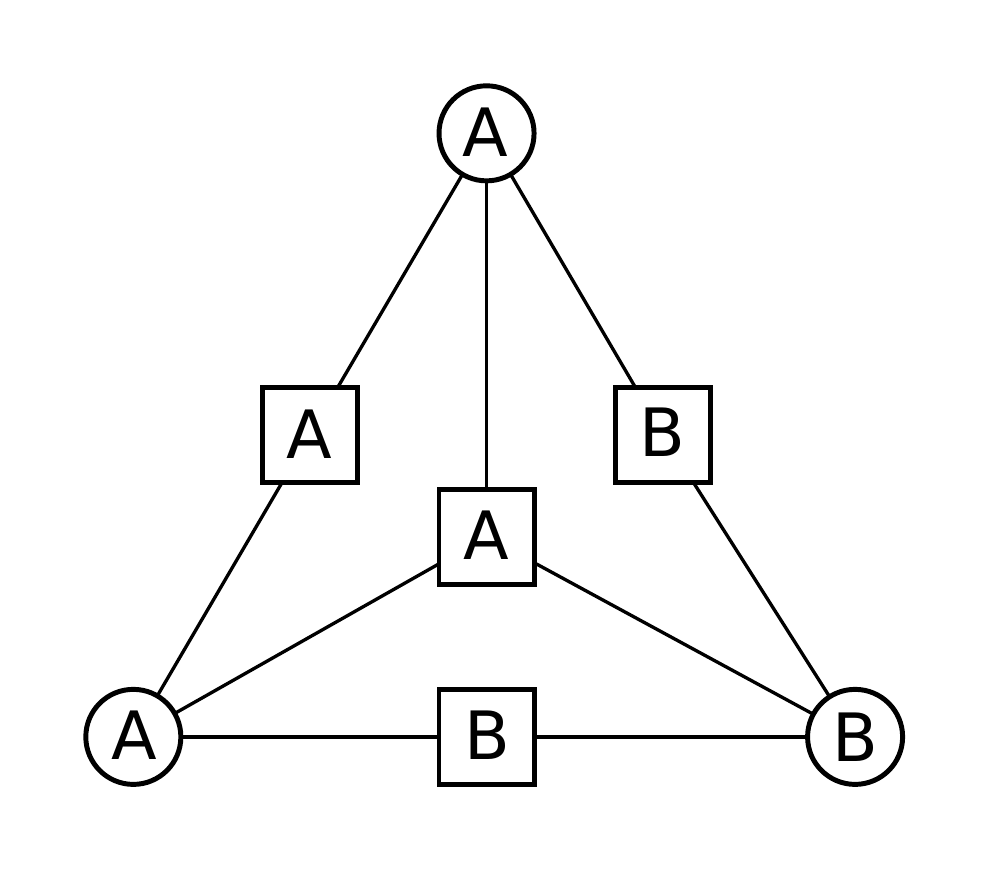}
	\label{fig::exampleRedun}
}

\caption{\label{fig::exampleDesign} Example of a $2$-defect correcting design for an alphabet $\matx = \{A, B\}$ of size 
$q = 2$. \textcolor{black}{The design is bipartite with the circles representing the left-side nodes and the squares representing the right-side nodes.} \Cref{fig::exampleFunc} shows a labeling of the primary (circle) nodes. To each such labeling, 
we strategically choose a labeling of the redundant (square) nodes, so that each primary node has $t=2$ neighbors with matching
labels (see \Cref{fig::exampleRedun}). Since such a choice is possible for each of the $2^3 = 8$ possible labelings of primary nodes, we
conclude that this design is $t = 2$ defect correcting in the sense of \Cref{def:main} \textcolor{black}{and is a $(3,4,2,9)_2$-design}.}
\end{figure}

%------------------PROBLEM SETUP/MAIN RESULTS--------------------%
%------------------PROBLEM SETUP/MAIN RESULTS--------------------%
%------------------PROBLEM SETUP/MAIN RESULTS--------------------%
%------------------PROBLEM SETUP/MAIN RESULTS--------------------%

\section{Problem setup and main results}\label{sec:def}

\subsection{Defect-tolerance model} 

This paper focuses on bipartite graph designs.\footnote{\textcolor{black}{The word choice of ``design'' is not intended to relate to the notion of
combinatorial (Steiner) designs or any other established mathematical definitions.}} The left-side nodes of the bipartite graph are called the primary nodes. These are denoted by circles and there are $k$ of these in the bipartite graph. The right-side nodes are the redundant nodes. These are denoted by squares and there are $m$ of these in the bipartite graph.

Let $\matx$ be a finite alphabet where $ q =|\matx|$. 

\begin{definition} Fix an alphabet of labels $\matx$ with size $q$. A $k \times m$ bipartite graph is called a \emph{$t$-defect correcting design}
if for any labeling of $k$ primary nodes with elements of
$\matx$ there exists a labeling of $m$ redundant nodes with elements of $\matx$
such that every primary node labeled $x\in\matx$ has at least $t$ neighbors labeled $x$.
We will call such a graph a $(k,m,t,E)_q$-design, with $E$ denoting the number of
edges. (See \Cref{fig::exampleDesign} for an illustration.)
\label{def:main}
\end{definition}
This paper is devoted to characterizing the fundamental trade-off between the two basic parameters of $t$-defect correcting designs: redundancy and wiring complexity. The redundancy of a $(k,m,t,E)_q$-design is $\rho={m/(kt)}$. The wiring complexity (or average degree per defect) of a $(k,m,t,E)_q$-design is $\varepsilon={E/(kt)}$.
This trade-off can be encoded in a two-dimensional region:
\begin{definition} For a fixed $q$ and $t\ge 1$ we define the region $\matr_t$ as the closure of the set of all achievable
pairs of $(\varepsilon, \rho)$:
\begin{equation}\label{eq:rtdef}
	\mathcal{R}_t \eqdef \mathrm{closure} \left\{\left(\frac{E}{kt}, \frac{m}{kt}\right): \exists (k,m,t,E)_q\text{-design}\right\}\,.
\end{equation}
\end{definition}

%The exact relation to the previous definition of the $t$-defect correcting design is as follows: the $k$ primary nodes in our model represents the placeholders intended for the components which are necessary for the chip to operate and the redundant nodes represents the placeholders for the redundant components. The labeling we apply to the nodes is the choice of components for each space. The edges correspond to the provisional wires. 

 To interpret between \Cref{def:main} and reconfigurable circuits (and other applications), we present the following association. 

\begin{proposition} An interconnect for a reconfigurable circuit can tolerate any $t$ manufacturing defects for any choice of primary nodes if and only if the interconnect is a $t$-defect correcting design.
\end{proposition}

\begin{proof}
If the interconnect corrects fewer than $t$ defects, there is some primary node labeling where any labeling of the redundant nodes would result in some primary node with label $x$ having fewer than $t$ neighbors with the same label $x$. If this primary node and all its matching neighbors have defects, then the defect in the primary node cannot be corrected.

If the interconnects is a $t$-defect correcting design, with the correct redundant node labeling scheme, any primary node labeled $x$ has $t$ redundant neighbors with the same label $x$. If there are only $t$ defects, either this primary node is working, or this primary node has a defect and at most $t-1$ of its neighbors have a defect or are used to correct another primary node. In the latter case, there is at least one redundant node with label $x$ available which can be used to replace this primary node. 
\end{proof}

As noted earlier, our performance metrics, $\rho$ and $\varepsilon$, correspond to the extra silicon area and wiring (and fan-out) required respectively for defect-tolerance.

Before proceeding further, we summarize some of the basic properties of regions $\matr_t$.
\begin{proposition}(Properties of $\matr_t$)\label{prop:basic}
Regions $\matr_t$ satisfy the following:
\begin{enumerate}
\item $(\varepsilon,\rho)\in\matr_t$ iff there exists a sequence of $(k,m,t,E)_q$-designs with ${E\over kt}\to\varepsilon,{m\over kt}\to\rho$ as $k,m\to\infty$; \label{prop:basic:1}
\item If $(\varepsilon,\rho) \in \matr_t$ and $\varepsilon'\ge \varepsilon,\rho'\ge \rho$ then $(\varepsilon',\rho')\in \matr_t$; \label{prop:basic:2}
\item $\matr_t$ are closed convex subsets of $\mreals_+^2$; \label{prop:basic:3}
\item We have
\begin{equation}\label{eq:rinfty}
	 \limsup_{t\to\infty} \matr_t = \mathrm{closure} \left\{\bigcup_{t\ge 1}\matr_t \right\} \eqdef \matr_\infty\,.
\end{equation} \label{prop:basic:4}
\item The limiting region $\matr_\infty$ is also a closed convex subset of $\mreals_+^2$ characterized as
\begin{equation}\label{eq:rinfty2}
	\mathcal{R}_\infty \eqdef \mathrm{closure} \left\{\left(\frac{E}{kt}, \frac{m}{kt}\right): \exists (k,m,t,E)_q-\text{design}\right\}\,.
\end{equation}\label{prop:basic:5}
\end{enumerate}
\end{proposition}
See Section~\ref{sec:proofbasic} for proofs.

%-----------------------MAIN RESULTS-----------------------%
%-----------------------MAIN RESULTS-----------------------%
%-----------------------MAIN RESULTS-----------------------%
%-----------------------MAIN RESULTS-----------------------%

\subsection{Preview of main results for binary alphabet}

\textcolor{black}{Notice the wiring complexity and redundancy metrics represent the linear scaling between the quantities $E$ and $m$ respectively with the product $kt$. Designs which satisfy \Cref{def:main} must have the number of redundant nodes and number of edges grow linearly with the product $kt$. The goal of our results is to find a tight understanding of the coefficient in this linear scaling.}\footnote{\textcolor{black}{For all our results, $q$ is always fixed. How wiring complexity and redundancy scales with $q$ is left for future work.}}

In this section, for the purpose of illustration, we give a summary of our results for the case of binary alphabet $\matx$ (i.e., $q = 2$). The rest of the paper will present various bounds and constructions which apply to general alphabet sizes, (i.e., arbitrarily values of $q$).

There are three separate results which are the main contributions of this paper. One is characterizing the region $\matr_t$ in the regime where $t$ is small, specifically for values where $t = 1$ and $t = 2$. The second main result is characterizing the region $\matr_\infty$, which corresponds to the limit of regions $\matr_t$ when $t$ tends to infinity. The third is characterizing the result when the number of primary nodes $k$ is finite (the first two results have infinite $k$) and $t$ tends to infinity.

The theorem for the small $t$ case is the following:  

\begin{theorem}\label{thm::finiteResult}
For binary alphabet $\matx$, if $t = 1$ or $t = 2$, we have
\begin{equation}
\label{eq::linRegion}
\mathcal{R}_t = \{(\varepsilon, \rho): \varepsilon \geq 1, \rho \geq 0 \text{ and } \varepsilon \geq 2 - \rho\}\,.
\end{equation}
\end{theorem}

This will be proved in Section \ref{sec:covering}. \textcolor{black}{The immediate conclusion from this result is that the designs for $t = 1$ and $t = 2$ achieve the same number of redundant nodes and edges needed per primary node per defect asymptotically over $k$. The region in \Cref{thm::finiteResult} has two corner points.}  We will also discuss the designs which attain these corner points.

The theorem for the asymptotic $t$ case is the following:

\begin{theorem}\label{thm::asymResult}
Let $\matx$ be a binary alphabet. The region $\matr_\infty$ defined in~\eqref{eq:rinfty} is the closure of the set of points
$(\varepsilon,\rho)$ defined as follows. For every distribution $P_S$ on $\mathbb{Z}_+$ with finite support, we define 
\begin{subequations}
\begin{equation}\label{eq:asymResult}
	\varepsilon = {\EE[S]\over F(P_S)}, \quad \rho = {1\over F(P_S)}\,,
\end{equation}
where
\begin{equation}\label{eq::optfunc}
\begin{split}
F(P_S) &\eqdef \min_{0 \leq \lambda \leq 1} \max_{0 \leq f(\cdot, \cdot) \leq 1} \min \left\{
\EE\left[\frac{L_0}{\lambda} f(L_0,L_1)\right], \right. \\
& \quad \quad \quad \quad \left. \EE\left[\frac{L_1}{1-\lambda} (1-f(L_0,L_1))\right] \right\}
\end{split}
\end{equation}
\end{subequations}
where the expectations are over $S\sim P_S$ and given $S$ the distribution of $L_1 \sim \text{Bino}(S, \lambda)$ and $L_0=S-L_1$. 
\end{theorem}

This theorem parametrically characterizes $\matr_\infty$ in terms of the function $F(P_S)$, which is evaluated on every $P_S$ with finite support. Note that evaluation of the bound~\eqref{eq:asymResult} is non-trivial as we will discuss in Section~\ref{sec::numSection}.

The generalization of \Cref{thm::asymResult} to larger alphabet sizes is given by \Cref{thm::asymResultLarger} and is developed in Section \ref{sec::asymSection}. Here the designs achieving the best trade-off are more complicated than those associated with Theorem~\ref{thm::finiteResult}. We call them \emph{subset designs} and develop them in \Cref{sec:subsetdesigns}. 

The resulting achievable regions for \Cref{thm::finiteResult} and \Cref{thm::asymResult} are depicted in \Cref{fig::AchRegion}. Via these results we can determine at any fixed redundancy level, how many connections are necessary. For example, at redundancy level $10\%$, the figure indicates that there exists designs which:
\begin{itemize}
\item correct $1$ defect if each primary node is connected on average to about $1.9$ redundant nodes
\item correct $2$ defects if each primary node is connected on average to about $1.9 \times 2$ redundant nodes
\item correct $10^3$ defects if each primary node is connected on average to about $1.7 \times 10^3$ redundant nodes.
\end{itemize}

\textcolor{black}{Immediate from \Cref{fig::AchRegion} is that the region $\matr_\infty$ contains the regions $\matr_1$ and $\matr_2$ implying that increasing the number of defects $t$ allows for lower redundancy and wiring complexity (recall both these quantities are divided by $t$). In this sense, it is more efficient to correct more defects.}  

According to~\eqref{eq:rinfty2} all regions $\mathcal{R}_t$ will lie between $\mathcal{R}_1$ and $\mathcal{R}_\infty$, approaching the
latter as $t\to\infty$. It is perhaps surprising that unlike most known asymptotic combinatorial problems, this one (for $t\to\infty$) admits a relatively simple solution.  

The third and the more practically useful result is the characterization of the achievable regions for asymptotic $t$ but with finite $k$. This is developed in \Cref{sec::finitek}. 

%We postpone further discussion of all the results till Section~\ref{sec::discussion}.

\begin{figure}
\centering
\includegraphics[scale = .45]{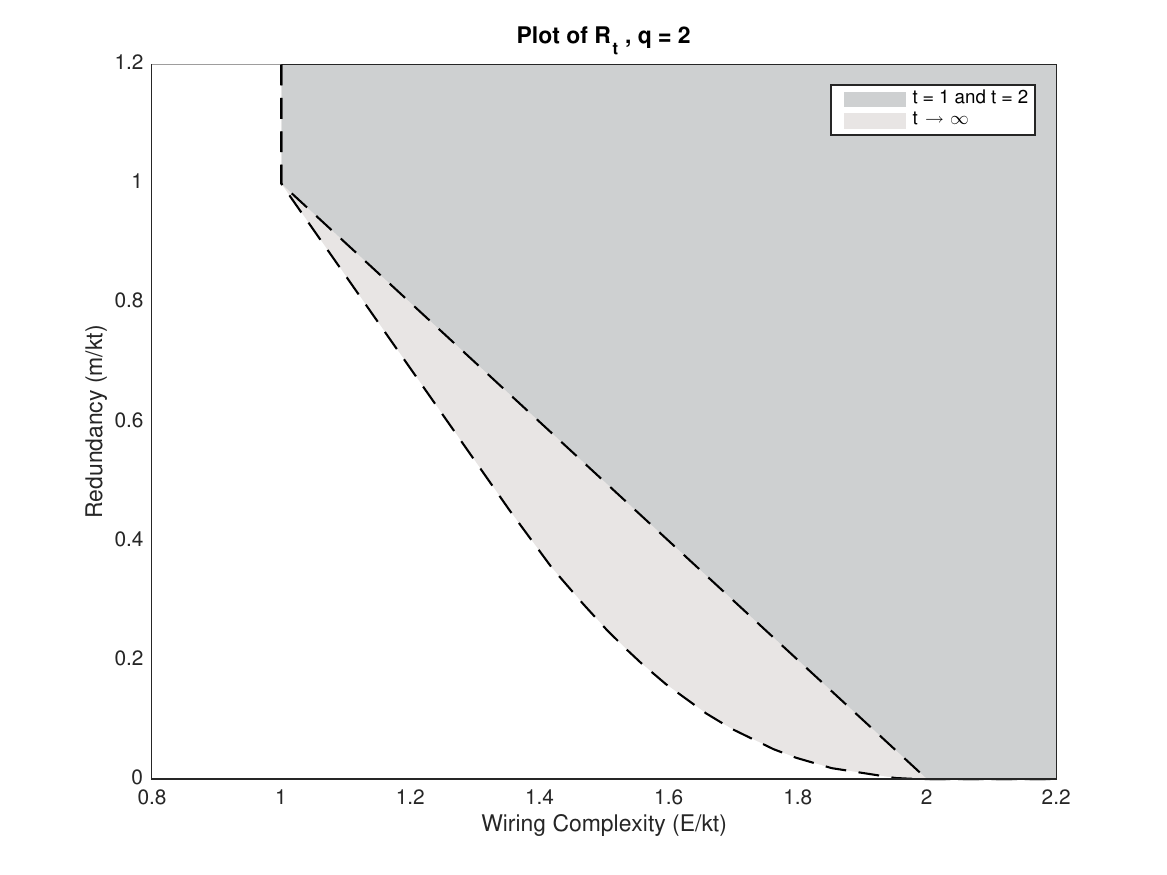}
\caption{Achievable regions for redundancy and wiring complexity trade-off when $q = 2$. Regions $\matr_1$ and $\matr_2$ are shown in darker gray. Region $\matr_\infty$ includes lighter and darker gray areas. All other regions $\matr_t$ lie between $\matr_1$ and $\matr_\infty$. The boundary of the region $\matr_\infty$ is calculated using the methods in \Cref{sec::numerical}.
\label{fig::AchRegion} }
\end{figure}

%-----------------------EXAMPLES-----------------------%
%-----------------------EXAMPLES-----------------------%
%-----------------------EXAMPLES-----------------------%
%-----------------------EXAMPLES-----------------------%

\section{Examples of good designs}\label{sec:example}

\begin{figure}
\centering
	\subfigure[Example of \textbf{complete design}. (Design written as $K(3,4)$). This design is $2$-defect correcting for $q = 2$ and $4$-defect correcting for $q = 1$.]{\hskip 40pt
	\includegraphics[scale = .2]{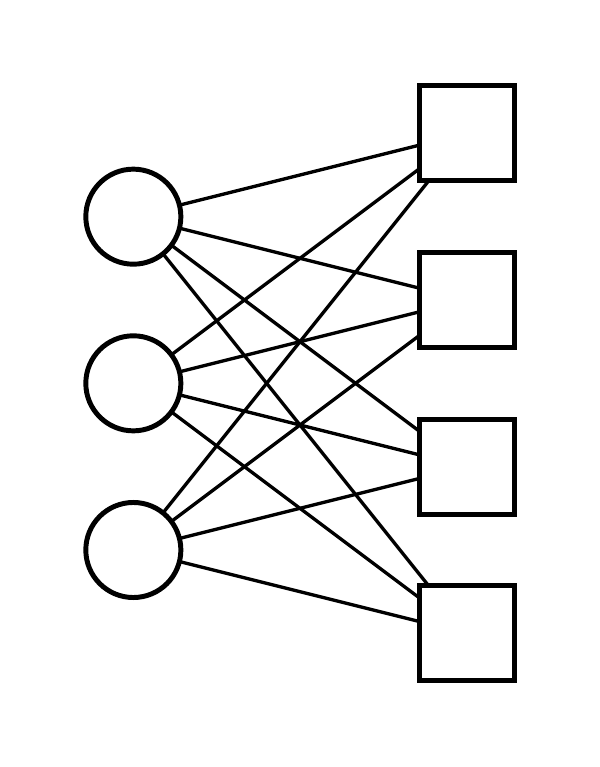}\hskip 40pt
	\label{fig::completeDesign1}
} \\[1em]
\subfigure[Example of \textbf{repetition design}. \textcolor{black}{(Design written as $2K(1,3)$)}. This design corrects 3 defects for any $q$.]{\hskip 40pt
	\includegraphics[scale = .2]{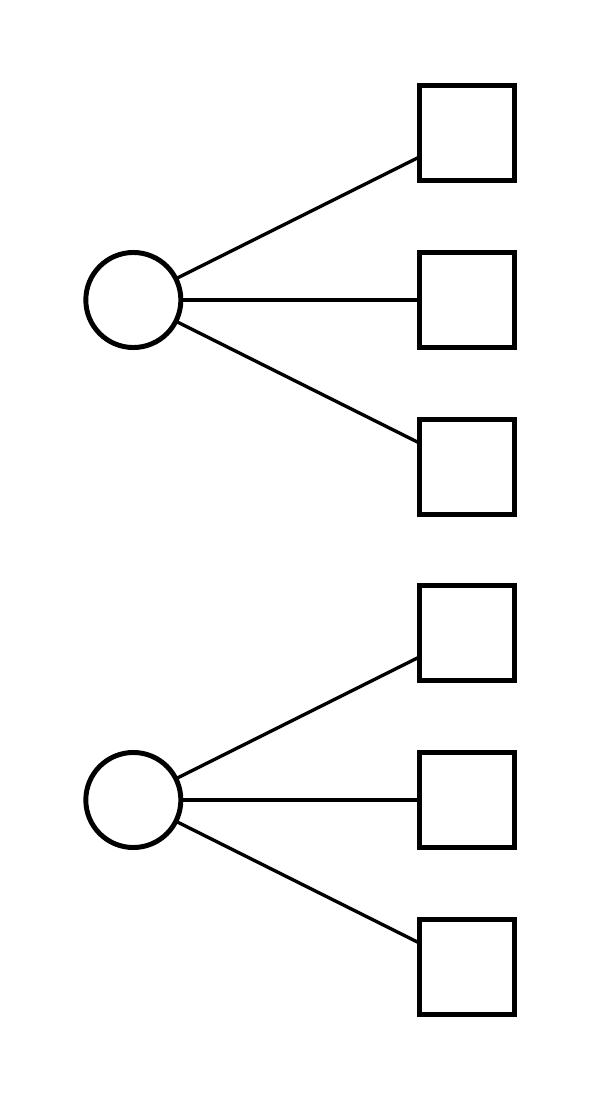}\hskip 40pt
	\label{fig::blockDesign1}
}
\caption{Two elementary designs.}\label{fig:elemdesigns}
\end{figure}

Before developing the main results, we will first introduce a few basic designs and analyze their performances. Some of these examples play major roles in subsequent developments. 

We denote by $K(k,m)$ a complete bipartite graph with $k$ primary nodes (circles) and $m$ redundant nodes (squares). The two most basic designs are the following:
\begin{enumerate}
\item \textit{Complete designs:} $K(k, qt)$ (recall that $q = |\matx|$) is $t$-defect correcting. (Just label
redundant nodes to hold $t$ copies of each value $\matx$. No matter how the primary nodes are labeled, each primary node will be connected to $t$ redundant nodes with the same label as itself.) See
\Cref{fig::completeDesign1} for illustration. 
\item \textit{Repetition designs:} $K(1,t)$ is capable of correcting $t$ defects over an arbitrary alphabet. (Just label all $t$ redundant nodes the same label as the neighboring primary node.) Taking $k$ disjoint copies of $K(1,t)$, denoted by $kK(1,t)$, we get a repetition design achieving $\rho=\varepsilon=1$. See
\Cref{fig::blockDesign1} for illustration. 
\end{enumerate}

If we take $k\to \infty$, the complete design achieves $\varepsilon=q$ and $\rho= \frac{qt}{kt} \rightarrow 0$ for any fixed $t$ and $q$, which is the best possible trade-off given the value of $\varepsilon$. For finite $k$, however, the complete design is not the design with the minimal number of edges: it is
possible to remove some of the edges and still maintain a $t$ defect correcting property, as we will show in the next
subsection. 

The repetition design uses the minimal number of edges (since any primary node needs at least $t$ edges in order to be a $t$-defect correcting design). If all primary nodes have exactly $t$ edges, then it is necessary for each primary node to have a distinct neighborhood, illustrating that the repetition design achieves the best trade-off at minimal wiring complexity. 

\subsection{Smallest non-trivial designs} \label{sec:smallOptimal}

\begin{figure}
\centering
\subfigure[$q=2$]{
	\includegraphics[scale = .18]{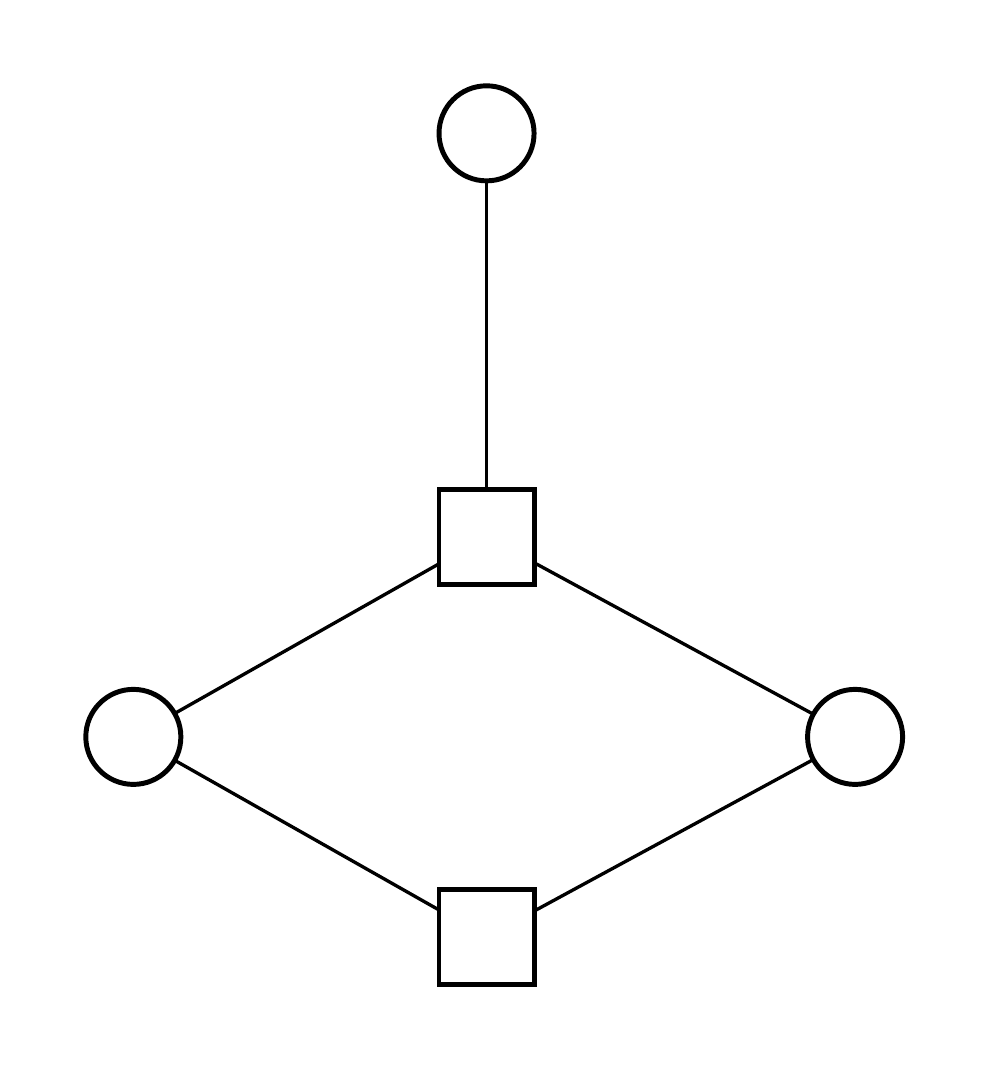}
    \label{fig::bint1}
}\\[1em]
\subfigure[$q=3$]{
	\includegraphics[scale = .18]{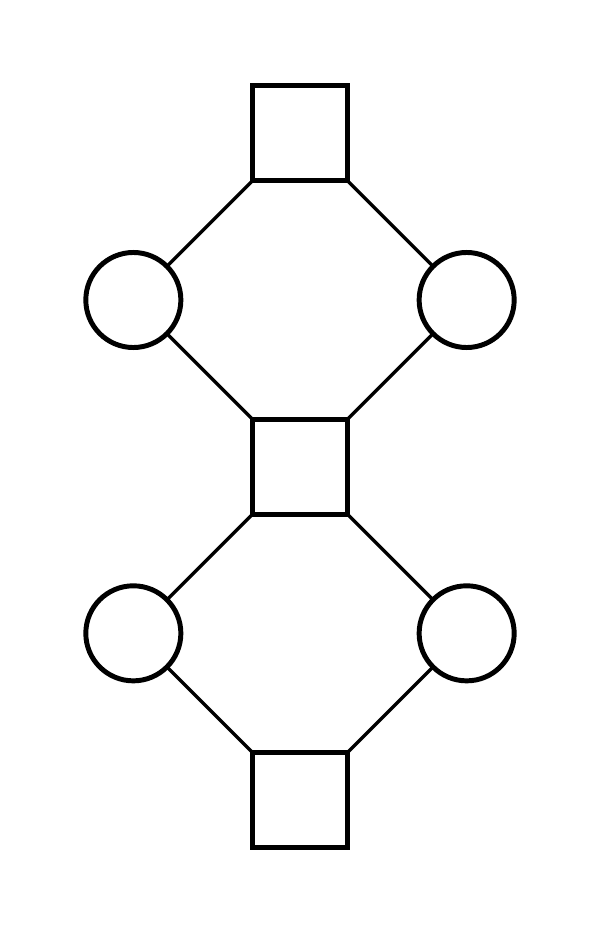}
	\label{fig::tert1}
} \\[1em]
\subfigure[$q=4$]{
	\includegraphics[scale = .18]{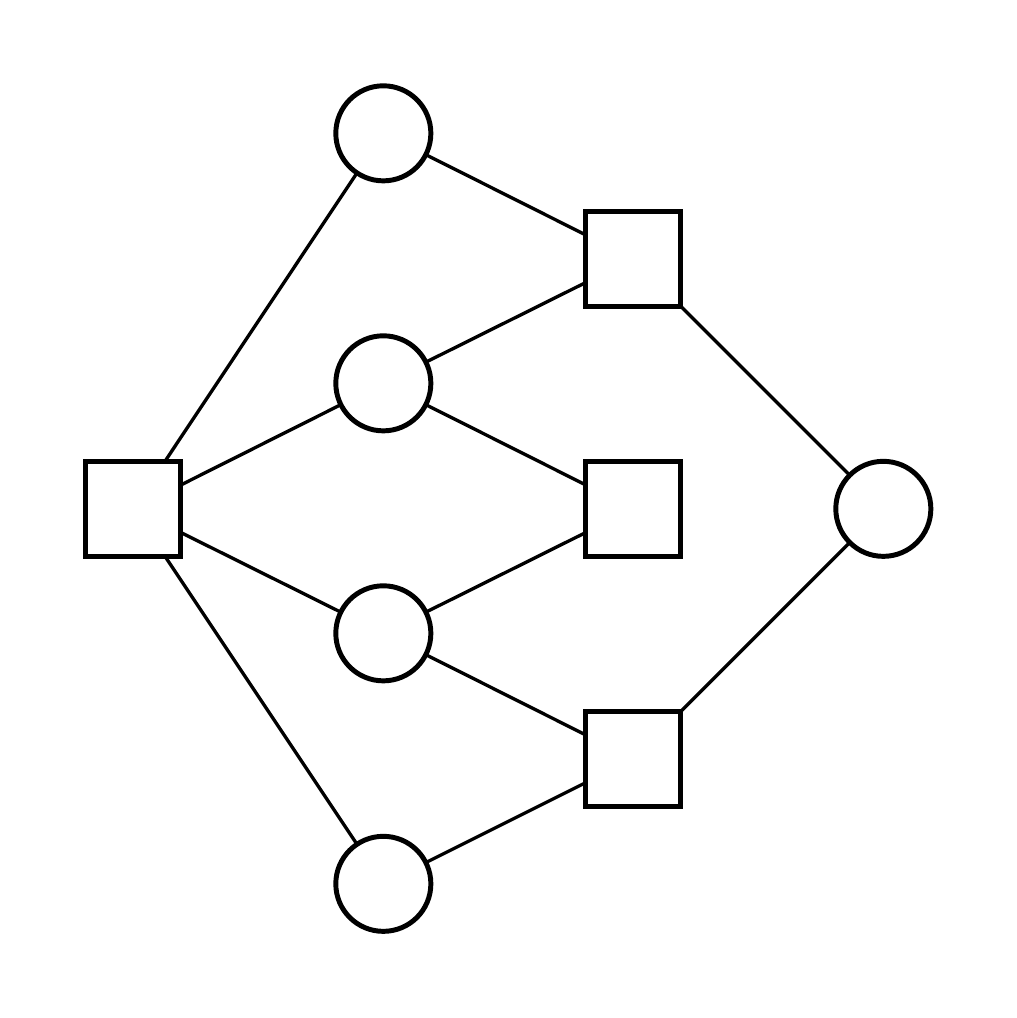}
	\label{fig::quat1}
}\\[1em]
\caption{Smallest non-trivial 1-defect correcting designs. }
\label{fig::oneecc}
\end{figure}

\begin{figure}
\centering
\subfigure[$q=2$]{
	\includegraphics[scale = .18]{plots/hammingDesign1.pdf}
	\label{fig::hammingDesign1}
}\\[1em]
\subfigure[$q=3$]{
	\includegraphics[scale = .18]{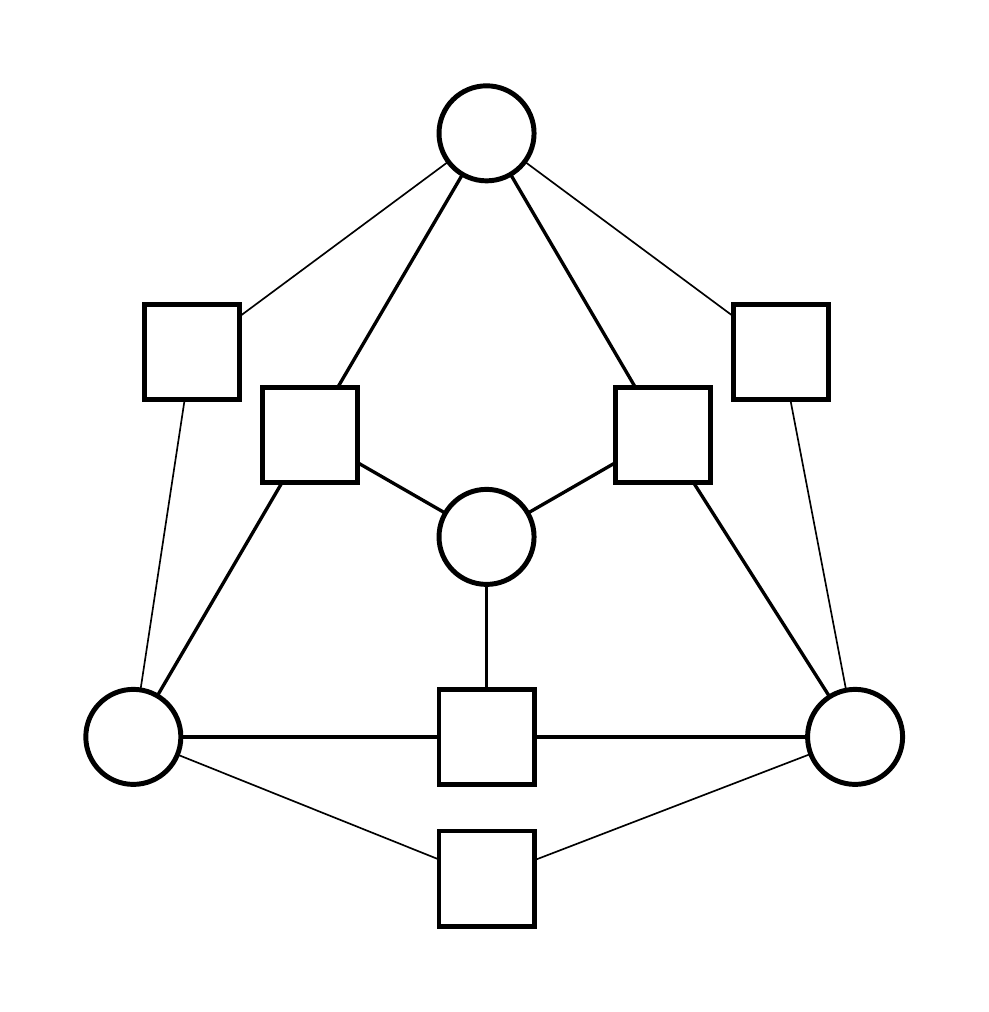}
	\label{fig::HB3}
}\\[1em]
\subfigure[$q=4$]{
	\includegraphics[scale = .18]{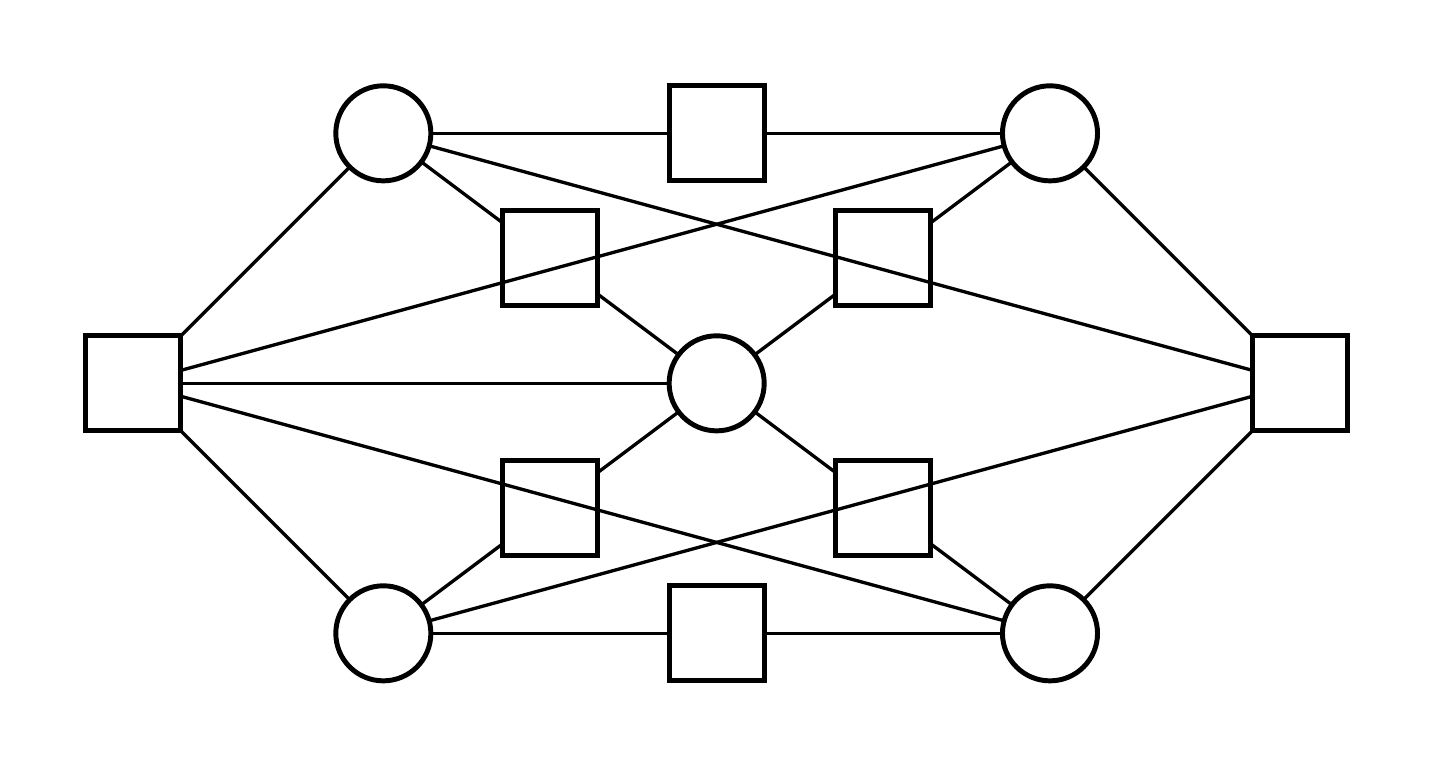}
	\label{fig::HB4}
}\\[1em]
\caption{Smallest non-trivial 2-defect correcting designs.}
\label{fig::twoecc}
\end{figure}

We now present designs which have the fewest number of edges given some fixed number of primary nodes $k$, redundant nodes $m$, and defect tolerance $t$. 

If $k\le q$ then all
primary nodes can have different values and thus one is forced to use the repetition design $kK(1,t)$ to correct $t$
defects. For $k=q+1$ the
question becomes more interesting. First, notice that the minimal possible $m$ equals $qt$ (this is achieved by
the complete design and cannot be reduced \textcolor{black}{since $t$ nodes with the same label can have defects and this can occur for each of the $q$ different labels}). However, some of the edges can be removed from the complete design while still preserving the number of defects corrected.

The optimal designs with $k=q+1$, $m=q$ and $t=1$ are as follows:
\begin{itemize}
\item Binary alphabet $(q = 2)$: $k=3$, $m=2$ with $5$ edges. See \Cref{fig::bint1}.
\item Ternary alphabet $(q = 3)$: $k=4$, $m=3$ with $8$ edges. See \Cref{fig::tert1}. (There exist two non-isomorphic optimal designs. \Cref{fig::tert1} shows the symmetric one.)
\item Quaternary alphabet $(q = 4)$: $k=5$, $m=4$ with $12$ edges. See \Cref{fig::quat1}. (There exist multiple non-isomorphic optimal designs. Only one is shown.)
\end{itemize}

The optimal designs with $k=q+1$, $m=2q$ and $t=2$ are as follows:
\begin{itemize}
\item Binary alphabet $(q = 2)$: $k=3$, $m=4$ with $9$ edges, see \Cref{fig::hammingDesign1}. This design is what we call the \textit{Hamming block}. \Cref{fig::exampleDesign} shows how it can correct $2$ defects. We will discuss its optimality in \Cref{prop::hammingOptimal}.
\item Ternary alphabet $(q = 3)$: $k=4, m=6$ with $15$ edges, see \Cref{fig::HB3}. (There exist two non-isomorphic optimal designs. \Cref{fig::HB3} shows the symmetric one.)

\item Quaternary alphabet $(q = 4)$: $k=5, m=8$ with $21$ edges, see \Cref{fig::HB4}. 
\end{itemize}

Some of these designs were found analytically and others by exhaustive search. %None lead to particularly interestingly corresponding design for higher alphabets or general constructions. 
None of these designs are at the performance boundary of any $\matr_t$ regions. To obtain designs that near the optimal trade-off boundary, we need to use a larger number of primary and redundant nodes (see
\Cref{prop:simpach}). However, a few of these designs, like the Hamming block in \Cref{fig::hammingDesign1}, achieve the best trade-off when restricted to the finite $k$ setting (as we will develop in \Cref{sec::finitek}).

%------------------SUBSET DESIGNS--------------------%
%------------------SUBSET DESIGNS--------------------%
%------------------SUBSET DESIGNS--------------------%
%------------------SUBSET DESIGNS--------------------%

\subsection{Subset designs} \label{sec:subsetdesigns}

Designs that form a key ingredient of our asymptotic (i.e., large $t$) constructions are subset designs. A subset design $S(k, s)$ is a bipartite graph with $k$ primary nodes and $m={k \choose s}$ redundant nodes, each connected to a distinct $s$-subset of $\{1,\ldots,k\}$. Note that the degree of each primary node is ${k-1\choose s-1}$.

In general, we allow subset designs to have multiple and possibly different subset sizes. For two values $s_1$ and
$s_2$, where $s_1, s_2 \in [k]$, a bipartite graph $S(k,s_1) \vee S(k,s_2)$  is defined to be the result of identifying the
$k$ primary nodes in two disjoint copies of
$S(k,s_1)$ and $S(k,s_2)$. The resulting graph has $k$ primary nodes and
$m={k\choose s_1} + {k\choose s_2}$ redundant nodes. We call the operation ($\vee)$ graph merging, which we state more precisely below. We will develop the properties of merging later.

\begin{definition}[Merging]\label{def::merging}
For any collection of designs $G_j$ on the same number of primary nodes $k$, the merging of $G_j$,
denoted $G=\bigvee_j G_j$ is a graph formed by taking disjoint copies of $G_j$ and identifying primary nodes.
\end{definition}
\textcolor{black}{See \Cref{fig::merge} for an illustration of merging.}

\begin{figure}
\centering
\includegraphics[scale = .15]{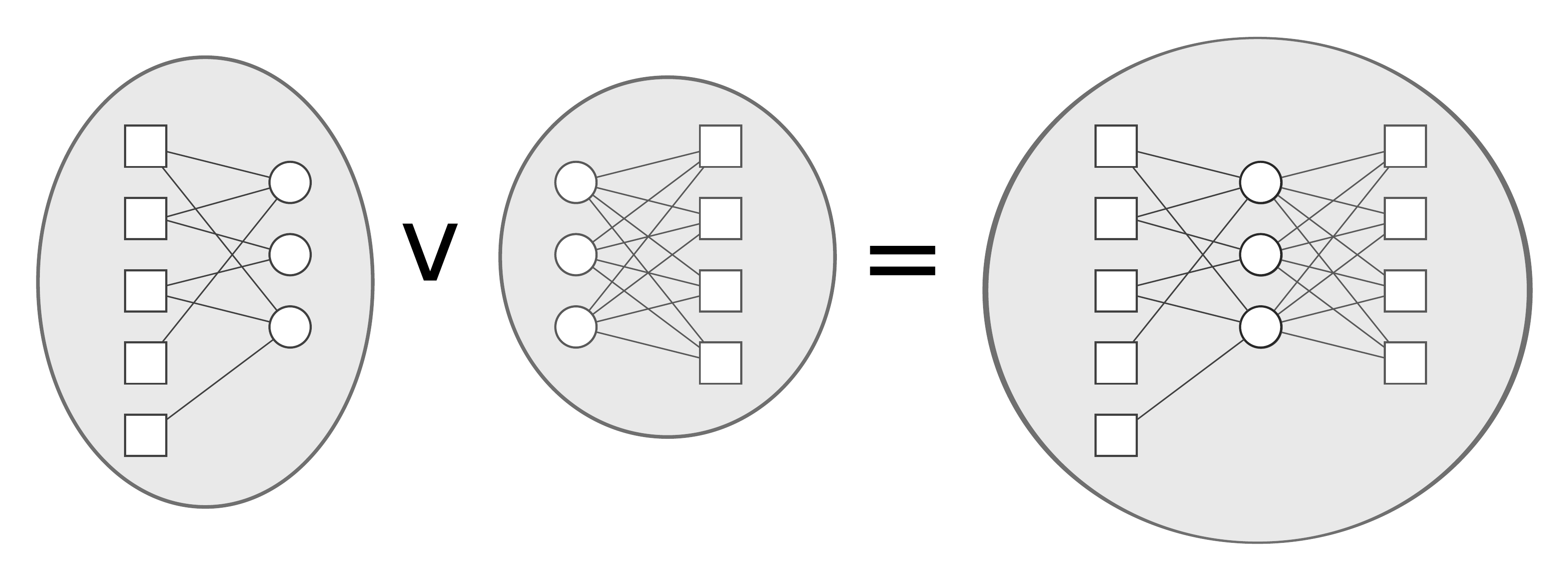}
\caption{\label{fig::merge} \textcolor{black}{Example of merging two designs.}}
\end{figure}

\begin{definition}[Subset design] Given $k$ and (not necessarily distinct) positive integers $s_1,s_2,\ldots,s_r \in [k]$,  
\begin{equation}
S(k,s_1)\vee S(k,s_2)\vee\cdots\vee S(k,s_r)
\end{equation}
is a subset design with $k$ primary nodes and $ m = \sum_{j=1}^r {k\choose s_j} $ redundant nodes.
\end{definition}

For example, the Hamming block, \Cref{fig::hammingDesign1}, is $S(3, 2)\vee S(3,3)$, the repetition design is
$S(k,1)\vee \cdots \vee S(k,1)$ ($t$ times) and the complete design is $S(k, k)\vee \cdots \vee S(k,k)$ ($qt$ times). \textcolor{black}{\Cref{fig::subset} shows the subset design $S(4,3) \vee S(4,2)$.}

\begin{figure}
\centering
\includegraphics[scale = .2]{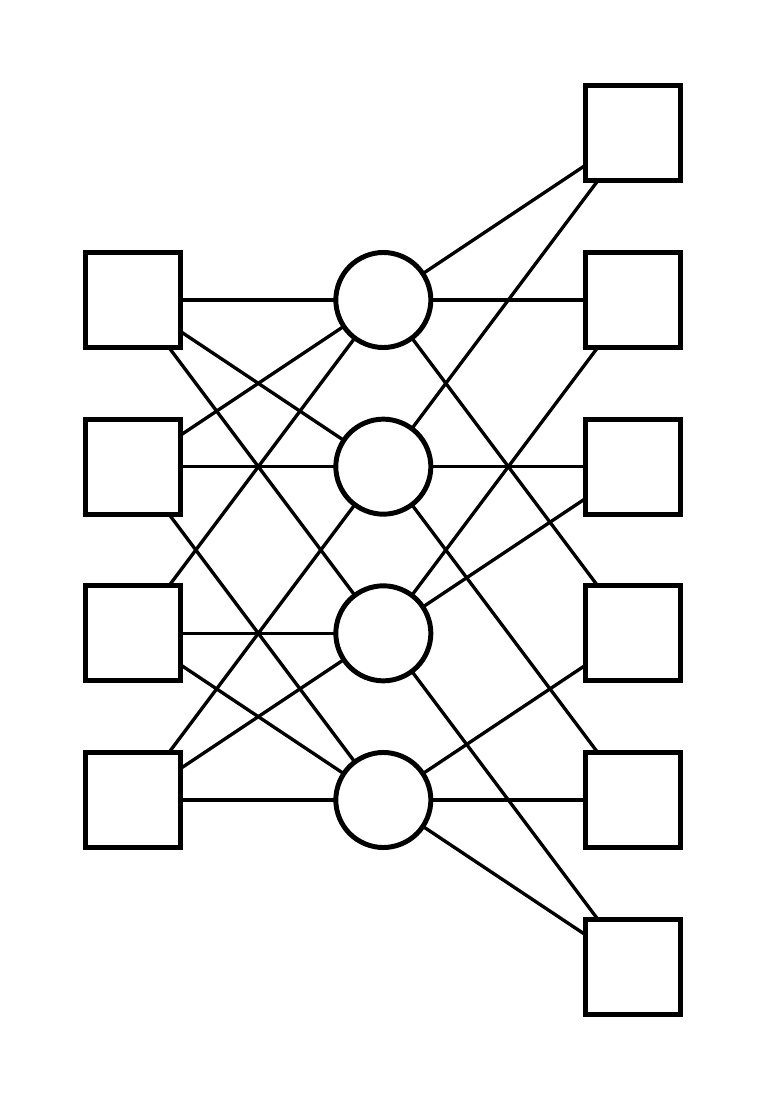}
\caption{\label{fig::subset} \textcolor{black}{Example of a subset design. This design is $S(4,3) \vee S(4,2)$. The redundant nodes corresponding to $S(4,3)$ are shown on the left side and those corresponding to $S(4,2)$ are shown on the right side.}}
\end{figure}

Subset designs are characterized by the following property:\begin{definition}[Permutation invariance]
A design is called \emph{permutation invariant} if there exists a group of bipartite-graph automorphisms (thus
preserving the left/right partition) that acts as the full symmetric group $S_k$ on primary nodes.\footnote{\textcolor{black}{For those not familiar with bipartite-graph automorphisms: Consider identifying each primary node and redundant node in the design with a distinct number. Primary node number $i$ is connected to some set of numbered redundant nodes $M_i$. We can equivalently say a design is permutation invariant if for all possible permutations of the numbers of the primary nodes, there is a way to permute the numbers of the redundant nodes, so that the new design still has primary node $i$ connected to the set of redundant nodes $M_i$.}} 
\end{definition}

\begin{proposition}\label{prop:perminv} A design is permutation invariant if and only if it is a subset design.
\end{proposition}
\begin{proof} Invariance of subset designs is clear. Conversely, given a permutation invariant design and an integer
$s\ge 1$, consider the subgraph induced by all degree-$s$ redundant nodes and their neighborhoods. By permutation invariance this subgraph must
contain all $k$ primary nodes and itself be permutation invariant (since automorphisms preserve degrees of nodes).
Therefore, every $s$-subset of the primary nodes must appear as a neighborhood of $n$ redundant nodes for some integer $n$. This degree-$s$ subgraph corresponds to merging of $n$ copies of
$S(k,s)$ and the original graph is a merging of degree-$s$ subgraphs.
\end{proof}

The number of redundant nodes used in subset designs is large and therefore it should be able to correct many defects. We
will find sharp estimates for the defect-correcting properties of subset designs later (\Cref{prop::subsetFiniteKResult} below), but for now we can give a simple order-of-magnitude result:

\begin{proposition}\label{prop:subset0} Fix alphabet $\matx$ and size $s\ge 1$. As $k\to\infty$ the design $S(k,s)$ corrects $t = \Theta(k^{s-1})$ defects.
\end{proposition}

\begin{proof}
We know that $t = O(k^{s-1})$, since each primary node has at most \textcolor{black}{$\binom{k-1}{s-1} = O(k^{s-1})$} neighbors. To show that $t =
\Omega(k^{s-1})$, fix a labeling of the primary nodes with the elements of $\matx$. Consider the following procedure for
labeling redundant nodes. First we declare an element of $x\in\matx$ to be \emph{rare} if the number of primary nodes
labeled $x$ is less than \textcolor{black}{${k\over q}$}. % where $c$ is a suitable constant. 
Now each redundant node is labeled the
value $x\in\matx$ if either all of its neighbors have label $x$ or if $x$ is the only rare label in its
neighborhood. \textcolor{black}{(Some redundant nodes may not be labeled, but the contribution from these nodes can be disregarded for this particular order of magnitude result.)} To see that this is an labeling that corrects $\Omega(k^{s-1})$ defects, simply notice that a non-rare
primary node labeled $x$ has at least \textcolor{black}{$\binom{k/q - 1}{s-1} = \Omega(k^{s-1})$} neighboring redundant nodes with all neighbors labeled $x$. 
\textcolor{black}{Similarly, for any choice of non-rare label $x$, each rare-labeled primary
node has at least $\binom{k/q - 1}{s-1} = \Omega(k^{s-1})$ neighboring redundant nodes connected to it such that all other neighbors of this redundant node is labeled $x$. Since $x$ is non-rare, these $\Omega(k^{s-1})$ are labeled the value of the rare primary node.}
\end{proof}

As we will see, subset designs turn out to be optimal for achieving the boundary of $\matr_\infty$. In other words, they can be
tuned to get the optimal speed of growth for redundancy and wiring complexity as $t\to\infty$.

%----------------------FINTIE BOUNDS-------------------------%
%----------------------FINTIE BOUNDS-------------------------%
%----------------------FINTIE BOUNDS-------------------------%
%----------------------FINTIE BOUNDS-------------------------%

%%% section on limit for finite t

\section{Bounds for finite $t$} \label{sec::finiteT}

In this section we prove a number of basic results, which will lead to the proof of Theorem~\ref{thm::finiteResult}. We will first show how two basic operations, copying and merging, can be used to combine existing designs into a new design with certain properties. Using these operations, we then proceed to prove the claims in \Cref{prop:basic}. 

Using the convexity results from \Cref{prop:basic}, we show achievability for \Cref{thm::finiteResult}. Following the achievability, we show the converse for \Cref{thm::finiteResult} which uses a technique we call covering. 

A similar result for the achievable region for ternary alphabet is stated at the end of the section. 

\subsection{Two basic operations on designs}

\begin{definition}[Copying]\label{def::copying}
\textcolor{black}{A disjoint union, which we call copying, of two designs $G_1$ and $G_2$ is denoted by $G_1+G_2$. A disjoint union of a collection of designs $G_j$ is denoted by $\sum_j G_j$. A disjoint union of $n$ identical designs $G$ is denoted as $nG$. }
\end{definition}

\textcolor{black}{
Copying is simply the idea of creating a new design from two designs where each design is a disjoint component of the new design. Defining this operation formally is useful for our analysis. See \Cref{fig::copy} for an illustration of copying.}

\begin{figure}
\centering
\includegraphics[scale = .15]{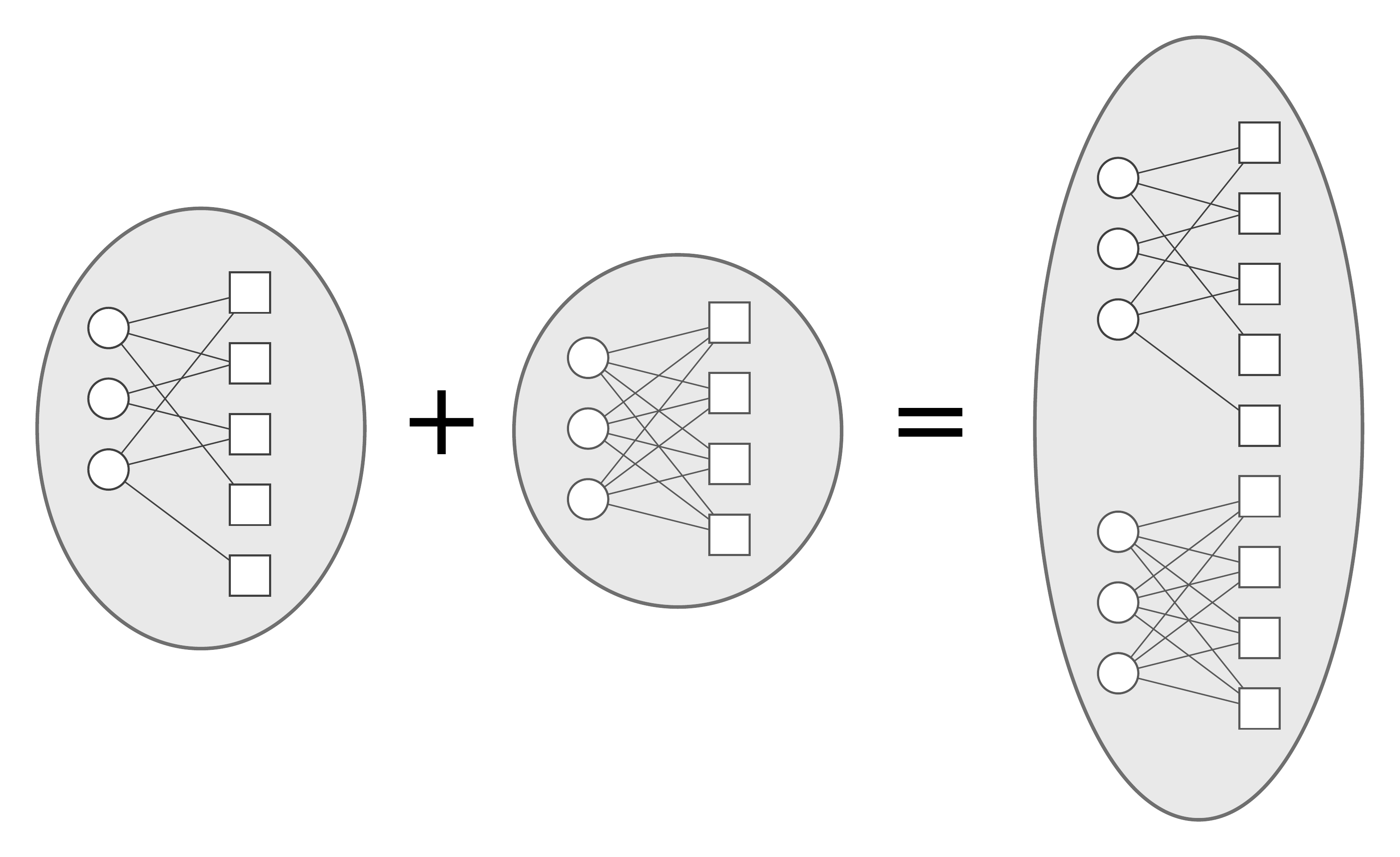}
\caption{\label{fig::copy} \textcolor{black}{Example of the copying operation on two designs.}}
\end{figure}

\begin{proposition}[Copying]\label{prop:copying} Consider $(k_j, m_j, t, E_j)_q$-designs $G_j$. Then $\sum_j G_j$,
forms a $(\sum k_j, \sum m_j, t, \sum E_j)_q$-design.
\end{proposition}
\textcolor{black}{The proof is clear after realizing that the number of defects corrected does not change while all other parameters must add. }We note here that the values of $\varepsilon$ and $\rho$ for $G_1+G_2$ is a convex combination of those of
$G_j$. That is  
\begin{align} 
\rho &= {k_1\over k_1+k_2}\rho_1 + \frac{k_2}{k_1 + k_2}\rho_2 \label{eq::copy1}\\ 
 \varepsilon &= \frac{k_1}{k_1 + k_2} \varepsilon_1 + \frac{k_2}{k_1 + k_2}\varepsilon_2\label{eq::copy2}
\end{align} 
where $\rho_j$ and $\varepsilon_j$ refer to $\frac{m_j}{k_jt}$ and $\frac{E_j}{k_jt}$ of $G_j$ respectively.

\begin{proposition}[Merging]\label{prop:merge} Consider $(k, m_j, t_j, E_j)_q$-designs $G_j$ and $G=\bigvee_j G_j$ (see \Cref{def::merging}). Then $G$ is
a $(k, \sum_j m_j, \sum_j t_j, \sum_j E_j)_q$-design.
\end{proposition}
(Note that it is possible that the merged design $G = \bigvee_j G_j$ can correct more than $\sum_j t_j$ defects.) \textcolor{black}{The proof is clear after realizing that the same labeling $G_j$ used for redundant nodes to be $t_j$ correcting for a specific labeling of the $k$ primary nodes can be used in the merged design. } 
As an example, we note that merging a design with itself, i.e., $G\vee G$, doubles all the parameters except $k$. However, the wiring complexity and redundancy stays constant. This will be the basis for showing
convexity of $\matr_\infty$, cf.~\eqref{eq:rinfty2}.

\subsection{Proof of Proposition~\ref{prop:basic}}\label{sec:proofbasic}

With the help of the two basic operations, we can prove the convexity of $\matr_t$ and $\matr_\infty$, as well as the other
properties claimed in Proposition~\ref{prop:basic}.

\begin{proof}[Proof of Proposition~\ref{prop:basic}]

\emph{Claim~\ref{prop:basic:1}.} From the definition of closure, $(\varepsilon, \rho) \in \matr_t$ if and only if there is a sequence of points  $\{(\varepsilon_i, \rho_i)\}_i \in \matr_t$ approaching $(\varepsilon, \rho)$. Each $(\varepsilon_i, \rho_i)$ must be associated with a $(k_i, m_i, t, E_i)_q$-design $G_i$, where $m_i = \rho_i k_it$ and $E_i = \varepsilon_i k_it$. To show that $k,m,E \rightarrow \infty$, we can copy $G_i$ with itself $n_i$ times, where $n_i$ is chosen so that $n_ik_i, n_im_i, n_iE_i \rightarrow \infty$.

\emph{Claim~\ref{prop:basic:2}.}
\textcolor{black}{For any $(\varepsilon, \rho) \in \mathcal{R}_t$, if there is a $(k, \rho kt, t, \varepsilon kt)$-design $G$, then we can copy $G$ with itself multiple times to get a $(k', \rho k't, t, \varepsilon k't)$-design $G'$ where $k'$ is arbitrarily large. }We can always add more redundant nodes or more edges to $G'$ (this is possible since $\rho k't$ can be arbitrarily large and adding a finite number of redundant nodes does not change the redundancy) to $G'$ to get a design with parameters $(\varepsilon', \rho')$. \textcolor{black}{If $(\varepsilon, \rho)$ is a limit point achieved by a sequence of designs, we can always similarly add more redundant nodes and edges to each design in the sequence that attains the limit. }

\emph{Claim~\ref{prop:basic:3}.}
This holds using copying from \Cref{prop:copying}.

If a pair of values $(\varepsilon_1, \rho_1)$ and $(\varepsilon_2, \rho_2)$ are in $\mathcal{R}_t$, there are sequences $(\varepsilon_{1,i}, \rho_{1,i}) \rightarrow (\varepsilon_{1}, \rho_{1}) $ and  $(\varepsilon_{2,i}, \rho_{2,i}) \rightarrow (\varepsilon_{2}, \rho_{2}) $, where for each $i$ there exists a $(k_{1,i}, \rho_{1,i} k_{1,i}t, t, \varepsilon_{1,i} k_{1,i}t)_q$-design $G_{1,i}$ and a $(k_{2,i,} \rho_{2,i} k_{2,i}t, t, \varepsilon_{2,i} k_{2,i}t)_q$-design $G_{2,i}$. For any $0 \leq \alpha \leq 1$, we can find a sequence of rational numbers $\alpha_i = \frac{p_i}{q_i}$ where $p_i, q_i \in \mathbb{Z}_{+}$ and $\alpha_i \rightarrow \alpha$. The copy $k_{2,i} p_i G_{1,i} + k_{1,i}(q_i - p_i) G_{2,i}$ achieves the point $(\varepsilon_i, \rho_i) = (\alpha_i \varepsilon_{1,i} + (1-\alpha_i) \varepsilon_{2,i}, \alpha_i \rho_{1,i} + (1-\alpha_i) \rho_{2,i})$ in $\matr_t$ and $(\varepsilon_i, \rho_i) \to (\alpha \varepsilon_1 + (1-\alpha) \varepsilon_2, \alpha \rho_1 + (1-\alpha) \rho_2 )$.

\emph{Claim~\ref{prop:basic:4}.}
Any point $({\varepsilon},{\rho})$ in $\matr_\infty$ and any point in  $\mathrm{closure} \left\{ \bigcup _{t = 1}^{\infty} \mathcal{R}_t \right\}$ must both be the limit of some sequence of $(k_i, m_i, t_i, E_i)_q$-designs. 
To see that $\mathcal{R}_\infty = \limsup \matr_t$, by merging in \Cref{prop:merge}, for any $t$, we have $\mathcal{R}_t \subset\mathcal{R}_{2t} \subset \mathcal{R}_{4t} \subset \mathcal{R}_{8t}\dots$.

\emph{Claim~\ref{prop:basic:5}.} This holds using merging from \Cref{prop:merge}. Given two designs $G_1$ and $G_2$, where $G_1$ is a $(k_1,{\rho_1} k_1 t_1 , t_1,{\varepsilon_1} k_1 t_1 )_q$-design and $G_2$ is a $(k_2, {\rho_2}k_2 t_2 , t_2, {\varepsilon_2}k_2 t_2 )_q$-design, if we want to create a design $G$ with the parameter $(\alpha {\varepsilon}_1 + (1-\alpha) {\varepsilon}_2, \alpha {\rho}_1 + (1-\alpha) {\rho}_2)$ for $\alpha = \frac{p}{q}$ where $p,q \in \mathbb{Z}_+$, then we can let 
\begin{equation} \label{eq::mergingConvex}
G = pk_2\left(\bigvee_{i = 1}^{t_2} G_1 \right) + (q-p)k_1\left(\bigvee_{i = 1}^{t_1} G_2 \right) \,.
\end{equation}
From here on, the proof proceeds similarly to the proof of Claim~\ref{prop:basic:3}. 
\end{proof}

\subsection{Elementary achievability}

\textcolor{black}{From the previous propositions, we can immediately make statements on what each region $\matr_t$ must contain. Recall that for any $t$, the point $(1,1)$ in $\matr_t$ is achievable using the repetition design. The point $(q, 0)$ is asymptotically achievable using the complete design. Thus, the line of points between $(1,1)$ and $(q, 0)$ is achievable by interpolating between the repetition
design $K(1,t)$ and the complete design $K(k, qt)$. We summarize this below:}

\begin{proposition}\label{prop:simpach}
The following region is achievable for any $t \ge 1$ and $q \geq 2$:
\begin{equation}\label{eq:interp}
\matr_t^{(K)} \eqdef \{(\varepsilon,\rho): \varepsilon \geq q  + (1 - q)\rho, \varepsilon \geq
1, \rho \ge 0\}\,.
\end{equation}
Furthermore, every point such that $(\varepsilon-1)$ is a multiple of $(q-1)$ can be achieved via a design with
constant degree $\varepsilon$ primary nodes. 
\end{proposition}
\begin{proof}
The corner points $(1,1)$ and $(q,0)$ are achieved by the repetition design and the complete design,
respectively. By Proposition~\ref{prop:basic} the region $\matr_t$ is convex and hence must contain $\matr_t^{(K)}$. \textcolor{black}{All rational points near the boundary of $\matr_t$ are achieved by $r_1 K(1,t)+r_2 K(k, q t)$ for some choice of integers $r_1,r_2$ and $k$. }

In order to get a design where the primary nodes have regular degree, we can combine the repetition design and complete design by merging. Find two integers $t_1,t_2$ where $t_1 + t_2 = t$. The combination $k K(1,t_1) \vee K(k, qt_2)$ also achieves the boundary point at $\varepsilon = (t_1 + qt_2)/t$ as $k \rightarrow \infty$. This proves the last sentence of the \Cref{prop:simpach}.
\end{proof}

\textcolor{black}{The region $\matr_t^{(K)}$ is an inner bound on all achievable regions, but for $q=2$ and $t=1,2$ the region $\matr_t^{(K)}$ happens to be tight and is the region plotted in Figure \ref{fig::AchRegion}.\footnote{\textcolor{black}{Note that in the worst case, the rate of convergence to get $\epsilon$ close to a point on the boundary of $\matr_t^{(K)}$ requires $k$ to be on the order of $\frac{1}{\epsilon}$. This occurs when trying to achieve the boundary point $(qt,0)$. On the other hand, achieving point the boundary point $(1,1)$ can be done with $k=1$. For other points on the boundary away from $(qt,0)$, it is not clear what the best rate of convergence is.}}}

\subsection{Covering converse}\label{sec:covering}

\textcolor{black}{This section presents a general converse bound which holds for all $\matr_t$ and all $q$, but in particular this converse shows that $\matr_2^{(K)}$ is tight for $q = 2$.}

\begin{theorem}\label{thm:covering} Fix $q=|\matx|$, $t$ and suppose $(\varepsilon,\rho)\in\matr_t$. Then there exists
$\pi_t, \pi_{t+1},\ldots,\pi_{qt} \ge 0$ satisfying 
\begin{align}\label{eq:cov0}
   \frac{1}{t}\sum_{j=t}^{qt} j \pi_j &\le \varepsilon\\
   \sum_{j=t}^{qt} \pi_j &= 1\label{eq:cov1}\\
   \sum_{j=t+1}^{qt} \pi_j \log_q \lfloor j/t \rfloor &\ge 1 + (t-1)\pi_t - \rho t\,. \label{eq:cov2} 
\end{align}   
In other words the smallest achievable $\varepsilon$ for a given $\rho$ is lower bounded as
\begin{equation}\label{eq:covering}
	\varepsilon^*(\rho, t) \ge \min \left\{ \frac{1}{t} \sum_{j=t}^{qt} j \pi_j: \pi_j\ge0\mbox{~satisfy~\eqref{eq:cov1}-\eqref{eq:cov2}}
\right\} 
\end{equation}
\end{theorem}
\begin{proof} 

%By Prop.~\ref{prop:basic} there exists a sequence of designs achieving $(\varepsilon,\rho)$ with $k,m\to\infty$. Given any design $G$ define a new design $G'$ with the same number of primary nodes $k$ and the number of redundant nodes as $m'=m+qt$. The added $qt$ redundant nodes are connected to each of the primary nodes that have degree (in $G$) larger or equal to $qt$. The remaining primary nodes are connected in $G'$ exactly as in $G$. It is clear that $G$ is still $t$-defect correcting, has the same (or smaller) number of edges and (asymptotically in $m$) the same redundancy $\rho$. This shows, that without loss of generality we can assume that \textit{there are no primary all nodes of degree $> qt$  and all nodes of degree equal to $qt$ form a complete bipartite graph disjoint from the rest of the design.} 

\textcolor{black}{The key idea of this proof is to look at how the degree of primary nodes relates to whether a design can correct defects for all sequences of labelings.} Let us define $\pi_j, j = t, t+ 1,...,qt -1$ to be the fraction of primary nodes with degree $j$. (Notice that every primary node clearly should have degree at least $t$.) Define $\pi_{qt}$ to be the fraction of primary nodes of degree $qt$ or larger. 
The fact that this satisfies~\eqref{eq:cov0}-\eqref{eq:cov1} is obvious. We only need to
show~\eqref{eq:cov2}. 

To that end, for each labeling $r^m \in \matx^m$ of redundant nodes let $\matg_t(r^m)$ be the set of primary node labelings for which conditions of Definition~\ref{def:main} are satisfied (we say that
$r^m$ covers $\matg_t(r^m)$ of the labelings). \textcolor{black}{The design is $t$-defect correcting if and only if every possible labeling is covered by some $r^m$. We can count the number of primary node labelings covered by some $r^m$ and make sure this is equivalent to all possible primary node labelings. Thus, a design is $t$-defect correcting if and only if }
\begin{equation}\label{eq:cov3}
	\left|\bigcup_{r^m \in \matx^m} \matg_t(r^m)\right| = |\matx|^k = q^k\,.
\end{equation}
We are aiming to apply the union bound to the right-hand side to get inequality~\eqref{eq:cov2}. Before doing so we make
the following observation.

%First, note that by the earlier remark the $k\pi_q$ primary nodes of degree $q$ are completely separate from the rest of the design and thus we only need to verify~\eqref{eq:cov3} on a sub-design with $k_1 = k(1-\pi_{qt})$ primary nodes and $m_1=m-qt$ redundant nodes. We substitute $k,m$ with $k_1,m_1$ from now on.

Two primary nodes of degree $t$ should have disjoint neighborhoods (otherwise labeling them different values clearly
violates Definition~\ref{def:main}). Thus $\matg_t(r^m)$ is empty unless each such neighborhood has a constant label. \textcolor{black}{This shows that for the $tk \pi_t$ redundant nodes connected to the primary nodes of degree $t$}, we are restricted to
only $q^{k\pi_t}$ choices, while the rest contribute $q^{m-tk\pi_t}$ more choices.

Given any of the $q^{m-(t-1)k\pi_t}$ choices of $r^m$ we can estimate $|\matg_t(r^m)|$ from above by assuming that each
primary node of degree $d$ can take any of the $\lfloor d/t \rfloor$ label in $\matx$ while still satisfying the
$t$-wise coverage condition of Definition~\ref{def:main}. This yields
\begin{equation}
 |\matg_t(r^m)| \le \prod_{j=t}^{qt} \lfloor j/t \rfloor ^{k\pi_j}\,, 
\end{equation}
and thus applying the union bound to~\eqref{eq:cov3}, %taking $m\to\infty$ 
we get~\eqref{eq:cov2}.
\end{proof}

For $t = 1,2$ and $q = 2$, it is only necessary to evaluate \eqref{eq:covering} at three separate points (two of which are $\varepsilon = 1$ and $2$, the third is anywhere inbetween) in order to show that the boundary of $\matr_1$ or $\matr_2$ from $\varepsilon = 1$ to $2$ is linear. In particular, for $t = 2$, we can first choose $\varepsilon = 3/2$. No matter how we choose the values of $\pi_2, \pi_3$ and $\pi_4$, to satisfy \eqref{eq:cov2} we must have $\rho \geq 1/2$. 

\begin{proof}[Proof of Theorem~\ref{thm::finiteResult}]
Achievability follows from \Cref{prop:simpach}. The converse is determined by evaluating~\ref{eq:covering}.  %The result simplifies to $\varepsilon \geq 2 - \rho$.
\end{proof}

\begin{remark}  While the bound~\eqref{eq:covering} is tight for $t=1$ and $t=2$ when $q = 2$, it is not tight in general. It however
allows us to make a general conclusion: since the bound is piecewise linear, it
follows that the slope of $\matr_t$ at the point $(qt, 0)$ of minimal redundancy \textit{is non-zero}. It is also the
best bound known to us for values of $\varepsilon$ near $qt$. 
\end{remark}

In the next section, we will discuss a bound that
is better for $\varepsilon$ away from $q$ and when $t$ is large. This converse outperforms the covering converse (\Cref{thm:covering}) at certain $\rho$ even for $q = 2$ and $t = 3$.
\subsection{Ternary alphabet and $t=1$}
Further progress on computing regions $\matr_t$ for values of $q > 2$ seems to require finer arguments
on graph structure. We can show the following result for $q = 3$ but the proof requires significant casework.

\begin{theorem} For $q=3$ and $t=1$ we have
\begin{equation}
\matr_1 = \{(\varepsilon,\rho): \varepsilon \ge 3-2\rho, \varepsilon \geq 1, \rho\ge 0\} \label{eq::3alphaRegion}
\end{equation}
and is achievable by the interpolation~\eqref{eq:interp}. \label{thm:3alpha}
\end{theorem}

We give the proof in \Cref{apx:3alpha}.

%----------------------ASYMPTOTIC-------------------------%
%----------------------ASYMPTOTIC-------------------------%
%----------------------ASYMPTOTIC-------------------------%
%----------------------ASYMPTOTIC-------------------------%

\section{Fundamental limit for $t\to \infty$} \label{sec::asymSection}

Recall that as $t\to \infty$ the fundamental limit $\matr_\infty \eqdef \limsup  \matr_t$ can be characterized as the set of wiring complexity-redundancy pairs, namely
\begin{equation}
 \varepsilon = {E\over kt}, \quad \rho = {m\over kt}
\end{equation}
over all values of $t$ (see Proposition~\ref{prop:basic}.)
The goal of this section is to prove the following result, that generalizes the binary version stated earlier in 
Theorem~\ref{thm::asymResult}.

\begin{theorem}
\label{thm::asymResultLarger}
Fix alphabet $|\matx|=q$. The region $\matr_\infty$ defined in~\eqref{eq:rinfty} is the closure of the set of points
$(\varepsilon,\rho)$, parameterized by the distribution $P_S$ on a finite support of $\mathbb{Z}_+$, and
\begin{align}\label{eq:asymResultLarger}
	 \varepsilon &= {\EE[S]\over F(P_S)}, \quad  \rho = {1\over F(P_S)}\,,\\
F(P_S) &\eqdef \min_{P_X} \max_{P_{Y|\vect L}} \min_{j\in[q]}{1 \over P_X(j)} \EE[L_j \mathbbm{1}\{Y=j\}]
\label{eq::optfunc2}
\end{align}
where $\EE[\cdot]$ is computed over random variables $S\in \mathbb{Z}_{+}, X\in[q], \vect L=(L_1,\ldots,L_q)\in\{0 \cup \mathbb{Z}_{+}\}^q, Y\in[q]$ with joint distribution
\begin{equation}\label{eq:jdist}
	P_{S,\vect L,Y}(s,\vect{\ell},y) \eqdef P_S(s) P_{\vect{L}|S} (\vect \ell | s) P_{Y|\vect L}(y|\vect \ell)\,. 
\end{equation}
where\footnote{$P_{\vect L|S}$ is the multinomial distribution, $\text{Mult}(s,[P_X(1), \cdots, P_X(q)])$} 
	\begin{equation}\label{eq:multdef}
		P_{\vect L|S}(\vect \ell|s) \eqdef \binom{s}{\ell_1,\cdots,\ell_q} \prod_{j = 1}^q P_X(j)^{\ell_j}\,.
\end{equation}	
\end{theorem}
%The function $P_{Y|\vect L}$ can be understood as proportion of redundant nodes whose neighbors have type $L$ that are assigned to a label $y$. The function $P_X$ can be understood as the proportion of primary nodes given each label.
\Cref{thm::asymResultLarger} gives \Cref{thm::asymResult} by substituting $P_{Y|\vect L}(0|(L_0, L_1))$ with $f(L_0,L_1)$, $P_X(0) $ with $\lambda$, and $P_X(1)$ with $1 - \lambda$. Also, the multinomial distribution is replaced by the binomial distribution. 

We start the section by proving relevant properties of $F(P_S)$. We then use these properties to prove the achievability  (i.e., upper bound) of \Cref{thm::asymResultLarger}. (This achievability proof explains why the quantities used in \Cref{thm::asymResultLarger} are important.) Next, we present a symmetrization property which is the key idea of the converse argument of \Cref{thm::asymResultLarger}. Putting these elements together gives the complete proof. 

Following the proof, we present a number of observations about \Cref{thm::asymResultLarger}. These include a section about how we compute \Cref{thm::asymResultLarger} numerically and a section on the achievable region for designs where $k$ is finite, but $t$ and $m$ are allowed to go to infinity. This result follows from the proof of \Cref{thm::asymResultLarger}. We also discuss how the Hamming block is optimal in this context.

\subsection{Auxiliary results about $F(P_S)$}

Before proceeding further, we need to describe some technical properties of $F(P_S)$ and related quantities.\footnote{The notation $\frac{1}{k} \mathbb{Z}$ refers to the set of fractions with denominator $k$.}

\begin{definition}[Finitary $F$]
We define $F_{k,n}(P_S)$ and $F_k(P_S)$ as follows:
\begin{align*} F_{k,n}(P_S) \eqdef& \min_{P_X \in {1\over k}\ZZ} \,\,\,\max_{P_{Y|\vect L^{(k)}} \in {1\over n}\ZZ}\,\,\, \min_{j\in[q]}\,\,\\& \quad \quad 
{1 \over P_X(j)} \EE[L_j^{(k)} \mathbbm{1}\{Y=j\}]\,, \numberthis \label{eq:fkndef}\\
 F_{k}(P_S) \eqdef& \min_{P_X \in {1\over k}\ZZ}\,\,\, \max_{P_{Y|\vect L^{(k)}}}\,\, \min_{j\in[q]} \,\,\, \\
 & \quad \quad{1 \over P_X(j)} \EE[L_j^{(k)}
 \mathbbm{1}\{Y=j\}]\,, \numberthis \label{eq:fkdef}
\end{align*} 
where $\EE[\cdot]$ is computed over random variables $S\in [k], X\in[q], \vect L^{(k)}=(L_1,\ldots,L_q)\in\{0 \cup \mathbb{Z}_+\}^q, Y\in[q]$ with joint distribution
\begin{equation}\label{eq:jdist_k}
	P_{S,\vect L^{(k)},Y}(s,\vect{\ell},y) \eqdef P_S(s) P_{\vect{L}^{(k)}|S} (\vect \ell | s) P_{Y|\vect L^{(k)}}(y|\vect \ell)\,. 
\end{equation}
where\footnote{$P_{\vect L^{(k)}|S}$ is the multivariate hypergeometric distribution, $\text{HyperGeom}(s,k,[P_X(1), \cdots, P_X(q)])$} 
\begin{equation}\label{eq:hyperdef}
	P_{\vect L^{(k)}|S}(\vect \ell|s) \eqdef {{kP_X(1) \choose \ell_1} \cdots {kP_X(j) \choose \ell_j } \cdots  {kP_X(q)\choose \ell_q} \over {k \choose s}}\,.  
\end{equation}
\end{definition}

Note that the definition of $F_{k,n}$ is similar to that of $F(P_S)$, see~\eqref{eq::optfunc2}, but with two changes: 1) values of $P_X$ and $P_{Y|\vect L^{(k)}}$ (instead of $P_{Y|\vect L}$)
are required to be integer multiples of ${1\over k}$ and $1\over n$, respectively; and b) $P_{\vect L^{(k)}|S}$ is (multivariate) hypergeometric, instead of multinomial. \textcolor{black}{The function $F(P_S)$ which we are ultimately interested in for this section is the limit of $F_k(P_S)$ for $k \to \infty$, which itself is a limit of $F_{k,n}(P_S)$. The function $F_k(P_S)$ is an important quantity which bounds the rate region for designs with finite $k$, which we will discuss later in~\ref{sec::finitek}.}

\begin{proposition} \label{prop::Fnkprops} For any $P_S$ with finite expectation we have
	\begin{equation}\label{eq:fnk_bounds}
		F_k(P_S) -  \frac{\mathbb{E}[S]}{n} \leq F_{k,n}(P_S) \leq F_{k}(P_S)\,.
\end{equation}
Also, there exists a sequence $\epsilon_{k}\to 0$ such
that for any $P_S$ on $\ZZ_+$ with finite third moment we have 
\begin{equation}\label{eq:fk_bounds}
		|F_{k}(P_S) - F(P_S)| \le {\EE[S^3]\over 2k} + \epsilon_{k}\,.
\end{equation}
\end{proposition}

See \Cref{sec::fnkpropsProof} for proofs.

\subsection{Subset design achievability and upper bound}
The next proposition gives bounds on the performance of subset designs in terms of $F_{k,n}(P_S)$ and $F_{k}(P_S)$.\footnote{This proposition initially used random coding as an argument. Random coding has since been replaced.}
% upper bound on t for single subset
\begin{proposition}\label{prop::subsetFiniteKResult}
Let $q = |\matx|$ and fix $k \in \mathbb{Z}$. Let $G = \bigvee_{i = 1}^n G'$, where $G'$ is a subset design with $P_S(s)$ as the proportion of redundant nodes with degree $s$ for $s \in [k]$.  If $G$ is a $(k, m, t, E)_q$-design, where $E = m\mathbb{E}[S]$ and $t$ is the maximum number of defects $G$ can correct, then 
\begin{equation}
\frac{m}{k} F_{k,n}(P_S)  \leq  t \leq \frac{m}{k} F_k(P_S)\,.
\end{equation}
\end{proposition}
\begin{proof}
First we show the upper bound that $t \leq \frac{m}{k} F_k(P_S)$.

Consider any labeling $w^k \in \matx ^k$ of the $k$ primary nodes of $G$. Let the frequency which each label occurs in the labeling have empirical distribution $P_X$ (that is, if $k_i$ of the $k$ primary nodes have label $i$, then $P_{X}(i) = \frac{k_i}{k}$). \textcolor{black}{Given this labeling, we define the \emph{type} of each redundant node $v$ to be $\vect \ell = (\ell_1, \cdots,\ell_q)$, where $\ell_j$ is the number of primary nodes with label $j$ which is a neighbor of redundant node $v$. }(If the degree of the redundant node is $s$, then $\sum_{i = 1}^q \ell_i = s$.) Because $G$ is a subset design, the proportion of degree $s$ redundant nodes in $G$ with type $\vect \ell$ is 
\begin{equation}
P_{\vect L^{(k)}|S}(\vect \ell|s) = {{kP_X(1) \choose \ell_1} \cdots {kP_X(q)\choose \ell_q} \over {k \choose s}}\,.
\end{equation}

Now, for any choice of labeling $r^m \in \matx^m$ of the $m$ redundant nodes, let $P_{Y|\vect L^{(k)}}(j|\vect \ell)$ represent the proportion (empirical distribution) of redundant nodes of type $\vect \ell$ which are labeled $j$. 
For each label $j$, we can count the \emph{average} number of matching redundant node neighbors a primary node $u$ with label $j$ has by summing up all the edges between primary and redundant nodes both with label $j$, and then dividing this by the total number of primary nodes with label $j$. This average is 
\begin{align}
\tilde{t}(j) &\eqdef \frac{1}{kP_X(j)} \sum_{s} m P_S(s) \sum_{\vect \ell}P_{\vect L^{(k)}|S}(\vect \ell|s) \ell_j P_{Y|\vect L^{(k)}}(j|\vect \ell) \\
&= \frac{m}{k}{1 \over P_X(j)} \EE[L_j \mathbbm{1}\{Y=j\}]\,.
\end{align}

The label $j$ where this average is lowest determines the upper bound on the number of defects $G$ with labeling $w^k$ and $r^m$ can correct. This upper bound is given by $\min_j \tilde{t}(j)$. We have the freedom to pick the redundant node labeling $r^m$ with the empirical distribution $P_{Y|\vect L^{(k)}}$ which maximizes the average. The defect correcting number needs to hold for all possible $w^k$, so the empirical distribution $P_X$ which gives the lowest value of $ \max_{P_{Y|\vect L^{(k)}}} \min_j \tilde{t}(j)$ determines $t$. This gives the upper bound on $t$.

We now show the lower bound $\frac{m}{k} F_{k,n}(P_S)  \leq  t $.

Given any labeling $w^k \in \matx^k$ of the primary nodes with empirical distribution $P_X$, let 
\begin{equation}
 P_{Y_n|\vect L^{(k)}} = \argmax_{P_{Y|\vect L^{(k)}} \in \frac{1}{n}\mathbb{Z}} \,\, \min_{j\in[q]}{1 \over P_X(j)} \EE[L_j \mathbbm{1}\{Y=j\}] \,.
\end{equation}

For each $\vect \ell$, $P_{Y_n|\vect L^{(k)}}(j|\vect \ell) = \frac{c_j}{n}$ for some $c_j \in \mathbb{Z}_{+} \cup 0$, and $\sum_j c_j = n$.
Because $G$ is a merging of $n$ copies of $G'$, we can partition the copies of $G'$ in $G$ into sets of size $c_1, ..., c_q$. The $j$th set is a set of $c_j$ copies of $G'$. Label all redundant nodes of type $\vect \ell$ in the $j$th set the value $j$. 
We can determine that each primary node $u$ with label $j$ has a total of $P_S(s) m \frac{\ell_j}{P_X(j)k} P_{\vect L^{(k)}|S}(\vect \ell |s)$ redundant nodes of type $\vect \ell$ in its neighborhood. This redundant node labeling scheme assigns exactly $ P_{Y_n|\vect L^{(k)}}(j|\vect \ell)$ of these neighbors the label $j$.

Repeat this labeling process for each redundant node type $\vect \ell$. Summing over all $\vect \ell$ and all $s$ will get that the total number of redundant nodes with label $j$ in the neighborhood of primary node $u$ is $ \sum_{s,\vect\ell}P_S(s) m \frac{\ell_j}{P_X(j)k} P_{\vect L^{(k)}|S}(\vect \ell |s) P_{Y_n|\vect L^{(k)}}(j|\vect \ell)$.

Using this scheme, $G$ can correct at least 
\begin{align}
t &\geq \min_{P_X \in \frac{1}{k} \mathbb{Z}} \min_{j\in[q]} \sum_{s,\vect\ell}P_S(s) m \frac{\ell_j}{P_X(j)k} P_{\vect L^{(k)}|S}(\vect \ell |s) P_{Y_n|\vect L^{(k)}}(j|\vect \ell)\\
& = \frac{m}{k} \min_{P_X \in \frac{1}{k}  \mathbb{Z}} \,\, \max_{P_{Y|\vect L^{(k)}} \in \frac{1}{n}\mathbb{Z}} \,\,\min_{j\in[q]}{1 \over P_X(j)} \EE[L_j \mathbbm{1}\{Y=j\}]\\
& = \frac{m}{k}F_{k,n}(P_S)
\end{align}
defects. 
\end{proof}

\subsection{Converse and proof of \Cref{thm::asymResultLarger}}

The converse needed to show \Cref{thm::asymResultLarger} is surprisingly simple. The main idea is the following:

\begin{proposition}[Symmetrization]\label{prop:sym}
If there exists a $(k,m,t,E)_q$-design then there exists a permutation-invariant $(k, m \cdot k!, t \cdot k!, E \cdot k!)_q$-design.
\end{proposition}

\begin{proof}
Let $G$ be a $(k,m,t,E)_q$-design. 
We will merge $G$ exactly $k!$ number of times. The key is that each copy will be merged by identifying with a permutation of the original primary nodes. 

Start with an ordering of the primary nodes in the design $G$. For each $\sigma \in S_k$ (the full symmetric group of $k$ elements), let $G_\sigma$ be isomorphic to the design $G$, with the order of its primary nodes transformed by $\sigma$. Then merge $G_\sigma$ for all $\sigma \in S_k$ identifying primary nodes in the same order. 

Let the result be 
\begin{equation}
G_\mathrm{PERM} = \bigvee_{\sigma \in S_k} G_\sigma \,.
\end{equation}
$G_\mathrm{PERM}$ is constructed to be permutation invariant. (For any redundant node $v$ in $G$, if $v$ has degree $s$, every set of $s$ nodes in $G_\mathrm{PERM}$ needs to be connected together by a copy of $v$.) By \Cref{prop:merge} $G_\mathrm{PERM}$ is a $(k, m \cdot k!, t \cdot k!, E \cdot k!)_q$-design.
\end{proof}

\textcolor{black}{See \Cref{fig::permMerge} for an example of merging permutations to obtain a subset design.} 

\begin{figure}
\centering
\includegraphics[scale = .12]{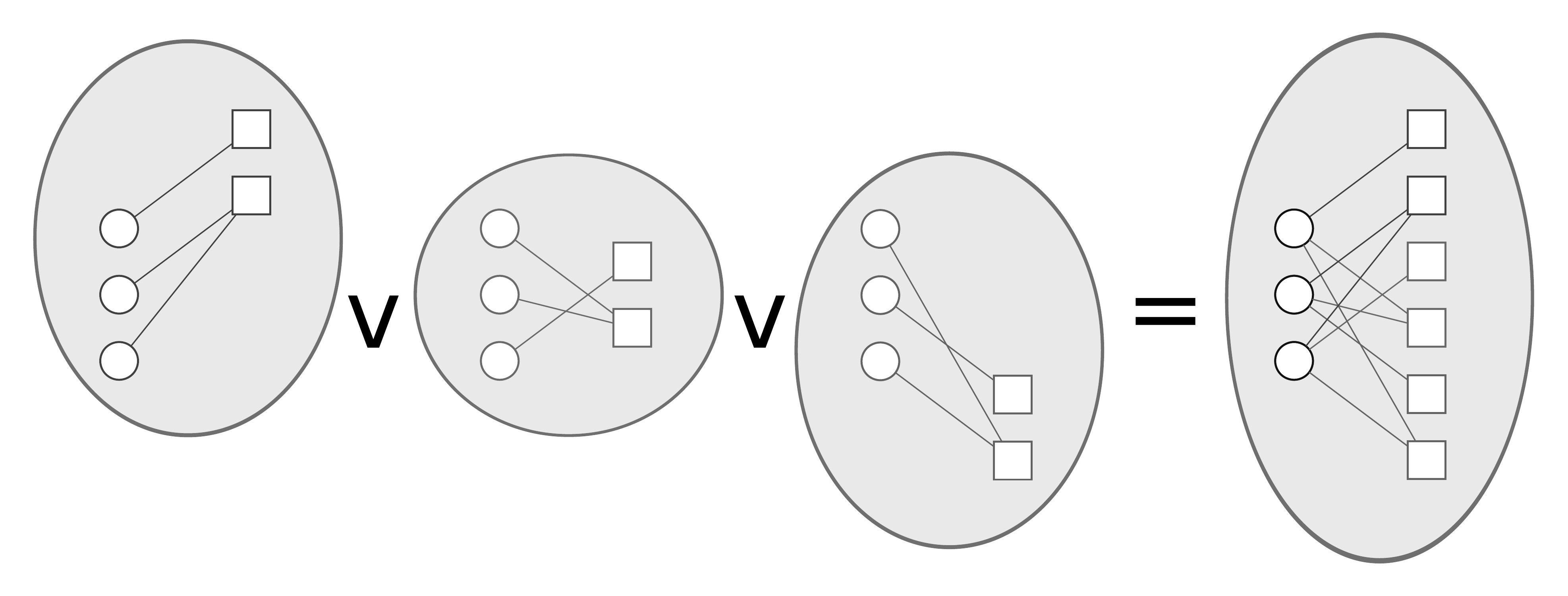}
\caption{\label{fig::permMerge}\textcolor{black}{ Example merging permutations of the same design. The resulting design is a subset design. In this example, for clarity, we did not show all $6$ possible permutations of the $3$ primary nodes in the original design. The $3$ distinct permutations shown was enough to create a subset design.}}
\end{figure}

In view of Propositions~\ref{prop:perminv} and~\ref{prop:sym}, we see that in terms of the values ${E\over
kt}, {m\over kt}$ every design on $k$ primary nodes is at most as good as a subset design on $k$ primary nodes \textcolor{black}{(meaning the pair of values an arbitrary design will achieve has the same or a worse trade-off than what a subset design can achieve)}. Performance of the latter is completely characterized by \Cref{prop::subsetFiniteKResult}. 
Now we can combine the results to prove \Cref{thm::asymResultLarger}.

\begin{proof}\emph{Proof of \Cref{thm::asymResultLarger}}

\emph{Achievability}: Fix $P_S \in \mathbb{Q}$ with finite support. For each $k$ and $n$, it is always possible to construct a subset design $G'$ on $k$ primary nodes where the proportion of redundant nodes of degree $s$ are given by $P_S(s)$. Let $G = \bigvee_{i = 1}^n G'$ so that by \Cref{prop::subsetFiniteKResult}, subset design $G$ is a $(k, m, t, E)_q$-design so that $\frac{tk}{m} \geq F_{k,n}(P_S)$ and $E = m \mathbb{E}[S]$. Since $F_{k,n}(P_S) \to F(P_S)$, there must exist a sequence of subset designs $G_i$ which are $(k_i, m_i, t_i, E_i)_q$-designs where 
\begin{align}
{\varepsilon} = \frac{E_i}{k_i t_i} \to \frac{\mathbb{E}[S]}{F(P_S)}\,, {\rho} = \frac{m_i}{k_i t_i} \to \frac{1}{F(P_S)} \,.
\end{align}
Thus 
\begin{equation} \label{eq::achPointInRInfty}
\left(\frac{\mathbb{E}[S]}{F(P_S)}, \frac{1}{F(P_S)}\right) \in \matr_\infty\,.
\end{equation}
Since $F(P_S)$ is continuous in $P_S$, \eqref{eq::achPointInRInfty} holds for any $P_S$ with finite support. 

\emph{Converse}: For any design $G$ which is a $(k, m, t, E)_q$-design, there exists a subset design $G'$ which is a $(k, m \cdot k!, t \cdot k!, E \cdot k!)_q$-design by \Cref{prop:sym}. Let $P_S$ be so that $P_S(s)$ represents the proportion of redundant nodes in $G'$ with degree $s$. Then $E = m \mathbb{E}[S]$. Let $t'$ be the number of defects $G'$ can correct. Using \Cref{prop::subsetFiniteKResult} and $F_k(P_S) \leq F(P_S)$ (cf. \Cref{lem::fklessf2k} in \Cref{sec::fkboundf}), 
\begin{equation}
t \cdot k! \leq t' \leq \frac{m \cdot k!}{k} F_k(P_S) \leq \frac{m \cdot k!}{k}F(P_S) \,.
\end{equation}

Then for design $G$, ${\varepsilon} = \frac{E}{tk}  \geq \frac{\mathbb{E}[S]}{F(P_S)}$ and ${\rho} = \frac{m}{tk}  \geq \frac{1}{F(P_S)}$. Thus, the limit of $(\frac{E_i}{k_it_i}, \frac{m_i}{k_it_i})$ for any sequence of $(k_i, m_i, t_i, E_i)_q$-designs must be in the closure of $(\frac{\mathbb{E}[S]}{F(P_S)}, \frac{1}{F(P_S)})$ for all $P_S$ with a finite support.
\end{proof}

\subsection{Observations about~\Cref{thm::asymResultLarger}} 

\subsubsection{Threshold solution} \label{sec:thresholdSolution}

The optimal value for $P_{Y|\vect L}$ tells us what the optimal labeling of redundant nodes should be. It turns out that for most values of $\vect \ell$, $P_{Y|\vect L}(j|\vect\ell)$ is either $0$ or $1$. 

For an illustration of this, consider the binary alphabet (or $q = 2$) case and the design $S(k, s)$. The types are $\vect \ell = (\ell_0, \ell_1)$. Given any empirical distribution  $P_X$ of the primary node labels, the optimal labeling of the redundant nodes must be so that redundant nodes with larger values of $\ell_1$ are assigned the label $1$ instead of redundant nodes with smaller values of $\ell_1$. Otherwise, we can always swap the labelings and increase the number of defects corrected. In fact, even when there are multiple subset sizes, it is possible to find an optimal solution where the value of $P_{Y|\vect L}$ depends only on the ratio of $\ell_0$ to $\ell_1$.  

\begin{proposition} \label{prop::splittingRatio}
For $\matx = \{0, 1\}$, the solutions $P_{Y|\vect L}$ which attain the maximum in \eqref{eq::optfunc2} must have the following form
\begin{align}\label{eq::splittingRatio}
P_{Y|\vect L}(0|\vect \ell) = \left \{ 
\begin{tabular}{cl} 
	1 &  if $\frac{\ell_0}{\ell_0 + \ell_1} > \gamma$\\
	0 &  if $\frac{\ell_0}{\ell_0 + \ell_1} < \gamma$ \\
	$\mu(\ell_0 + \ell_1)$ & if $\frac{\ell_0}{\ell_0 + \ell_1}= \gamma$  
	\end{tabular}
\right.
\end{align} 
where $\gamma \in [0,1]$ and $\mu(s) \in [0,1]$ for each $s \in \mathbb{Z}_+$. \footnote{There is not necessarily a unique solution for $\mu(s)$. One such solution has $\mu(s_a) = \mu(s_b)$ for all $s_a, s_b$.}
\end{proposition}

(See \Cref{sec::splittingRatioProof} for proof.)
Generalizing to larger alphabet sizes, the space of all possible $\vect \ell$ will be partitioned into $q$ pieces \textcolor{black}{depending on the relative ratios of values in $\vect \ell$}. The interior of each piece will have all types assigned to the same label, that is $P_{Y|L}(j|\vect \ell) = 1$ for some $j$.  The values of $\vect \ell$ on the boundary may be split between 2 or more values. 

Notice that in light of \Cref{thm::asymResultLarger} and \Cref{prop::splittingRatio}, computing the optimal redundant node labeling for subset designs given a fixed primary node labeling is easy. For general designs, this is NP-Hard. 

\subsubsection{Worst-case $P_X$}

The worst-case distribution of primary node labels which gives the result in~\Cref{thm::asymResultLarger} is not obvious, even in the binary alphabet case. When $\matx = \{0, 1\}$, we can easily determine that for subset designs $S(k, s)$ with even $s$, the worst-case $P_X$ is when $P_X(0) = P_X(1) = \frac{1}{2}$. However, when $s$ is odd, this is not true. When $s = 3$, the worst-case $P_X$ is determined by a solution to a cubic polynomial. When a merging of different subset designs are used or a larger alphabet is used, it is unclear how to find the worst-case $P_X$ analytically. This makes finding the worst-case $P_X$ the main difficulty in evaluating the optimization equation in~\Cref{thm::asymResultLarger} for given values of $P_S$. (The equation is non-convex in $P_X$.)

\subsection{Numerical upper and lower bounds}\label{sec::numSection}

Since the optimization presented in \Cref{thm::asymResult} is difficult to evaluate exactly, instead,
we give an approximation for the boundary by establishing computable almost tight upper and lower bounds for when $q = 2$. The details can be found in Appendix~\ref{sec::numerical} and a comparison is presented on \Cref{fig::maxProbConverse}. As can be seen, the gap between the bounds is on the order of $10^{-3}$ and virtually indistinguishable on the plot. The best known achievable point in $\matr_\infty$ for selected fixed values of $\EE[S]$ are given in~\Cref{fig:achNumbers}. These points are found by searching and using weights from the converse bound method presented in \Cref{sec::numerical}.  

\begin{figure}
\centering
\includegraphics[scale = .45]{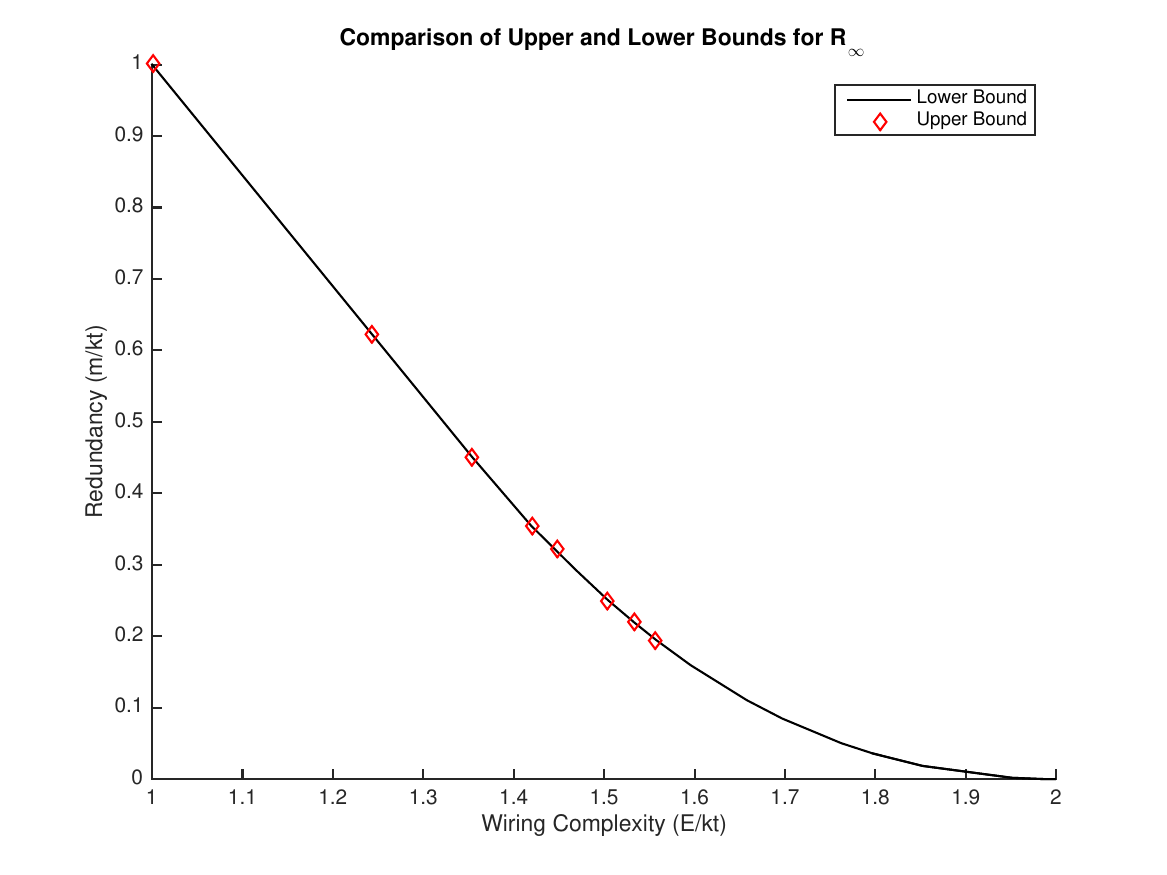}
\caption{Approximate converse bound compared with achievable points for $q = 2$. %The second plot zooms in one part of the plot (the fifth diamond marker from the left) to emphasize the difference between the converse and achievability.
\label{fig::maxProbConverse}}
\end{figure}

\begin{table}
\caption{Achievable points\label{fig:achNumbers}}
\centering
\begin{tabular}{|c||l|l|}\hline
$\EE[S]$ & Support set with corresponding $P_S$ & Point $({\varepsilon}, {\rho})$ in $\matr_\infty$ \\ \hline \hline
2 &  [1,3,4,5] , [0.62, 0.21, 0.10, 0.07] & (1.24, 0.61) \\ \hline
3 &  [1,3,4,5] , [0.24, 0.41, 0.20, 0.14] & (1.35, 0.45) \\ \hline
4 &  [3,4,5,6,7] , [0.52, 0.21, 0.13, 0.02, 0.12] & (1.42, 0.35) \\ \hline
5 &  [3,4,5,7] , [0.31, 0.23, 0.28, 0.18] & (1.40, 0.29) \\ \hline
6 &  [5,6,7,9] , [0.45, 0.31, 0.14, 0.10] & (1.47, 0.29) \\ \hline
7 &  [5,6,7,8,9] , [0.35, 0.01, 0.13, 0.32, 0.19] & (1.53, 0.22) \\ \hline
8 &  [7,8,9,11] , [0.40, 0.36, 0.16, 0.08] & (1.56, 0.19) \\ \hline
\end{tabular}

\end{table}

We observed the following effects about designs near the boundary of $\matr_\infty$ while experimenting with~\Cref{thm::asymResult}:
\begin{itemize}
\item The design has 4 or 5 distinct subset sizes
\item Odd number subset sizes are more common
\item The subset sizes which make up most of the design are consecutive, possibly skipping even subset sizes
\end{itemize} 

\subsection{Results for finite $k$}\label{sec::finitek}

To develop the proof for \Cref{thm::asymResultLarger}, we showed intermediate results on designs with $k$ primary nodes and observed what occurs when $k \to \infty$. We can use these intermediate results to determine the achievable regions for designs on $k$ primary nodes.

\begin{definition}
For fixed $q$ and $k \in \mathbb{Z}_+$, we define the region $\matr_\infty ^k$ as the closure of the set of all achievable pairs $ (\frac{E}{kt}, \frac{m}{kt}):$
\begin{equation}\label{eq::rinftyFiniteK}
 \matr_\infty ^k \eqdef \mathrm{closure} \left\{\left(\frac{E}{kt}, \frac{m}{kt}\right): \exists (k,m,t,E)_q\text{-design}\right\}\,.
\end{equation}
\end{definition}

Similar to regions $\matr_t$ and $\matr_\infty$, the region $\matr_{\infty}^k$ is convex. We can apply the proof for Claim~\ref{prop:basic:5} in \Cref{prop:basic} replacing the expression \eqref{eq::mergingConvex} with 
\begin{equation}
G = \left(\bigvee_{i = 1}^{pt_2} G_1 \right) \vee \left(\bigvee_{i = 1}^{(q-p)t_1} G_2 \right)
\end{equation}
to show this.

Claims~\ref{prop:basic:1} and Claim~\ref{prop:basic:2} of \Cref{prop:basic} also hold for $\matr_\infty^k$.

\begin{theorem} \label{thm::finiteAsymptoticResult}
Fix alphabet $|\matx|=q$. The region $\matr_\infty^k$ defined in \eqref{eq::rinftyFiniteK} is the closure of the set of points
$(\varepsilon,\rho)$, parameterized by the distribution $P_S$ on $[k]$, where
\begin{align}\label{eq:asymResultLarger_k}
	\varepsilon &= {\EE[S]\over F_k(P_S)}, \quad  \rho = {1\over F_k(P_S)}\,,
\end{align}
and $F_k(\cdot)$ is defined in~\eqref{eq:fkdef}.
\end{theorem} 

\begin{proof}
The achievability and converse of this theorem follows from \Cref{prop::subsetFiniteKResult} 
(with $F_k(P_S)$ is continuous in $P_S$) and \Cref{prop:sym} respectively.
\end{proof}

Using \eqref{eq:asymResultLarger_k}, we can plot the achievable region $\matr_\infty^3$ when $q = 2$ (see \Cref{fig::R3}).
The most salient aspect of $\matr_{\infty}^3$ is that the point achieved by the Hamming block (see \Cref{fig::hammingDesign1}) is a corner point of this region.  It is the only corner other than the usual corner point $(1, 1)$ achieved by the repetition design. 
\begin{figure}
\centering
\includegraphics[scale = .45]{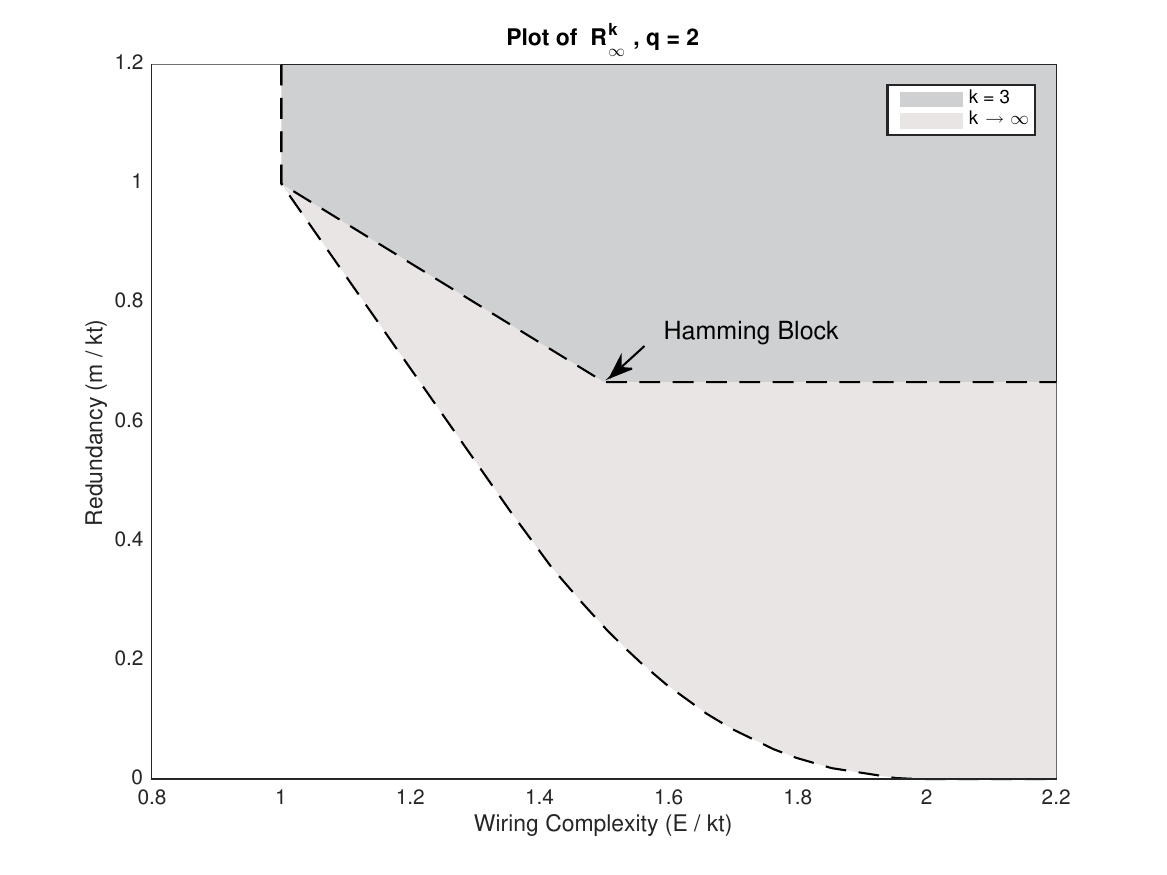}
\caption{\label{fig::R3} Region $\matr^3_\infty$ (see \eqref{eq::rinftyFiniteK}) compared with $\matr_\infty$ for $q = 2$. }
\end{figure}

\begin{corollary}[Hamming block corner point]\label{prop::hammingOptimal}
The value given by the Hamming block is a corner point of $\matr_\infty^3$ for $\matx = \{0, 1\}$.  
\end{corollary}

The proof for \Cref{prop::hammingOptimal} and the methods used to calculate $\matr_{\infty}^3$ are in \Cref{sec::HammingOptimalProof}. The implication of this result is that for any design on $k = 3$ primary nodes, no design has a better trade-off between redundancy $\frac{m}{kt}$ and wiring complexity $\frac{E}{kt}$ than the Hamming block, even if we allow the design to have arbitrarily many edges and redundant nodes.

%-------------------------DISCUSSION---------------------------%
%-------------------------DISCUSSION---------------------------%
%-------------------------DISCUSSION---------------------------%
%-------------------------DISCUSSION---------------------------%

\section{Discussion}\label{sec::discussion}

We conclude with a discussion of some implications of our results, some extensions and future work.

\subsection{Implications on practical designs}

The result for $\mathcal{R}_1$ and $\mathcal{R}_2$ (for $q = 2$) demonstrates that for correcting small defects, the best solution in the limit of a large number of primary nodes is a linear combination of two basic designs, the repetition design and the complete design. (Though this design is not optimal for finite $k$. Slight improvements can be made \textcolor{black}{by removing a few edges.})

\Cref{thm::asymResultLarger} gives a result for asymptotic $t$, and while practically no application is going to need to correct asymptotically many defects, the region defined by the result gives a converse bound for $\matr_t$ for all finite $t$ by virtue of Claim \ref{prop:basic:4} from \Cref{prop:basic}. All regions $\matr_t$ must lie between $\matr_1$ and $\matr_\infty$, approaching the latter as $ t \to \infty$. Hence \Cref{thm::asymResultLarger}  describes the fundamental limit for the trade-off between redundancy and wiring complexity. 

The numerical results for asymptotic $t$ and $q = 2$ imply that the designs which are close to optimal for large $t$ use redundant nodes with a limited set of degrees. The best achievable points found for $\matr_\infty$ for fixed values of $\EE[S]$ each use redundant nodes with degrees within $2$ or $3$ values of $\EE[S]$.

Results for $\matr_{\infty}^k$ define what is optimal for finite $k$ in terms of the number of defects correctable per use of redundancy and edges. We know exactly what this region looks like for $k = 3$ and can determine that the Hamming block is in fact the optimal design. Evaluating $\matr^k_\infty$ for larger values of $k$ gives exactly what trade-offs are realizable.

Also note that in the asymptotic $t$ results, the optimal trade-off is obtainable by designs which has regular primary node degree \textcolor{black}{(since subset designs are always regular). Not only that, but finding the best labeling of redundant nodes for subset desgins corresponds to finding $P_{Y|\vect{L}}$ in the statement of \Cref{thm::asymResultLarger}, which is easy to compute.}

\subsection{Comparison to other models for defect tolerance} \label{sec::otherScenarios}

This paper studies the defect-tolerance model where steps proceed as follows:
\begin{enumerate}
\item[a.] bipartite graph (interconnect) is designed;
\item[b.] primary nodes get $q$-ary labeling;
\item[c.] redundant nodes are assigned $q$-ary labels (so that each primary node has $t$ neighbors with matching
label).
\end{enumerate}
There are two natural variations where the sequence of steps are interchanged:
\begin{itemize}
\item \textit{adaptive graph}: b.$\to$a.$\to$c.
\item \textit{non-adaptive redundancy}: a.$\to$c.$\to$b.
\end{itemize}
In the first case, the design of the edges of the graph is a function of the $q$-ary labels, while in the second case the redundant nodes are not
allowed to depend on the labeling of primary nodes.

It is clear that the setting considered in this paper (a.$\to$b.$\to$c.) is an intermediate case. That is, any
$t$-defect correcting design in the sense of Definition~\ref{def:main} is also $t$-defect correcting in the sense of the
\textit{adaptive graph}. Similarly every design with \textit{non-adaptive redundancy} should work in the sense of
Definition~\ref{def:main}.

The fundamental redundancy-wiring complexity trade-off is defined similarly to~\eqref{eq:rtdef}. However, for both cases it is rather easy to determine this trade-off for any $t\ge 1$:
\begin{itemize}
\item \textit{adaptive graph}: Clearly the number of edges $E\ge kt$. This can be attained with (asymptotically)
zero-redundancy by adding $t$ redundant nodes of each label (for a total of $m=qt$) and connecting every primary
node only to relevant $t$ redundant nodes. Consequently, here
\begin{equation}
 \mathcal{R}_t = \{(\varepsilon, \rho): \varepsilon \ge 1, \rho \ge 0\}\,. 
\end{equation}
\item \textit{non-adaptive redundancy}: Again, clearly the number of edges is $E\ge qkt$. This can be attained with (asymptotically)
zero-redundancy by adding $t$ redundant nodes of each label (for a total of $m=qt$) and connecting every primary
nodes to all of $qt$ redundant ones (i.e., using $K(k, qt)$ design). Consequently
\begin{equation}
 \mathcal{R}_t = \{(\varepsilon, \rho): \varepsilon \ge q, \rho \ge 0\}\,. 
\end{equation}
\end{itemize}

These observations are summarized in \Cref{fig::scenarioCompare} for $q = 2$.

\begin{figure}
\centering
\includegraphics[scale = .45]{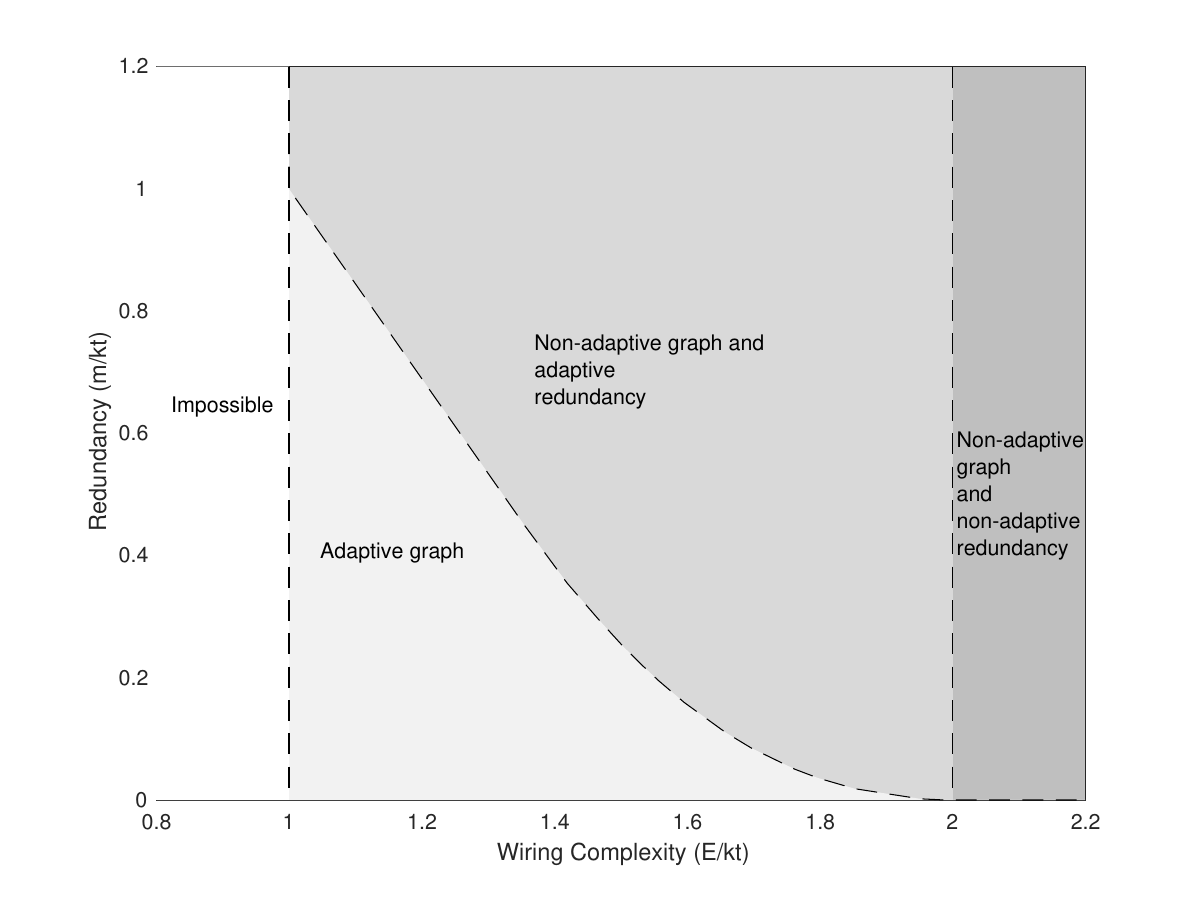}
\caption{Comparison of redundancy-wiring complexity trade-offs for different levels of adaptivity for $q = 2$.} \label{fig::scenarioCompare}
\end{figure}
\subsection{Relation to $(t,t)$-colorable hypergraphs}\label{sec:hyper}

There is a purely graph-theoretic way to look at our problem. For this we bring up the concept of a $(t,t)$-graph coloring
introduced in~\cite{alon2009power}.  A hypergraph is called $(t,t)$-colorable if for every $\{0,1\}$-coloring of
hyperedges there exists a $\{0,1\}$-coloring of vertices so that each edge contains $t$ vertices of its color.  Define
\begin{multline}
 d_t(k,m) = \min (\text{average edge-size: } \\ \text{all } (t, t) \text{ colorable hypergraphs on } \\ m \text{ vertices and } k
\text{ hyperedges})\,.
\end{multline}

It is not hard to see that our problem with binary $\matx$ and $(t,t)$-coloring are one-to-one related: the vertices correspond to primary nodes and the hyperedges correspond to redundant nodes. More precisely we have

\begin{proposition} \label{prop::hyperLimit}Fix binary $\matx$.
The boundary of $\mathcal{R}_t$ is given by 
\begin{equation}
\liminf_{k\to\infty} \frac{1}{t}d_t(k, \lceil \rho kt\rceil) \,. 
\end{equation}
The boundary of $\mathcal{R}_\infty$ is given by 
\begin{equation}
 \liminf_{t \rightarrow \infty} \liminf_{k \rightarrow \infty} \frac{1}{t} d_t(k, \lceil \rho k t\rceil)\,. 
\end{equation}
\end{proposition}

\begin{proof} 

Note that for a fixed $t$, if for some pair $(\varepsilon, \rho)$ we have $\frac{1}{t}d_t(k, \lceil \rho kt\rceil) = \varepsilon$ for some $k$, then by copying (\Cref{prop:copying}), there exists infinitely many values of $k' > k$, where $\frac{1}{t}d_t(k', \lceil \rho k't\rceil) \leq \varepsilon$. It follows from \Cref{prop:basic} that $\liminf_{k\to\infty} \frac{1}{t}d_t(k, \lceil \rho kt\rceil)$ must correspond to the boundary of $\matr_t$. 

Similarly, by merging (\Cref{prop:merge}), any pair $(\varepsilon, \rho)$ where there is some $t$ such that $$ \liminf_{k \rightarrow \infty} \frac{1}{t} d_t(k, \lceil  \rho k t\rceil) = \varepsilon $$ for some $t$, must also have infinitely many values of $t' > t$ where $\liminf_{k \rightarrow \infty} \frac{1}{t'} d_{t'}(k, \lceil  \rho k t'\rceil) \leq \varepsilon $.

\end{proof}

\textcolor{black}{Hypergraphs were used \cite{alon2009power} to show a specific achievability scheme for storing data with bitprobes. This achievability scheme thresholds of the number of neighbors with value $0$ to determine values of data points. Using the connection our defect correcting designs have with hypergraphs, we can use \Cref{thm::asymResult} to show a converse bound on the size of the encoded vector for bitprobes that use thresholding. However, the constants we get from applying our work to bitprobes does not do better than those cited in \cite{alon2009power}. For instance, for $3$ bitprobes and vectors where at most $1/3$ of the entries are $1$, our result gives that the ratio of the length of the encoded vector to the length of the original vector must be greater than $.21$ whereas the method cited in \cite{alon2009power} gives that the ratio must be greater than $.48$.}

\subsection{Stochastic defects}

This work considered correcting arbitrary (worst-case) defect patterns. Suppose that instead we are interested in correcting fraction $\alpha$ of defects (i.e., $t = \alpha (k + m)$) on $k$ primary and $m$ redundant nodes. In this scenario, the number of redundant nodes $m$ would need to grow as a function of $k$ in order to keep up with the number of defects needed to be corrected. If $\alpha$ is too large, it is not possible to find designs which corrects $\alpha(k+ m)$ defects for arbitrarily large $k$. 

To see this, note that correcting worst case $t$ defects with alphabet size $q$ requires at least $qt$ redundant nodes. 
\begin{align}
m & \geq qt \\
m & \geq q \alpha (k+m)\\
m (1-q\alpha) & \geq q\alpha k\,. \label{eq::stoc_q_bound}
\end{align}  

The quantity on the right-hand side of \eqref{eq::stoc_q_bound} needs to be positive, so it must be that $\alpha < \frac{1}{q}$.  

Additionally, the only designs which can correct fraction $\alpha < \frac{1}{q}$ of defects for growing $k$ are designs with the same redundancy and wiring complexity as complete designs. From our results in \Cref{thm::asymResultLarger}, we know that there exists $(k,m,t,E)_q$-designs so that 

\begin{equation}
\frac{m}{kt} \rightarrow c
\end{equation}
for some constant $c$. When $t = \alpha(k+m)$, 

\begin{align}
\frac{m}{k\alpha(k + m)} &> c  \\
m(1 - ck \alpha) &> c k^2 \alpha \,.\label{eq::stoc_c_zero}
\end{align} 
 
In order for \eqref{eq::stoc_c_zero} to hold, the right-hand side must be positive, so it must be that $c \rightarrow 0$ as $k$ becomes arbitrarily large. The point in $\matr_\infty$ where $\frac{m}{kt} \rightarrow 0$ corresponds to the complete design. 
 
In light of these results, it is natural to ask what happens if instead we relaxed the requirement to correcting i.i.d. Bernoulli($\alpha$) defects in the sense of high probability (computed over distribution of defects and primary assignments). It turns out that in such probabilistic model, correcting fraction-$\alpha$ of defects is possible with designs possessing $O(k \log k)$ edges and $O(k)$ redundant nodes. See~\cite{dathesis} Theorems 4.10 and 4.15 in Section 4.4 for more (pp. 63-66).

\subsection{Future work}\label{sec::future}

One direction for future work involves extensions beyond the bipartite graph. We chose to study the one-level bipartite graph model for simplicity, but experiments like Teramac~\cite{teramac98} have demonstrated the effectiveness of multi-level hierarchical designs. This leads to the question of what are the optimal trade-offs when hierarchical models of redundancy are used. The hierarchical model would include intermediate nodes which can facilitate connections of edges. The presence of the intermediate nodes can greatly reduce the number of edges. To correct $t$ defects, we can connect each primary node to $t$ intermediate nodes. Regardless of the number of primary nodes, the intermediate nodes can connect to finitely many redundant nodes. This way, we are able to achieve a wiring complexity of $t$ and redundancy of $0$ (asymptotically). In such a case, we would be interested in finding the fundamental trade-offs with the number of intermediate nodes as a parameter.

\subsection{Open Problems}

Regions which are still to be determined include:
\begin{itemize}
\item $\matr_t$ for $ t > 2$ and $q = 2$ 
\item $\matr_t$ for $ t > 1 $ and $q \geq 3$
\end{itemize}

For $q = 2$, it is also unknown what the smallest value of $t$ is for which $\matr_t$ does not equal the region defined in \Cref{eq::linRegion}.

\section*{Acknowledgement}

YP would like to thank Prof. Jaikumar Radhakrishnan for interesting discussions at the Simons Institute for the Theory
of Computing (UC Berkeley), in particular for bringing~\cite{alon2009power} to our attention.

%----------------------APPENDIX-------------------------%
%----------------------APPENDIX-------------------------%
%----------------------APPENDIX-------------------------%
%----------------------APPENDIX-------------------------%

\appendix

%----------------------PROOF FOR TERNARY-------------------------%
%----------------------PROOF FOR TERNARY-------------------------%
%----------------------PROOF FOR TERNARY-------------------------%
%----------------------PROOF FOR TERNARY-------------------------%
\section{}

\subsection{Proof of \Cref{thm:3alpha}}
\label{apx:3alpha}

\begin{proof}

Define $\matr_1$ as in \eqref{eq::3alphaRegion}. We will show that all $(k,m,1,E)_3$-designs must lie in $\matr_1$. Let the primary nodes have labels in $\mathcal{X} = \{0,1,2\}$. 

Instead of saying that a given bipartite graph is $1$-defect correcting for alphabet of size $q = 3$, we will say (for brevity) that a graph satisfies property (*).

%\begin{definition}
\emph{(*) is the property that for any labeling of the primary nodes in $\mathcal{X}^k$, where $k$ is the number of primary nodes, there exists a labeling of the redundant nodes so that each primary node has at least one redundant node neighbor with the same labeling.} 
%\end{definition}

The steps for this proof are:
\begin{enumerate}
\item Show that designs with primary nodes of degree $3$ and greater can be disregarded. \label{t3step::deg3}
\item Show that in order to satisfy (*), designs with any primary nodes of degree $1$ must be in $\mathcal{R}_1$ . \label{t3step::deg1}
\item Show that in order to satisfy (*), designs with primary nodes all of degree $2$ must be in $\mathcal{R}_1$ . \label{t3step::deg2}
\begin{enumerate}
	\item Show designs containing two disjoint cycles connected by a path (see \Cref{fig::case1}) violate (*)\label{t3step::deg2::case1}
	\item Show designs containing two cycles which intersect at one point (see \Cref{fig::case2}) violate (*)  \label{t3step::deg2::case2}
	\item Show designs containing two cycles which intersect at multiple points (see \Cref{fig::case3}) violate (*) \label{t3step::deg2::case3}
\end{enumerate}
\end{enumerate}

\textbf{Step (\ref{t3step::deg3})} 
The key to this step is to make a graph with (almost) equivalent parameters where nodes of degree $3$ or more are in a separate component. 
For any $(k, m, 1, E)_3$-design $G$, define a new design $G'$ with the same number of primary nodes $k$ and the number of redundant nodes as $m'=m+qt = m + 3$. The added redundant nodes are connected to each of the primary nodes that have degree (in $G$) larger or equal to $3$. The remaining primary nodes are connected in $G'$ exactly as in $G$. It is clear that $G'$ still satisfies (*), has the same (or smaller) number of edges and (asymptotically in $m$) the same redundancy $\rho$. This shows, that without loss of generality we can assume that there are no primary nodes of degree greater than $3$  and all nodes of degree $3$ form a complete bipartite graph disjoint from the rest of the design. We can ignore this separate component.

%we can find an equivalent design where all primary nodes of degree $3$ or more are in a separate component with an addition of finitely many more redundant nodes. Thus, it is sufficient to prove (*) for connected graphs with primary nodes of degree $1$ and $2$ only. 

\textbf{Step (\ref{t3step::deg1})}
The main argument of this step is to show that if the design has any primary node of degree $1$, the design must be a tree. 

%\begin{definition}
{We will say a primary node is \emph{adjacent} to another primary node if the two primary nodes share a redundant node as a common neighbor.}
%\end{definition}
 
Suppose the design has a primary node $u_0$ of degree $1$ and no primary nodes of degree $3$ or more. For all primary and redundant nodes, we will consider the node's shortest distance to $u_0$. If a node is distance $i$ from $u_0$, we say that the node is at level $i$. Since the design is bipartite, even levels have primary nodes, and odd levels have redundant nodes. Let tier $n$ mean the levels $2n$ and $2n+1$.

Consider the labeling of the primary nodes where all primary nodes in even tiers are labeled $0$'s and all primary nodes in odd tiers are labeled $1$'s. Since $u_0$ only has one neighbor, in order to satisfy (*), we must label the one redundant node in tier $0$ the value $0$. The primary nodes in tier $1$ each have the redundant node in tier $0$ and some redundant node in tier $1$ as neighbors. The redundant node in tier $0$ is already labeled a $0$, so all redundant nodes in tier $1$ must be labeled $1$'s in order to satisfy (*). Since primary nodes in tier $2$ are labeled with $0$'s, then all redundant nodes in tier $2$ must also be labeled $0$'s. Continuing this argument by induction, all redundant nodes in a tier must be labeled the same value as the primary nodes in that tier. If the design is a tree, then this labeling scheme satisfies (*).

Now suppose the design has cycles. Find the lowest tiered primary node which completes a cycle, that is the lowest tiered primary node $u_c$ in tier $c$ which has one redundant node in tier $c - 1$ and the other redundant node it has is shared by a primary node in tier $b$, where $b \leq c$. (It could be that two nodes in tier $c$ share the same redundant node in tier $c$. Pick either as $u_c$). Now switch the label of $u_c$ to $2$ and keep all the other labels the same. Both redundant neighbors of $u_c$ must be labeled a $0$ or $1$, so we do not satisfy (*).  

Condition (*) is not satisfied unless the design has no cycles and is a tree. If the design is a tree, it must have at least the same number of redundant nodes as primary nodes, so the design lies in $\mathcal{R}_1$. We can now assume that all primary nodes are of degree $2$. 

%\begin{remark}
\emph{Notice having $q \geq 3$ is important to avoid existence of even cycles.}
%\end{remark}

\textbf{Step (\ref{t3step::deg2})}
Our goal is to prove that if all primary nodes have degree $2$, then $\frac{m}{k} \geq \frac{1}{2}$. We will instead prove something stronger: For $k > 4$ we must have $m \geq k$. For $k = 4$ we must have $m \geq 3$.

Because of copying (see \eqref{eq::copy1}-\eqref{eq::copy2}), it is sufficient to prove the above for designs on a single component. If there is a redundant node with degree $1$, we can remove this redundant node with its neighboring primary node from the design and make it to a separate component. We can assume all redundant nodes must also have degree $2$ or more. 

%\begin{definition}
{We will call a labeling of primary nodes \emph{alternating} if adjacent primary nodes have different labels.}
%\end{definition} 

\begin{lemma}\label{lem::alternating}
{If a design with $k$ primary nodes, all of degree $2$, and $k-1$ or fewer redundant nodes, all of degree 2 or 3, can be labeled alternatingly, then the design cannot satisfy (*).}
\end{lemma}

\begin{proof}
If there is an alternating labeling, at most each redundant node can only match the labeling of one of its neighboring primary nodes. There can only be at most $k-1$ matches, so there exists one primary node which does not have a neighboring redundant node with the same label as itself. 
\end{proof}

Suppose a design with all primary nodes of degree $2$ is such that $m < k$. Then, some redundant node in the design must have degree $3$ or more. Pick the separate component with this redundant node, and let $A$ be a cycle in this separate component (if this component does not have cycles, then $m \geq k$ as in Step \ref{t3step::deg1}). In order for cycle $A$ to be in this component, a redundant node with degree more than $2$ must also be in $A$. Call this redundant node $v_0$. Call the neighbor of $v_0$ which is not in $A$ $u_0$. 

Build a path $B$ in the design starting at primary node $u_0$ as follows: The second node in path $B$ will be the neighbor of $u_0$ which is not $v_0$. We can pick the next node in the path arbitrarily. The path ends when we reach a node in $A$ or a node already in $B$. To show that (*) does not hold on the design, it is sufficient to show that (*) is not satisfied on the subgraph $A \cup B$. 

Depending on the endpoint of path $B$, we have several cases:

%\begin{enumerate}
%\item
%\vskip5mm

\textbf{Case (\ref{t3step::deg2::case1}): Endpoint of $B$ coincides with an intermediate point of
   $B$ (see \Cref{fig::case1})}
   
\textcolor{black}{Let $v_1$ be the redundant node in path $B$ where the path $B$ ends. Rename the cycle created by path $B$ to cycle $C$. The subgraph $A \cup B \cup C$ satisfies the conditions of \Cref{lem::alternating} so we need only show that we can find an alternating label. The nodes $v_0$ and $v_1$ are the only redundant nodes with degree $3$. It is clear that by starting with a alternating label of the neighbors of $v_0$, we can find an alternating label for the path between $v_0$ and $v_1$, and then an alternating label of the neighbors of $v_2$. After this, it is easy to find an alternating label for the rest of $A$ and $C$.}

%\item
\textbf{Case (\ref{t3step::deg2::case2}): Endpoint of $B$ is node $v_0$ (see \Cref{fig::case2})}

Let cycle $C$ be the cycle formed using path $B$ and $v_0$. As long as one of cycle $A$ and cycle $C$ have more than $2$ primary nodes, the design violates (*). 

Consider when the labeling is so that the two primary nodes in the larger cycle (assume this to be cycle $A$) neighboring $v_0$ are labeled the same value, say $0$. The rest of the primary nodes of $A$ are labeled alternatingly, which is possible because cycle $A$ has at least three primary nodes. Let the two primary nodes neighboring $v_0$ in cycle $C$ be labeled $1$ and $2$. Then $v_0$ must be labeled $0$ in order for each node in cycle $A$ to have a neighbor with the same label. Then if cycle $C$ is labeled alternatingly, we will violate (*).

If both cycles have only $2$ primary nodes, it is possible for this design to have $k = 4$ and $m = 3$ and satisfy (*). See \Cref{fig::L3t1k4}.

%\item
\textbf{Case (\ref{t3step::deg2::case3}): Endpoint of $B$ is some node of $A$ different from $v_0$ (see \Cref{fig::case3})}

Two redundant nodes in the design have degree at least $3$. Call them $v_0$ and $v_1$. Cycle $A$ and path $B$ make up three distinct paths which go from $v_0$ to $v_1$, which we will refer to as $E, F$ and $G$. As long as no two paths have only $1$ primary node, then we can find an alternating label and use \Cref{lem::alternating}.

Label $u_{i, X}$ to be the primary node neighboring $v_i$ and in path $X$. If all paths $E,F,$ and $G$ have two or more primary nodes, assign labels $0, 1, 2$ to $u_{0, E}, u_{0, F}, u_{0, G}$ and $1, 2, 0$ to $u_{1, E}, u_{1, F}, u_{1, G}$. Each path can be labeled alternatingly.     

If there is one path with only one primary node, say path $E$, assign labels $0, 1, 2$ to $u_{0, E}, u_{0, F}, u_{0, G}$ and $2, 1$ to $u_{1, F}, u_{1, G}$. Each path can be labeled alternatingly. 

If two paths both have one primary node, say $E$ and $F$, as long as the third path $G$ has at least $3$ primary nodes, the design can be labeled alternatingly. We can label the two primary nodes cycle created by paths $E$ and $F$ the values $0$ and $1$. Then since $G$ has at least $3$ primary nodes, we can label $u_{0, G}$ and $u_{1, G}$ the value $2$ and label the rest of $G$ alternatingly. 

If $G$ only has $2$ primary nodes, then this is a design on $k = 4$ and $m = 3$ which satisfies (*). See \Cref{fig::L3t1k4_2}.

%\end{enumerate}

\begin{figure}
\centering
\subfigure[Case (\ref{t3step::deg2::case1})]{
	\includegraphics[scale = .1]{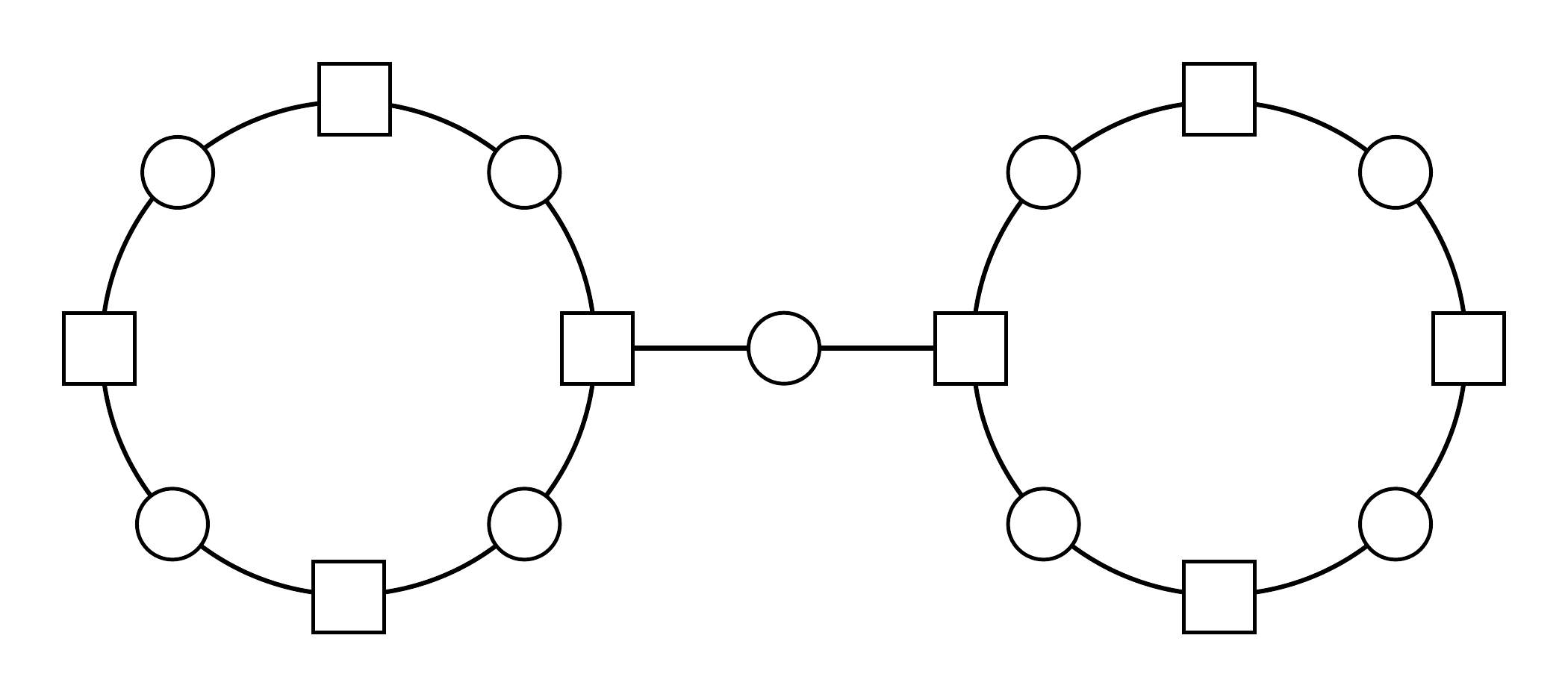}
	\label{fig::case1}
} \\[1em]
\subfigure[Case (\ref{t3step::deg2::case2})]{
	\includegraphics[scale = .1]{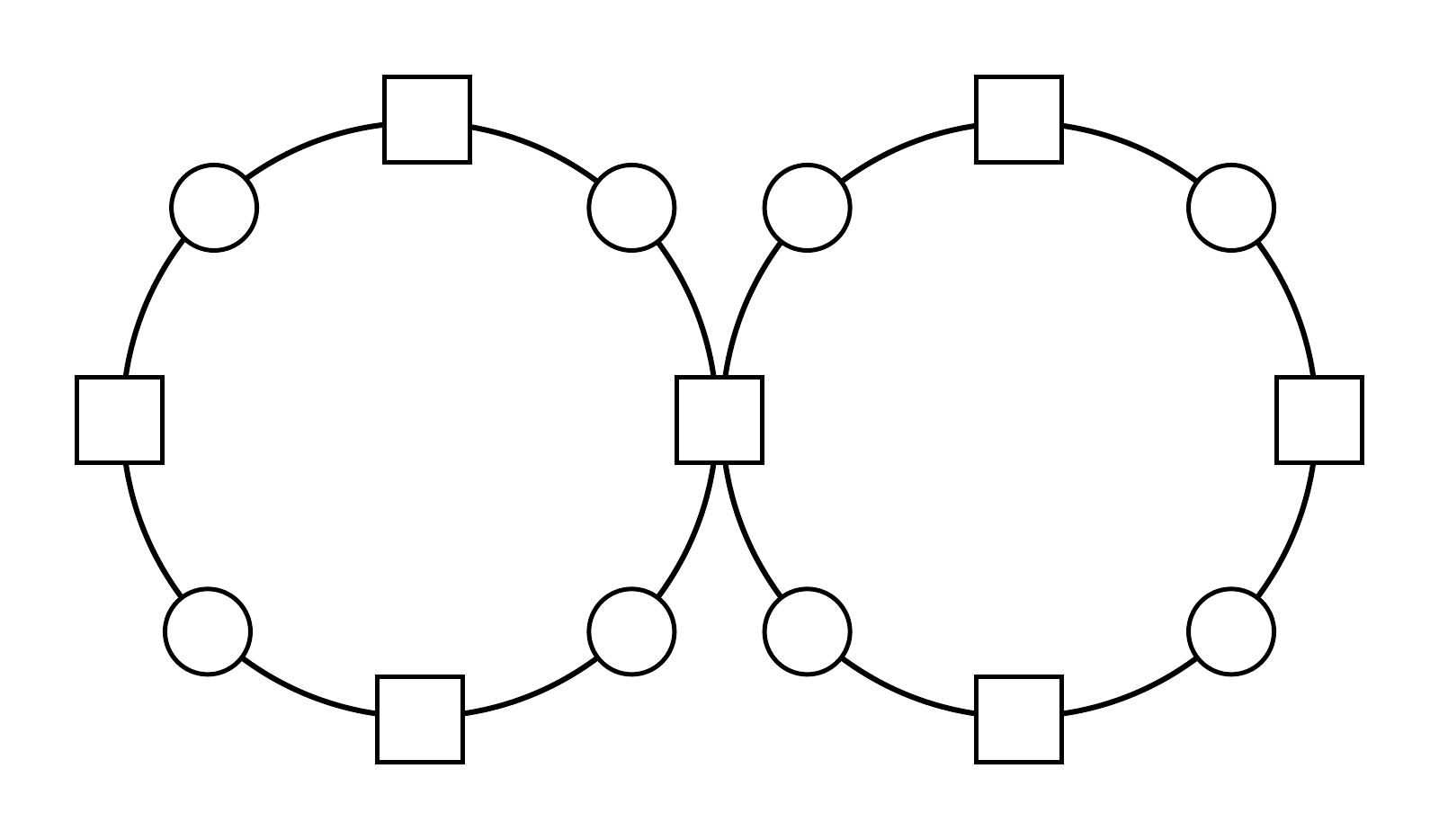}
	\label{fig::case2}
} \\[1em]
\subfigure[Case (\ref{t3step::deg2::case3})]{
	\includegraphics[scale = .1]{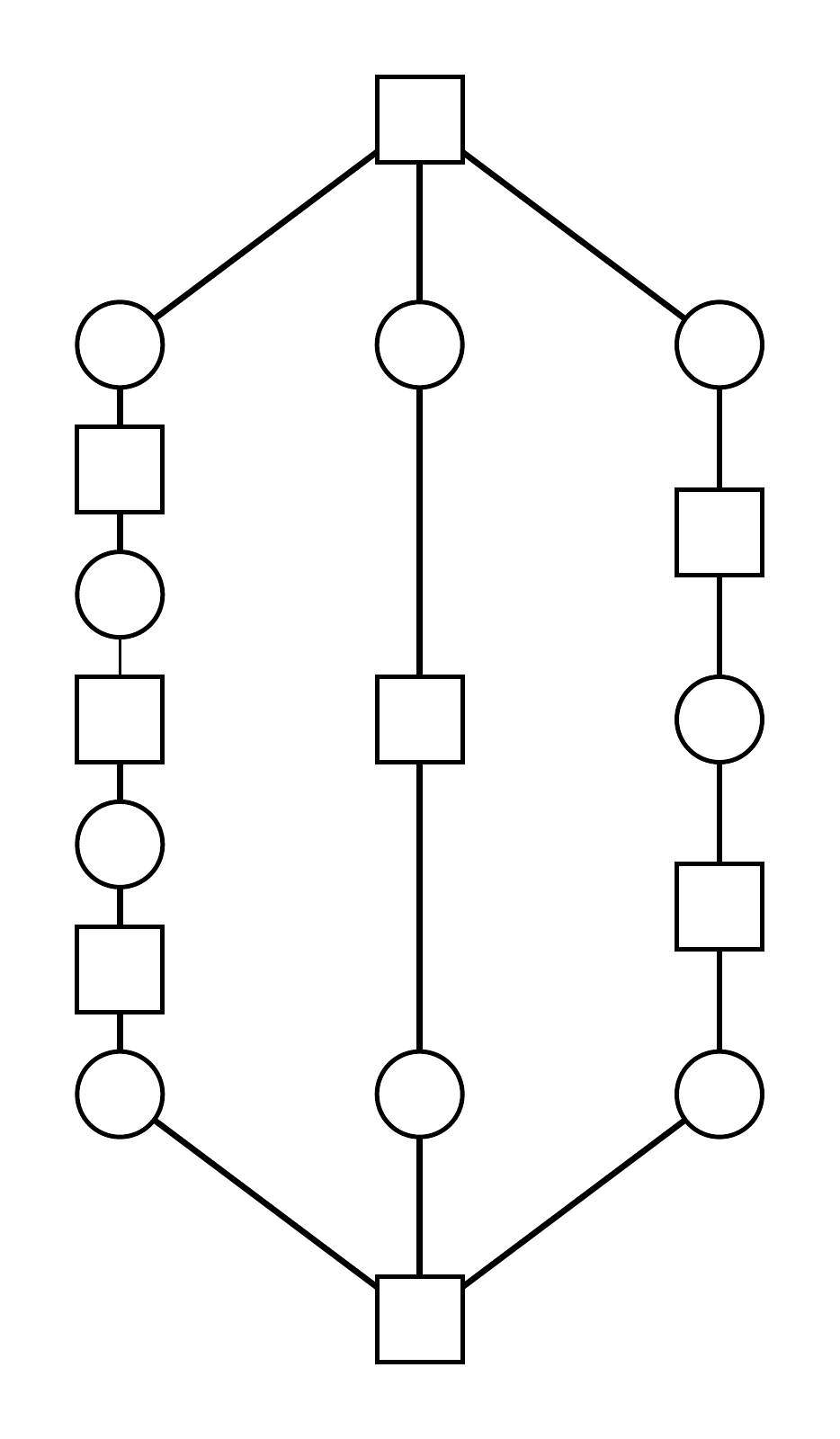}
	\label{fig::case3}
}
\caption{Example designs for the different cases considered in the proof of \Cref{thm:3alpha}.}
\label{fig::L3cases}
\end{figure}

\begin{figure}
\centering
\includegraphics[scale = .2, angle = 90]{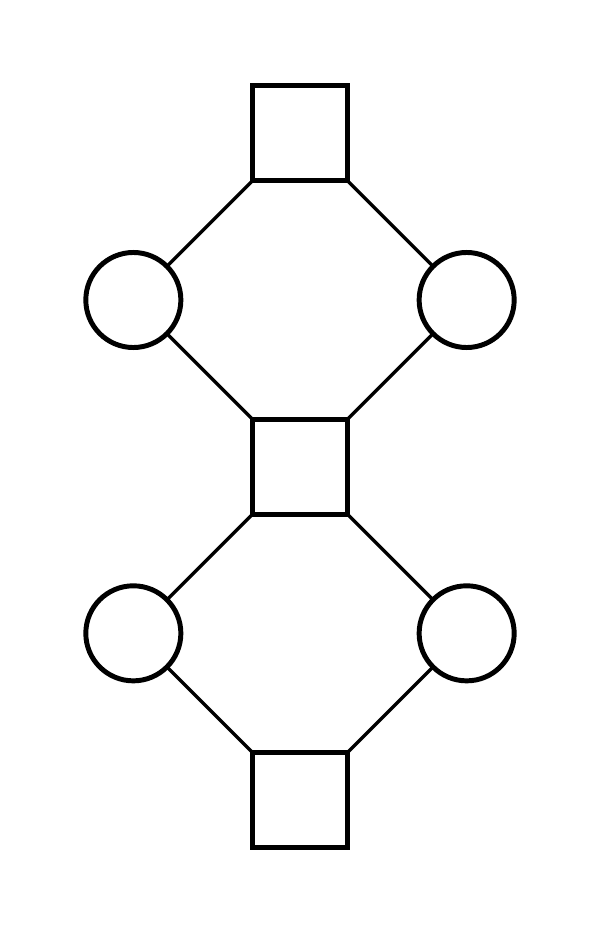}
\caption{Design which satisfies (*). Exception to case (\ref{t3step::deg2::case2}). \label{fig::L3t1k4}}
\end{figure}

\begin{figure}
\centering
\includegraphics[scale = .2]{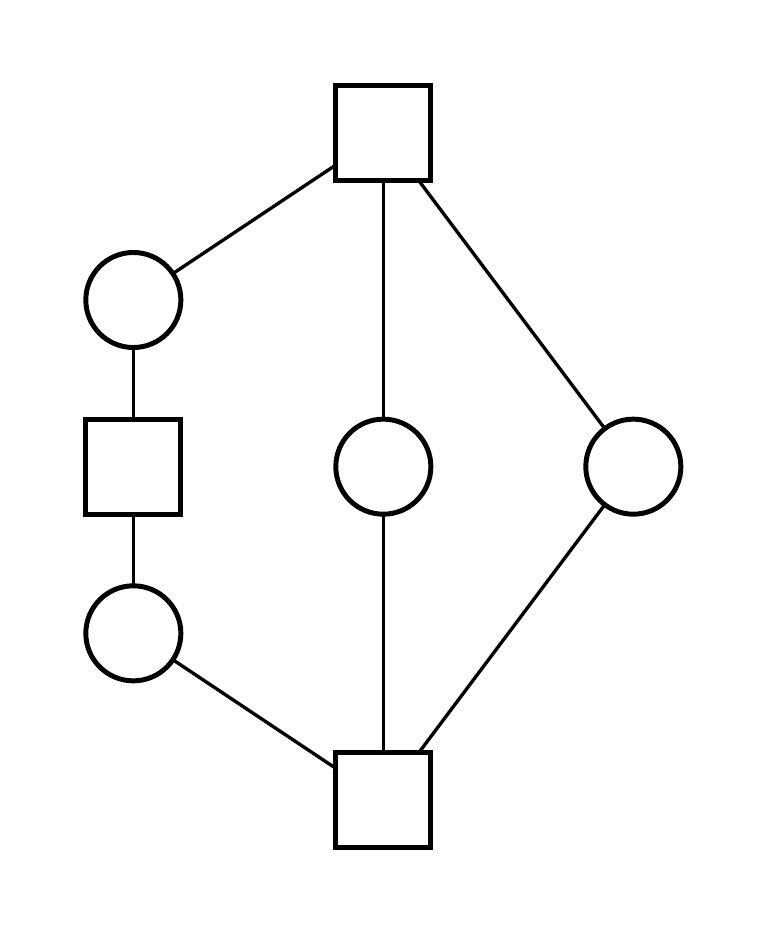}
\caption{Design which satisfies (*). Exception to case (\ref{t3step::deg2::case3}). \label{fig::L3t1k4_2}}
\end{figure}

Note that the two exceptions with $k= 4$ are precisely the minimal non-trivial $1$ defect correcting designs. One of these designs was discussed in~\Cref{sec:smallOptimal}.

\end{proof}

%----------------------Asymptotic Lemmas-----------------------%
%----------------------Asymptotic Lemmas-----------------------%
%----------------------Asymptotic Lemmas-----------------------%
%----------------------Asymptotic Lemmas-----------------------%

\subsection{Proof of \Cref{prop::Fnkprops}}\label{sec::fnkpropsProof}

Before we present the proof of \Cref{prop::Fnkprops}, we will first show the following lemma:

% show TV difference between multinomial and hypergeometric
\begin{lemma}\label{lem::TVmult}
Fix $P_X \in \frac{1}{k}\mathbb{Z}$ and $s\in \ZZ_+$. Let $\vect L \sim \mathrm{Mult}(s,P_X)$, cf.~\eqref{eq:multdef},
and $\vect L^{(k)}\sim \mathrm{HyperGeom}(k,s,P_X)$, cf.~\eqref{eq:hyperdef}. Then we have the following total variation
estimate:\footnote{Total variation distance $\TV$ for probability measures $P$ and $Q$ on sigma algebra $\matf$ defined
as $\TV(P, Q) = \sup_{A \in \matf} |P(A) - Q(A)|$.}
\begin{equation}\label{eq:tv1}
	\TV(P_{\vect L^{(k)}}, P_\vect L) \le {s^2\over 2k}\,. 
\end{equation}
Similarly, if $\vect M \sim \mathrm{Mult}(s-1,P_X)$ and $\vect M^{(k-1)}$ has distribution
\begin{equation}
\PP[\vect M^{(k-1)} = \vect m] = {{k\pi_1 \choose m_1} \cdots {k\pi_j-1 \choose m_j} \cdots  {k\pi_q\choose m_q} \over
{k-1 \choose s-1}}
\end{equation}
for an arbitrary $j$, we also have
\begin{equation}\label{eq:tv2}
	\TV(P_{\vect M}, P_{\vect M^{(k-1)}}) \le {s^2\over 2k}\,.
\end{equation}
\end{lemma}

\begin{proof}
Standard estimate for total variation via coupling states that for any joint distribution $P_{\vect L, \vect L^{(k)}}$:
\begin{equation}
\TV(P_{\vect L^{(k)}}, P_\vect L) \leq \PP[\vect L^{(k)} \neq \vect L]\,.
\end{equation}

Notice that $\vect L^{(k)}$ encodes the color distribution after sampling $s$ balls from a collection of $k$ colored
balls (with composition given by $P_X$) without replacement, while $\vect L$ is the color distribution for sampling $s$ balls
with replacement. Let us couple these two samples as follows. Number all balls from $1$ to $k$ and define infinite 
string of i.i.d. uniform $X_i \in [k]$. Let our sample with replacement be the balls with indices $X_1,\ldots, X_s$, while the
sample without replacement be the balls $X_1, X_{i_2}, \ldots, X_{i_s}$ where ${i_t}$ denotes the first element of the
sequence where $t$-th unique index appeared (e.g., for $X=(1,2,2,3,\ldots)$ we have $i_2=2$, $i_3=4$, etc).
Now the two samples are going to be different only if $X_1,\ldots,X_s$ are not distinct and this happens with
probability at most
\begin{equation}\label{eq:tv3}
	\sum_{i = 1}^{s-1} \frac{i}{k} = \frac{s^2 - s}{2k} < \frac{s^2}{2k}\,.
\end{equation}
This proves~\eqref{eq:tv1}. For~\eqref{eq:tv2} modify distribution of $X$ sequence by setting $X_1=j$ and the rest are
still i.i.d. uniform on $[k]$. Then $\vect M$ is the color composition of $X_2,\ldots,X_s$ while $\vect M^{(k-1)}$ is the
color composition of $X_{i_2},\ldots,X_{i_s}$. Again, $X_2,\ldots,X_s$ are not all distinct with probability at
most~\eqref{eq:tv3}.
\end{proof}

\begin{proof}
Proof of~\eqref{eq:fnk_bounds}: Simply by definition we have $ F_{k,n}(P_S) \leq F_k(P_S)$, so we focus on the
opposite direction. 
First, we show that if $\vect L^{(k)} \sim \text{HyperGeom}(s,k,[\pi_1, \cdots, \pi_q])$, for any function $f: \vect L^{(k)}\to
\mathbb{R}$ and any fixed $j\in[q]$ we have 
\begin{equation} \label{eq::binomSimp}
\frac{1}{\pi_j} \mathbb{E}[L^{(k)}_j  f(\vect L^{(k)})] = \mathbb{E}[S \cdot f(\vect M^{(k-1)} + \vect e_j)]\,,
\end{equation}
where $\vect e_j$ is a vector with one in $j$-th position and the rest zeros, and $\vect M^{(k-1)}$ has hypergeometric
distribution
\begin{equation}
 \PP[\vect M^{(k-1)} = \vect m] = {{k\pi_1 \choose m_1} \cdots {k\pi_j-1 \choose m_j } \cdots  {k\pi_q\choose m_q} \over
{k-1 \choose s-1}}\,.
\end{equation}
To that end, simply notice that
\begin{align*}
&\frac{1}{\pi_j} \mathbb{E}[L^{(k)}_j  f(\vect L^{(k)})]  \\
&\quad =  \sum_{s} P_S(s)\sum_{\vect \ell} \frac{\ell_j}{\pi_j} {{k\pi_1 \choose \ell_1}
\cdots {k\pi_q\choose \ell_q} \over {k \choose s}}  f(\vect \ell) \numberthis\\
&\quad =   \sum_{s} P_S(s) \sum_{\vect m} s{{k\pi_1 \choose m_1} \cdots {k\pi_j - 1\choose m_j} \cdots {k\pi_q\choose m_q} \over {k - 1 \choose
s - 1}}  f(\vect m + \vect e_j) \label{eq::hypergom_s-1} \numberthis\,.
\end{align*}

Now fix $P_X \in \frac{1}{k}\mathbb{Z}$ and $P_{Y^*|\vect L^{(k)}}$ to be the optimal distributions achieving $F_k(P_S)$ in~\eqref{eq:fkdef}. By rounding there must exist $P_{Y_n|\vect L^{(k)}} \in \frac{1}{n} \mathbb{Z}$ so
that $| P_{Y_n|\vect L^{(k)}}(j|\vect \ell) - P_{Y^*|\vect L^{(k)}}(j|\vect \ell)| \leq \frac{1}{n}$ for every $\vect
\ell$. Then for any fixed $j \in [q]$ we have  in view of~\eqref{eq::binomSimp}
\begin{align*}\label{eq:rdd1}
 \left| \frac{1}{P_X(j)} \mathbb{E}[L^{(k)}_j \mathbbm{1}\{Y^*=j\} ] -  \frac{1}{P_X(j)} \mathbb{E}[L^{(k)}_j \mathbbm{1}\{Y_n=j\}] \right| \\
  \le \frac{\mathbb{E}[S]}{n} \numberthis\,.
\end{align*}
Taking $\min_j$ of~\eqref{eq:rdd1} recovers the lower bound in~\eqref{eq:fnk_bounds}.

We proceed to proving~\eqref{eq:fk_bounds}. Fix $P_S$ and let 
\begin{equation}
 h(P_X, P_{Y|\vect L}, j) \eqdef {1\over P_X(j)}\EE[L_j \mathbbm{1}\{Y=j\}]\,, 
\end{equation}
where given $S=s$ we have $\vect L \sim \mathrm{Mult}(s, P_X)$, cf.~\eqref{eq:jdist}. Similar to~\eqref{eq::binomSimp} we have
\begin{equation}\label{eq::binomSimp2}
	h(P_X, P_{Y|\vect L}, j) = \EE[S \cdot P_{Y|\vect L}(j|\vect M + \vect e_j)]\,,
\end{equation}
where this time given $S=s$ we have $\vect M \sim \mathrm{Mult}(s-1, P_X)$. Now, for $P_X\in {1\over k}\ZZ$ define also 
\begin{equation}
h_{k}(P_X, P_{Y|\vect L},j)  \eqdef {1\over P_X(j)}\EE[L_j^{(k)} \mathbbm{1}\{Y=j\}]\,,  
\end{equation}
where given $S=s$ we have $\vect L^{(k)} \sim \mathrm{HyperGeom}(s,k, P_X)$.
From~\eqref{eq::binomSimp},~\eqref{eq::binomSimp2} and Lemma~\ref{lem::TVmult} (namely~\eqref{eq:tv2}) we have then

\begin{equation}\label{eq:rdd2}
	|h(P_X, P_{Y|\vect L}, j) - h_k(P_X, P_{Y|\vect L}, j)| \le {\EE[S^3]\over 2k}\,.
\end{equation}
Finally, since 
\begin{equation}
 (P_X, P_{Y|\vect L}) \mapsto \min_j h(P_X, P_{Y|\vect L}, j) 
\end{equation}
is uniformly continuous on a compact set, we also have
\begin{equation}
 (P_X, P_{Y|\vect L}) \mapsto \max_{P_{Y|\vect L}} \min_j h(P_X, P_{Y|\vect L}, j) 
\end{equation}
is uniformly continuous by \Cref{prop::maxUnifCont}. Hence for some $\epsilon_k\to0$ we have 
\begin{align*} 
&\left| \min_{P_X \in {1\over k}\ZZ} \max_{P_{Y|\vect L}} \min_j h(P_X, P_{Y|\vect L}, j) - \right.\\
&\left. \quad\quad\quad \min_{P_X} \max_{P_{Y|\vect L}} \min_j h(P_X, P_{Y|\vect L}, j)\right| \le \epsilon_k \numberthis \,. 
\end{align*}
Using~\eqref{eq:rdd2} to replace $h$ with $h_k$ in the first term of the latter we get~\eqref{eq:fk_bounds}.
\end{proof}

\begin{proposition} \label{prop::maxUnifCont}
Let $f:X \times Y \to \mathbb{R}$ where $X$ and $Y$ are compact and $f$ is uniformly continuous. Then $\max_{y} f(x,y)$ is uniformly continuous on $X$. 
\end{proposition}

\begin{proof}
Let $h(x) = \max_{y} f(x,y)$. Because $f$ is uniformly continuous, for every $\epsilon > 0$ there exists a $\delta$ so that if the distance between $(x_1, y_1)$ and $(x_2, y_2)$ is less than $\delta$, then $|f(x_1,y_1) - f(x_2, y_2)| < \epsilon$ for all $x_1,x_2 \in X$ and $y_1, y_2 \in Y$. We want to show that for $h$, the same $\delta$ can be used for each $\epsilon$. 
Suppose there exists values $x, x' \in X$ where $|h(x') - h(x)| > \epsilon$ and $|x - x'| <\delta$. Assume that $h(x') > h(x)$. There exists a value of $y$ so that $f(x', y) = h(x')$. Since $|f(x', y) - f(x, y)| \leq \epsilon$ then $h(x') = f(x',y) \leq f(x,y) + \epsilon \leq h(x) + \epsilon$, which is a contradiction. 
\end{proof}

\subsection{Upper bound on $F_k(P_S)$}\label{sec::fkboundf}

\begin{lemma}\label{lem::fklessf2k}
For any $P_S \in \mathbb{Q}$ with finite support, 
\begin{equation}
F_k (P_S) \leq F (P_S)\,.
\end{equation}
\end{lemma}

\begin{proof}
Fix $P_S \in \mathbb{Q}$ with finite support. First, we will show that $F_k (P_S) \leq F_{2k} (P_S)\,.$ Using \Cref{prop::subsetFiniteKResult} and \eqref{eq:fnk_bounds} from \Cref{prop::Fnkprops}, for each $k$ there exists a sequence of subset designs $G_i$ which are $(k, m_i, t_i, E_i)_q$-designs with $E_i = m_i \mathbb{E}[S]$, and $\frac{t_i k}{m_i} \to F_k(P_S)$.

For each $G_i$, we will construct subset design $G_i'$ on $2k$ primary nodes by copying (see \Cref{prop:copying}) two copies of $G_i$. Then $G_i'$ is a $(2k, 2m_i, t_i, 2E_i)_q$-design. By \Cref{prop:sym}, for each $G_i'$, there exists a subset design $G_i''$ which is a $(2k, 2m_i \cdot (2k)!, t_i \cdot (2k)!, 2E_i \cdot (2k)!)$. 
\begin{equation}
F_{2k}(P_S) \geq \lim_{i \to \infty} \frac{t_i \cdot (2k)! 2k}{2m_i \cdot (2k)!} = \lim_{i \to \infty} \frac{t_i k}{m_i} = F_k(P_S)\,.
\end{equation}

Then, $F_k (P_S) \leq F_{2k} (P_S) \leq F_{4k} (P_S) \leq F_{8k} (P_S) \leq F_{16k} (P_S) \cdots$. Since $F_{2^i k} (P_S) \to F(P_S)$ monotonically with convergence given by \eqref{eq:fk_bounds} from \Cref{prop::Fnkprops}, this gives the desired result. 
\end{proof}

%-------------------SPLITTING RATIO PROOF---------------------%
%-------------------SPLITTING RATIO PROOF---------------------%
%-------------------SPLITTING RATIO PROOF---------------------%
%-------------------SPLITTING RATIO PROOF---------------------%

\subsection{Proof of \Cref{prop::splittingRatio}}\label{sec::splittingRatioProof}

\begin{proof}
Fixed $P_S$ with finite support and let $c = \mathbb{E}[S]$. Let $\hat{P}_S(s) = \frac{{P}_S(s)s}{\sum_s {P}_S(s)s } = \frac{1}{c}{P}_S(s)s$. We can substitute in $\hat{P}(s)$ and take the expectation with respect to $\hat{P}_S$ instead of $P_S$ by adjusting \eqref{eq::optfunc2} to 
\begin{align*}\label{eq::optRatio}
F(\hat P_S) = \min_{P_X} \max_{P_{Y|L}} \min \left\{\frac{c}{P_X(0)}\frac{L_0}{S} \mathbbm{1}\{Y = 0\} , \right. \\
\left.\frac{c}{P_X(1)}\frac{L_1}{S} \mathbbm{1}\{Y = 1\}\right\}\,. \numberthis
\end{align*}
Let the redundant node ratio of a redundant node with type $(\ell_0, \ell_1)$ be $\nu = \frac{\ell_0}{\ell_0 + \ell_1}$. Suppose that $P_{Y|L}$ is a labeling so that
\begin{enumerate}
\item $P_{Y|L}(1 | \vect \ell_a) > 0$ where $\vect \ell_a$ so that $\ell_0 + \ell_1 = s_a$ and has ratio $\nu_a$
\item $P_{Y|L}(0 | \vect \ell_b) > 0$ where $\vect \ell_b$ so that $\ell_0 + \ell_1 = s_b$ and has ratio $\nu_b$
\item $\nu_a > \nu_b$
\end{enumerate}

Let $P_{Y'|L}$ be equivalent to $P_{Y|L}$  except that 
\begin{subequations}
\begin{align} 
&P_{Y'|L}(1 | \vect \ell_a) =  P_{Y|L}(1 | \vect \ell_a) - \alpha {\hat P_S(s_b) P_{\vect L|S}(\vect \ell_b |s_b)}
\label{eq::splittingShifts1}\\
&P_{Y'|L}(0 | \vect \ell_a) =  P_{Y|L}(0 | \vect \ell_a) + \alpha {\hat P_S(s_b) P_{\vect L|S}(\vect \ell_b |s_b)}\label{eq::splittingShifts2}\\
&P_{Y'|L}(1 | \vect \ell_b) =  P_{Y|L}(1 | \vect \ell_b) + \alpha {\hat P_S(s_a) P_{\vect L|S}(\vect \ell_a |s_a)}
\label{eq::splittingShifts3}\\
&P_{Y'|L}(0 | \vect \ell_b) =  P_{Y|L}(0 | \vect \ell_b) - \alpha {\hat P_S(s_a) P_{\vect L|S}(\vect \ell_a |s_a)}
\label{eq::splittingShifts4}
\end{align}
\end{subequations}
for an appropriate $\alpha > 0$ which is small enough so that $P_{Y'|L}$ is still a valid distribution. Compared to $P_{Y|L}$, $P_{Y'|L}$ increases both quantities in the brackets in \eqref{eq::optRatio}. So $P_{Y|L}$ cannot be optimal and any optimal $P_{Y|L}$ must have the form of \eqref{eq::splittingRatio}.

For two redundant node type ratios where $\nu_a = \nu_b$, we can also see from \eqref{eq::splittingShifts1}-\eqref{eq::splittingShifts4} that there is a value of $\alpha$ (possibly negative unlike above) so that  $P_{Y'|L}(0 | \vect \ell_a) = P_{Y'|L}(0 | \vect \ell_b)$ and the value of \eqref{eq::optRatio} is not affected by the change. %Thus, there exists a solution where $P_{Y|L}$ is a function of the redundant node type ratio.
\end{proof}

%--------------------NUMERICAL RESULTS--------------------%
%--------------------NUMERICAL RESULTS--------------------%
%--------------------NUMERICAL RESULTS--------------------%
%--------------------NUMERICAL RESULTS--------------------%

\subsection{Numerical results derivation}\label{sec::numerical}

Here we develop upper and lower bounds for the expression found in \Cref{thm::asymResult} (the particular case when $\matx = \{0,1\}$). 

\subsubsection{Almost tight lower bound}\label{apx:numapx}

Our lower bound for the boundary of $\matr_\infty$ will be parametrized by $c$. To get this lower bound, we want to find an upper bound for
\begin{equation}
Z^*(c) = \max_{P_S: \mathbb{E}[S] = c} F(P_S)\,.
\end{equation}

For notation, let $\lambda = P_X(0)$ and $1 - \lambda = P_X(1)$.  Let $\vect\ell = (\ell_0, \ell_1)$ and $f(\ell_0, \ell_1) = f(\vect\ell) = P_{Y|L}(0|\vect\ell)$ where $f$ can take any value between $[0, 1]$.
For a fixed $\lambda$ and $s$, define random variable \textcolor{black}{$\vect M(s,\lambda) = (V, s - 1- V)$} where $V \sim \text{Bino}(s-1, \lambda)$ and $\vect e_0 = (1, 0)$ and $\vect e_1 = (0, 1)$ according to \Cref{lem::TVmult} and the proof of \Cref{prop::Fnkprops}. \textcolor{black}{(Just for clarity in this section, we added arguments in paranthesis for $\vect M$.)} Fix $P_S$ to have finite support.

First, we have that
\begin{align*}
F(P_S) =& \min_{0 \leq \lambda \leq 1} \max_{0 \leq f \leq 1} \min \left\{\frac{1}{\lambda} \mathbb{E}[L_0 f(\vect L)] , \right. \\ & \quad \left.\frac{1}{1 - \lambda}\mathbb{E}[L_1 (1 - f(\vect L))] \right\} \numberthis\\
=& \min_{0 \leq \lambda \leq 1} \max_{0 \leq f \leq 1} \min_{0 \leq \alpha \leq 1} \alpha\mathbb{E}[S \cdot f(\vect M(S,\lambda) + \vect e_0) ] \\ & \quad +(1-\alpha) \mathbb{E}[S \cdot (1 - f(\vect M(S,\lambda) + \vect e_1)) ]  \label{eq::swapc} \numberthis\\
\leq& \min_{0 \leq \lambda \leq 1}  \frac{1}{2}\mathbb{E}\left[S \cdot \max_{0 \leq f \leq 1}   \left(1 \right. \right.\\ & \quad \left. + f(\vect M(S,\lambda) + \vect e_0) - f(\vect M(S,\lambda) + \vect e_1)\right)\bigg] \label{eq::chalf} \numberthis\\
 =& \min_{0 \leq \lambda \leq 1}  \frac{1}{2} \mathbb{E}\bigg[S \cdot \left(1 +\right.\\ & \quad \left.\max_{0 \leq \ell_0 \leq S - 1} \mathbb{P}[\vect M(S,\lambda) = (\ell_0, S - 1 - \ell_0)] ) \right] \label{eq::maxProbBound}\numberthis \\
 \eqdef& \min_{0 \leq \lambda \leq 1} \mathbb{E}[ \phi(S, \lambda)] \numberthis
\end{align*}
where in \eqref{eq::swapc} we use \eqref{eq::binomSimp2} and convexify the minimum using $\alpha$, and then \eqref{eq::chalf} follows by setting $\alpha = \frac{1}{2}$. To get \eqref{eq::maxProbBound}, notice that for a fixed $s$

\begin{align}
&\mathbb{E}[1 + f(\vect M(s,\lambda) + \vect e_0)   - f(\vect M(s,\lambda) + \vect e_1)]  \\
\begin{split}
& = 1 + \sum _{\ell_0 = 0}^{s-1} \mathbb{P}[\vect M(s,\lambda) = (\ell_0, s - 1 - \ell_0)] f(\ell_0 + 1, s-1 - \ell_0) \\& \quad \quad - \sum _{\ell_0 = 0}^{s-1}  \mathbb{P}[\vect M(s,\lambda) = (\ell_0,  s -1 - \ell_0)] f(\ell_0, s -\ell_0) 
\end{split}\label{eq:subftilde}\\
%& = 1 + \sum _{\ell_0 = 0}^{s} (\mathbb{P}[\vect M(s,\lambda) = (\ell_0, s - 1 - \ell_0)] - \mathbb{P}[\vect M(s,\lambda) = (\ell_0 + 1,  s - \ell_0)])f(\ell_0 + 1, s - 1 - \ell_0)\\
\begin{split}
& = 1 + \mathbb{P}[\vect M(s,\lambda) = (x, s - 1 - x)] f(x + 1, s - 1 - x)  \\& \quad \quad + \mathbb{P}[\vect M(s,\lambda) = (x + 1,  s - x)](1 - f(x + 1, s - 1 - x))\,.
\end{split} \label{eq::subRemaining}
\end{align}

By \Cref{prop::splittingRatio}, the optimal $f$ must have a threshold solution. We can express this threshold solution by letting $x$ be be smallest value of $\ell_0$ where $f(\ell_0, s - \ell_0)$ is non-zero. Applying the cancellations to \eqref{eq:subftilde}, we get that only two terms remains. The value of $f$ which obtains the maximum must be where only the maximum value of $\mathbb{P}[\vect M(s,\lambda) = (x, s - 1 - x)]$ over all $x$ appears in \eqref{eq::subRemaining}, and this gives \eqref{eq::maxProbBound}.

We will bound 
\begin{align}
Z^*(c) \leq &\max_{{P_S: \EE[S] = {c}}}  \min_{0 \leq \lambda \leq 1}\mathbb{E}[ \phi(S,\lambda)] \\
\leq & \max_{{P_S: \EE[S] = {c}}} \min_{\lambda \in L_n}\mathbb{E}[ \phi(S,\lambda)] \\\eqdef& Z_n'(c)
\end{align}
where we defined\footnote{$L_n$ is
defined so that $L_n = \{\lambda \in (0,\frac{1}{2}]:\phi(s,\lambda) \leq \phi(s,\lambda')  \text{ for some } 1 < s \leq
2n \text{ and }\forall \lambda' \in [0,1] \}$ which is the set of all $\lambda$ which minimizes $\phi(s, \lambda)$ for
some $1< s \leq 2n$.} 
$$ L_n = \left\{\frac{\floor{s/2}}{s}: \text{ where } 1 < s \leq 2n \right\}\,.$$
Note that increasing $n$ makes the approximation tighter.
%The specific $L$ we choose is conjectured to make the inequality tight. (We conjecture this because the plot for $\phi_s(\lambda)$ looks to be concave at most points, and not differentiable at points $\frac{x}{s}$ where $x$ is an integer. The function $\sum_s P_s \phi_s(\lambda)$ will also be concave at most points. Thus, the minimum is probably at a non-differentiable point.)
Index the elements of $L_n$ as $\lambda_i$ where $\lambda_1 = \frac{1}{2}, \lambda_2 = \frac{1}{3}, \lambda_3 = \frac{2}{5}, ..., \lambda_n = \frac{n-1}{2n -1}$, so that
%\begin{equation}
$  \min_{\lambda \in L_n}\mathbb{E}[ \phi(S,\lambda)] =  \min_{i} \mathbb{E}[ \phi(S,\lambda_i)] \,.$
%\end{equation}

The quantity $Z_n'(c)$ is equivalent to maximizing the value of $t$ under the constraints that $\mathbb{E}[\phi(S,\lambda_i)] \geq t$ for all $1 \leq i \leq n$ and $\mathbb{E}[S] = {{c}}$.
%\begin{subequations}
%\begin{align}
%\textbf{maximize  } & t \label{subeq::maxProbOpt1} \\ 
%\textbf{subject to  } &  \sum_s^\infty P_S(s) \phi(s,\lambda_i) \geq t,  1 \leq i \leq n \label{subeq::maxProbOpt2}\\
%& P_S (s)\geq 0, s \in \mathbb{Z}_+\label{subeq::maxProbOpt3}\\
%&\sum_{s = 1}^{\infty} P_S(s) = 1\label{subeq::maxProbOpt4}\\
%& \sum_{s = 1}^\infty P_S(s) s = {{c}}\label{subeq::maxProbOpt5}
%\end{align}
%\end{subequations}
We can substitute 
\begin{align}
\phi(s, \lambda_i) =& \frac{s}{2}\left(1 + \max_{1 \leq \ell_0 \leq s} \mathbb{P}[\vect M(s,\lambda_i) = (\ell_0, s - 1 - \ell_0)]\right) \\
\eqdef& \frac{s}{2} \left(1 + \psi(s, \lambda_i)\right) \,.
\end{align}
%and taking the expectation over $\hat{P}_S = \frac{P_S(s) s}{\sum_s P_S(s) s}$ instead of $P_S(s)$. Then $Z_n'(c)$ is equivalent to maximizing $\frac{c}{2}(t + 1)$ under the constraints that $\mathbb{E}[\psi(S,\lambda_i)] \geq \frac{t}{c}$ for all $1 \leq i \leq n$ and $\mathbb{E}[\frac{1}{S}] = \frac{1}{c}$.
Then  $ \mathbb{E}[ \phi(S,\lambda_i)] = \frac{1}{2}\mathbb{E}[S] + \frac{1}{2}\mathbb{E}[S \cdot \psi(S, \lambda_i)] = \frac{c}{2} + \frac{1}{2}\mathbb{E}[S \cdot \psi(S, \lambda_i)]$. Note that  $\psi(S,\lambda_i) \to 0$ as $s \to \infty$ for all $i$.

For any value of $\pi_i \geq 0$, where $1 \leq i \leq n$, $\eta \geq 0$ and $\mu \geq 0$, we can define 

\begin{align}
\begin{split}
&Z_n''(c, \pi_1, ..., \pi_n, \eta, \mu) ] \eqdef \\ &\quad  \max_{P_s(s) \geq 0, \forall s} t + \sum_i \pi_i \left(\frac{c}{2} + \sum_{s = 1}^\infty {P}_S (s) \frac{s}{2} \psi(s, \lambda_i) - t \right) \\ & \quad - \eta\left(\sum_{s = 1}^\infty {P}_S (s) s - {c}\right) - \mu\left(\sum_{s = 1}^\infty P_S (s) - 1 \right)\,. \label{subeq::maxProbConvert1} 
\end{split}
\end{align}

%Then $Z_n'(c) \leq Z_n''(c, \pi_1, ..., \pi_n, \eta, \mu)$ for any $\pi_i, \eta, \mu$.

%\begin{subequations}
%\begin{align}
%\textbf{maximize  }  \begin{split} t + \sum_i^n \pi_i \left(\sum_{s}^{\infty}\hat{P}_S (s) \psi(s, \lambda_i) - t\right) \\- \eta\left(\sum_{s}^{\infty} \hat{P}_S (s) - {c}\right) - \mu\left(\sum_s^{\infty} \hat{P}_S (s) \frac{1}{s} - 1\right) \end{split}\label{subeq::maxProbConvert1} \\  
%\textbf{subject to  } & \hat{P}_S (s)\geq 0, s \in \mathbb{Z}_+ \label{subeq::maxProbConvert2}\\
%& \sum_{s = 1}^\infty \hat{P}_S(s) = {{c}} \label{subeq::maxProbConvert3}
%\end{align}
%\end{subequations}

Consider the set of $\pi_i, \eta, \mu$ which is the solution \textcolor{black}{to the dual problem}

\begin{subequations}
\begin{align}
\textbf{minimize  } & \frac{c}{2} + \eta c + {\mu}\label{subeq::dualOpt} \\ 
\textbf{subject to  }  &\sum_i^n \frac{1}{2}\pi_i \psi(s, \lambda_i) - \eta - \mu \frac{1}{s} \leq 0, s \in \mathbb{Z}_+ \label{subeq::dualConstraint}\\
&\sum_{i = 1}^{n} \pi_i - 1 = 0 \label{subeq::dualConstraint2}\\
& \eta \geq 0, \mu \geq0,  \pi_i \geq 0,  1 \leq i \leq n\,. \label{subeq::dualConstrain3}
\end{align}
\end{subequations}

Such an optimization has a solution which is easy to find despite having infinitely many constraints. The constraints \eqref{subeq::dualConstraint} will hold for all $s$ greater than some $s_0$ because $\psi(s,\lambda_i) \to 0$. By choosing a large enough $s_0$, we can solve the optimization by replacing it with an optimization where only the first $s_0$ constraints in \eqref{subeq::dualConstraint} are present.\footnote{\textcolor{black}{We can show that for each $c$, only considering constraints \eqref{subeq::dualConstraint} for $s \leq 16c$ is more than sufficient. All other infinite constraints can be removed without affecting the optimal solution.}}
%The minimum value of \eqref{subeq::dualOpt} is an upper bound to the maximization in \eqref{subeq::maxProbConvert1}.
Set the values of  $\pi_i, \eta, \mu$ in \eqref{subeq::maxProbConvert1} to be the values which obtain the minimum for \eqref{subeq::dualOpt}-\eqref{subeq::dualConstrain3}. Select a value of $s_1 \in \mathbb{Z}_+$. Then

\begin{align*}
Z_n'(c) &\leq Z_n''(c, \pi_1, ..., \pi_n, \eta, \mu) \numberthis\\
& =  \max_{P_s(s) \geq 0, \forall s} t \\ & \quad\quad + \sum_i \pi_i \left(\frac{c}{2} + \sum_{s = 1}^\infty {P}_S (s) \frac{s}{2} \psi(s, \lambda_i) - t \right) \\ & \quad\quad  - \eta\left(\sum_{s = 1}^\infty {P}_S (s) s - {c}\right) - \mu\left(\sum_{s = 1}^\infty P_S (s) - 1 \right) \numberthis\\
& =  \max_{P_s(s) \geq 0, \forall s} \sum_i \pi_i \left( \sum_{s = s_1}^\infty {P}_S (s) \frac{s}{2} \psi(s, \lambda_i) \right) \\& \quad\quad  - \eta\left(\sum_{s = s_1}^\infty {P}_S (s) s\right) - \mu\left(\sum_{s = s_1}^\infty P_S (s)  \right) \\ & \quad\quad + \frac{c}{2} + \eta c + \mu \numberthis\\ 
%\begin{split}
%=& t + \sum_i^n \pi_i \left(\sum_{s}^{s_1}\hat{P}_S (s) \psi(s, \lambda_i) + \sum_{s_1 + 1}^{\infty}\hat{P}_S (s) \psi(s, \lambda_i) - t\right)  \\& - \eta\left(\sum_{s}^{s_1} \hat{P}_S (s) + \sum_{s_1 + 1}^{\infty} \hat{P}_S (s)- {c}\right) - \mu\left(\sum_s^{s_1} \hat{P}_S (s) \frac{1}{s} + \sum_{s_1+1}^{\infty} \hat{P}_S (s) \frac{1}{s} - 1\right)
%\end{split}\\
%= & \sum_i^n \pi_i \sum_{s_1 + 1}^{\infty}\hat{P}_S (s) \psi(s, \lambda_i)  - \eta\left( \sum_{s_1 + 1}^{\infty} \hat{P}_S (s)- {c}\right) - \mu\left( \sum_{s_1+1}^{\infty} \hat{P}_S (s) \frac{1}{s} - 1\right)\\
%\leq & \sum_i^n \pi_i \sum_{s_1 + 1}^{\infty}\hat{P}_S (s) \psi(s, \lambda_i) + \eta{c} + \mu \\
&\leq \max_{s > s_1} \max_{i} \frac{c}{2}\psi(s, \lambda_i)  + \frac{c}{2} + \eta c + \mu \,.\numberthis
\end{align*}

Since $\psi(s, \lambda_i) \to 0$, the optimal $ \frac{c}{2} + \eta c + \mu $ given by \eqref{subeq::dualOpt}-\eqref{subeq::dualConstrain3} is an upper bound to $Z_n'(c)$ and hence also to $Z^*(c)$. This computes a lower bound on $\matr_\infty$. In \Cref{fig::maxProbConverse}, we found the lower bound using $n = 10$.

\subsubsection{Upper bounds}

To show a point in $\matr_{\infty}$ is achievable, it is sufficient to find a set of masses $P_S$ that achieves that point. Searching all possible masses $P_S$ is not computationally efficient. It turns out we can get decently close to the lower bound approximation by using the same masses which are solutions to $Z_n'(c)$ for each ${{c}}$ when restricting $P_S$ to only have finite support. While these results are close to the almost tight converse bound, they are not necessarily the best known. A few best known achievable points were found by simple search. The results are plotted in \Cref{fig::maxProbConverse} and shown in \Cref{fig:achNumbers}.

\subsection{Proof of \Cref{prop::hammingOptimal}} \label{sec::HammingOptimalProof}

\begin{proof}

The Hamming block achieves the point $(\frac{3}{2}, \frac{2}{3})$ in $\matr_{\infty}^3$. The proof that this is a corner point amounts to computing the region $\matr_{\infty}^3$. 

To solve for $\matr_{\infty}^3$, we will first simplify the expression for $F_3(P_S)$.
For any $P_S$ on $s \in [3]$, the labeling of primary nodes which gives the minimum value of $F_3(P_S)$ is when $P_X(0) = \frac{2}{3}$ and $P_X(1) = \frac{1}{3}$ (or these flipped). With this insight, we can simplify $F_3(P_S)$ to solve for the optimal $P_S$ given any parameter $\mathbb{E}[S] = c$ for some $1 \leq c \leq 3 $. Let $P_{Y|\vect L^{(3)}}(j|\ell_0, \ell_1)$ denote the proportion of redundant nodes of type $\vect \ell = (\ell_0, \ell_1)$ to label $j$.
\begin{align*}
F_3(P_S) &= \max_{P_{Y|\vect L^{(3)}}} \min_{j \in \{0, 1\}} {1 \over P_X(j)} \EE[L_j \mathbbm{1}\{Y=j\}] \numberthis\\
&= \max_{P_{Y|\vect L^{(3)}}} \min_{j \in \{0, 1\}}  \bigg\{ \frac{1}{P_X(j)} \sum_{s = 1}^3 P_S(s) \\ & \quad \quad \sum_{\vect \ell}  \ell_j P_{\vect L^{(3)}|S,P_X}(\vect \ell |s,P_X) P_{Y|\vect L^{(3)}}(j|\vect \ell)  \bigg\} \numberthis\\
&= \max_{P_{Y|\vect L^{(3)}}} \min \left\{ \frac{3}{2} \left(P_S(1) \frac{2}{3} P_{Y|\vect L^{(3)}}(0|1, 0) \right.\right.\\  
& \quad \quad +P_S(2) \left[ \frac{2}{3} P_{Y|\vect L^{(3)}}(0|1, 1) + 2\frac{1}{3} P_{Y|\vect L^{(3)}}(0|2, 0) \right]  \\ 
& \quad \quad  + P_S(3)2 P_{Y|\vect L^{(3)}}(0|2, 1) \bigg),  \\
& \quad \quad  \frac{3}{1}\left(P_S(1) \frac{1}{3} P_{Y|\vect L^{(3)}}(1|0, 1) \right. \\
& \quad \quad \quad+ P_S(2)  \frac{2}{3} P_{Y|\vect L^{(3)}}(1|1, 1)\\ 
& \left. \quad \quad \quad + P_S(3) P_{Y|\vect L^{(3)}}(1|2, 1) \bigg)\right\}\numberthis \label{eq::F3expand1}
\end{align*}

\textcolor{black}{To get \eqref{eq::F3expand1}, we expanded the summation into each term replacing $\ell_j$ and $P_{\vect L^{(3)}|S,P_X}(\vect \ell |s,P_X)$ with their numerical values. Since $P_X(0) = \frac{2}{3}$ and $P_X(1) = \frac{1}{3}$ and $k = 3$, we only need the $P_{Y|\vect L^{(3)}}$ terms for which $\vect \ell$ is  
a subset on $2$ zeros and $1$ one. }

We will first solve for the portion of $\matr_{\infty}^3$ where ${\varepsilon} > \frac{3}{2}$.

Set $\mathbb{E}[S] = 3$. There is a unique point of the form $(\frac{3}{\eta}, \frac{1}{\eta})$ for some $\eta > 0 $ which is a boundary point of the convex region $\matr_{\infty}^3$. The only distribution $P_S$ which can achieve $\mathbb{E}[S] = 3$ is when $P_S(3) = 1$ and $P_S(s) = 0$ for all other $s \neq 3$. With this $P_S$, we get that $F_3(P_S) = \frac{3}{2}$. Since no other $P_S$ is possible, the point 
\begin{equation}
\left(\frac{\mathbb{E}[S]}{F_3(P_S)}, \frac{1}{F_3(P_S)}\right) = \left(2, \frac{2}{3}\right)
\end{equation}
must be the boundary point of the form $(\frac{3}{\eta}, \frac{1}{\eta})$ in $\matr_{\infty}^3$. The line of points between this value and the value given by the Hamming block is achievable by convexity and by Claim~\ref{prop:basic:2} of \Cref{prop:basic} they must be optimal.

For the remaining portion of the region, we want to fixed a $1< c < \frac{9}{4}$ (the Hamming block has $\mathbb{E}[S] = \frac{9}{4}$), and solve for $P_S^* = \argmax_{P_S: \mathbb{E}[S] = c} F_3(P_S)$ and determine $F_3(P_S^*)$.

Note that it is optimal to set $P_{Y|\vect L^{(3)}}(0|\ell_0, 0) = 1$ and $P_{Y|\vect L^{(3)}}(1|0, \ell_1) = 1$. Then we can simplify notation by letting $P_{Y|\vect L^{(3)}}(0|1, 1) = x_{1,1}$ and $P_{Y|\vect L^{(3)}}(0|2, 1) = x_{2,1}$.
We can simplify \eqref{eq::F3expand1} by applying the constraints that $\sum_{s = 1}^3 P_S(s) = 1$ and $\sum_{s = 1}^3 P(s)s = c$.
%\begin{multline}
%F_3(P_S) = \max_{x_{1,1}, x_{2,1} \in [0,1]} \min \{3P_S(3) x_{2,1} + (c - 1 - 2_SP(3)) x_{1,1} + 1 - P_S(3), \\ c - 3P_S(3)x_{2,1} - 2(c - 1 - 2P_S(3))x_{1,1}\}
%\end{multline}
At the maximum point, the two quantities after the minimum must be equal.
%\begin{equation}
%P(3) = \frac{c - 1 - (3c-3)x_{1,1}}{6x_{2,1} - 6x_{1,1} - 1}
%\end{equation}
Simplifying the equation with these constraints, we have
\begin{align*} \label{eq::f3max}
&F_3(P_S^*) = 
\max_{x_{1,1}, x_{2,1} \in [0,1]} \\
& \quad\frac{c - 1 - (3c-3)x_{1,1}}{6x_{2,1} - 6x_{1,1} - 1} (3x_{2,1} - 2x_{1,1} - 1)  + (c-1)x_{1,1} + 1 \numberthis                                                                                                                                                                                                                                                                                                                                                                                                                                                                                                                                                                                                                                                                                                                                                                                                                                                                                                                                                                                                                                                                                                                                                                                                                                                                                                                                                                                                                                                                                                                                                                                                                                                                                                                                                                                                                                                                                                                                                                                                                                                                                                                                                                                                                                                                                                                                                                                                                                                                                                                                                                                                                                                                                                                                                                                                                                                                                                                                                                                                                                                                                                                                                                                                                                                                                                                                                                                                                                                                                                                                                                                                                                                                                                                                                                                                                                                                                                                                                                                                                                                                                                                                                                                                                                                                                                                                                                                                                                                                                                                                                                                                                                                                                                                                                                                                                                                                                                                                                                                                                                                                                                                                                                                                                                                                                                                                                                                                                                                                                                                                                                                                                                                                                                                                                          
\end{align*}
under the constraints that the variables are in $[0,1]$.

The optimal labeling must have that either $x_{1,1} = 0$ and $x_{2,1} \in [0,1]$ or that $x_{1,1} \in [0,1]$ and $x_{2,1} = 1$ by \Cref{prop::splittingRatio}. We try the cases $x_{1,1} = 0$ and $x_{2,1} = 1$ and take derivatives to solve for the best value of $x_{2,1}$ or $x_{1,1}$. For any value of $c$ we pick\footnote{For example, we can pick $c = \frac{12}{7}$. We get that the maximum value occurs when $x_{1,1} = 0$ and $x_{2,1} = 1$. With this setting of variables, $P_S^*(1) = \frac{3}{7}, P_S^*(2) = \frac{3}{7}, P_S^*(3) = \frac{1}{7}$ which corresponds to the subset design $S(3,3) \vee S(3,2) \vee S(3,1)$ and achieves $\left(\frac{12/7}{F_3(P_S^*)}, \frac{1}{F_3(P_S^*)}\right) = (\frac{12}{9}, \frac{7}{9})$.}, the point
$(\frac{c}{F_3(P_S^*)}, \frac{1}{F_3(P_S^*)})$ lies on the line between the point achievable by the repetition design and the Hamming block. By convexity, it must be that all points on the line between the values achievable by the repetition design and the Hamming block are optimal.

\end{proof}

\bibliographystyle{IEEEtran}
\bibliography{IEEEabrv,reports}
\end{document}